\documentclass[aps,prx,10pt,longbibliography,superscriptaddress,twocolumn,showpacs,oneside]{revtex4-2}

\usepackage{xcolor}
\usepackage{wrapfig}
\usepackage{graphicx}
\usepackage[caption=false]{subfig}
\usepackage{float}
\usepackage{amsmath,amssymb,amsthm,mathrsfs,mathtools,amsfonts,physics,bm,bbm,booktabs}
\usepackage{mathdots}
\usepackage{dsfont}
\usepackage[version=4]{mhchem}

\usepackage[colorlinks=true, linktocpage=true]{hyperref}
\definecolor{linkblue}{HTML}{0057B8} 
\hypersetup{allcolors = linkblue}

\usepackage{tikz}
\usetikzlibrary{shapes.geometric, arrows}

\usepackage{algorithm}
\usepackage{algpseudocode}
\usepackage{setspace}

\tikzstyle{Atensor} = [rectangle, rounded corners, minimum width=1cm, minimum height=1cm, text centered, draw=black]
\tikzstyle{Gtensor} = [circle, minimum width=1cm, minimum height=1cm, text centered, draw=black]
\tikzstyle{arrow} = [thick,->]

\newtheorem{theorem}{Theorem}

\newtheorem{lemma}{Lemma}[section]

\newtheorem{conjecture}{Conjecture}
\usepackage{mleftright}
\newtheorem{itheorem}{Theorem}
\newtheorem{iconjecture}{Conjecture} 
\newtheorem{atheorem}{Theorem}

\usepackage{soul}

\DeclareRobustCommand{\circlednum}[1]{%
  \raisebox{0.15ex}{\textcircled{\raisebox{-0.25ex}
  {#1}}}
}

\usepackage{bbm}
\usepackage{braket}
\definecolor{THc}{rgb}{0.9,0.3,0.2}

\usepackage{enumitem}

\begin{document}
\title{Emergent universality in Kraus maps of quantum chaotic many-body dynamics}

\author{Qi Camm Huang}
\email{qi.huang@u.nus.edu}
\affiliation{Department of Physics, National University of Singapore, Singapore 117551}
\affiliation{Centre for Quantum Technologies, National University of Singapore, Singapore 117543}

\author{Wai-Keong Mok}
\affiliation{Institute for Quantum Information and Matter,
California Institute of Technology, Pasadena, CA 91125, USA}

\author{Tobias Haug}
\affiliation{Quantum Research Center, Technology Innovation Institute, Abu Dhabi, UAE}

\author{Wen Wei Ho}
\email{wenweiho@nus.edu.sg}
\affiliation{Department of Physics, National University of Singapore, Singapore 117551}
\affiliation{Centre for Quantum Technologies, National University of Singapore, Singapore 117543}

\date{\today}

\begin{abstract}
    Recent studies of ``deep thermalization'' have revealed universal physics in quantum many-body dynamics beyond 
    equilibration towards Gibbs states:  maximally random quantum state ensembles can emerge on local subsystems, generated by measurements on their complement.
    In this work, we further identify a new form of 
    universality exhibited in the ``finer fingerprints'' of quantum dynamics for local subsystems, induced by global unitary time-evolution.
    Specifically, we consider the \textit{projected Kraus ensemble}, an ensemble of Kraus operators obtained by unraveling the quantum channel on a small subsystem with respect to knowledge of the classical configurations of its complement (the ``bath'').
    Our central result is a one-parameter random matrix Ansatz that captures the ensemble's emergent statistical behavior along particular scalings of space and time, valid for generic 1D circuit dynamics without conservation laws: the ensemble is described by the product of a complex Ginibre random matrix and an independent log-normal random real scalar.
    The former encodes information scrambling within the subsystem, captured by rotational invariance of the Ginibre measure, whereas the latter encodes fluctuations in the probabilities that each bath configuration is seen, arising from locality of the underlying dynamics. 
    Our Ansatz can be established in the special cases of dynamics generated by global Haar random unitaries and dual-unitary circuits, while we motivate it for generic circuits using arguments of spacetime duality and the multiplicative ergodic theorem on long products of spatial transfer matrices.
    Extensive numerical simulations for random and Floquet circuit models further verify these predictions. 
    This universality in Kraus operators also provides a  microscopic, state-agnostic mechanism for deep thermalization in generic 1D quantum circuit dynamics, 
    and has implications for local quantum information recoverability involving classical side information tied to knowledge of the bath.
\end{abstract}
\maketitle

\begingroup
\small
\linespread{0.972}\selectfont
\tableofcontents
\endgroup

\section{Introduction}
\label{sec:intro}

    \begin{figure*}[!t]
    \centering
    \includegraphics[width=0.95\linewidth]{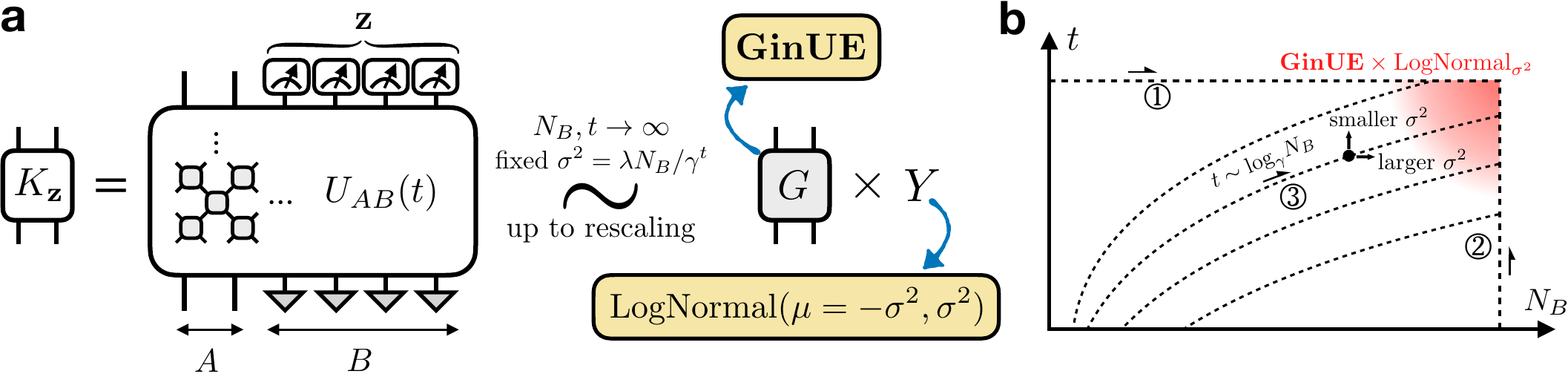}
    \caption{(a)~Emergent universality of Kraus maps. In this work, we focus on 1D circuit dynamics, fixing initial state on $B$ to be a simple product state, and the final measurement basis on $B$ to be the computational basis. A Kraus operator $K_\mathbf{z}$ is obtained by partially projecting the global unitary $U_{AB}(t)$ as in Eq.~\eqref{eq:q-channel-kraus-decomposition}. Our central result (Conjecture~\ref{conj:1-informal}) is that, up to a rescaling factor $\sqrt{d}=\sqrt{d_Ad_B}$ where $d_A$ ($d_B$) is the Hilbert space dimension of $A$ ($B$), the Kraus operators can be statistically described by the product of a standard complex random Ginibre matrix ($\mathbf{GinUE}$) and an independent log-normal real random scalar ($\mathrm{LogNormal}$) in certain spacetime joint scaling limits.
    The log-normal part is parametrized by $\sigma^2 = {\lambda N_B}/{\gamma^t}$, where $\lambda\ge0$ and $\gamma > 1$ are circuit-dependent non-universal constants.
    (b)~Spacetime diagram of bath size~($N_B$) and evolution time~($t$) depicting various scaling limits. We highlight three scaling limits, each corresponding to a theorem/conjecture stated in the main results. \textbf{\circlednum{1}}~{\!\!Long-time} limit ($t\rightarrow\infty$) first, corresponding to Theorem~\ref{thm:1-informal}, where we study the projected Kraus ensemble for global Haar random unitaries, effectively taking the long-time limit first before scaling up the bath size $N_B$, thus $\sigma^2=0$. \textbf{\circlednum{2}}~{\!\!Thermodynamic} limit ($N_B\rightarrow\infty$) first, corresponding to Theorem~\ref{thm:2-informal}, where we consider exactly solvable dual-unitary (DU) circuits and take the thermodynamic limit (TDL) first before taking the long-time limit. In this special case, $\lambda=0$ and hence $\sigma^2=0$. \textbf{\circlednum{3}}~{\!\!Joint} limit $N_B, t\rightarrow\infty$ with $t=\log_\gamma(N_B) + c$ for arbitrary constant $c$, corresponding to Conjecture~\ref{conj:1-informal}, where we consider a generic 1D circuit following joint scaling limit $N_B, t\rightarrow\infty$ with $\sigma^2 = \lambda N_B/\gamma^t$ held fixed.
    We conjecture and provide evidence that the projected Kraus ensemble follows the Ansatz Eq.~\eqref{eq:1}, a family of universal distributions parametrized by $\sigma^2$. {The top-right region schematically indicates the parametric regime where the universal Ansatz holds.}
    }
    \label{fig:concept}
    \end{figure*}

Thermalization is a ubiquitous fate of quantum chaotic many-body systems: quantum information initially localized on a small subsystem is typically scrambled into nonlocal degrees of freedom, such that the subsystem relaxes to an equilibrium Gibbs state determined only by globally conserved quantities such as energy or charge. This universal behavior underpins quantum statistical mechanics, and provides a mechanism for how irreversible equilibration arises despite globally reversible evolution in isolated quantum systems~\cite{deutsch_quantum_1991, srednicki_chaos_1994, Rigol_thermalization_2008, Dalessio_from_2016}.

Recently, there has been a surge of interest in studying novel physics that arises in the equilibration dynamics of quantum chaotic many-body systems beyond ``regular'' thermalization~\cite{Ares_mpemba_2025, Turkeshi_mpemba_2025, Mueller_mpemba_2026, pappalardi_eigenstate_2022, Pappalardi_full_2025, Fava_designs_2025, jonay_two-stage_2024, Jonay_two-stage_2025}. 
\textit{Deep thermalization} constitutes a prominent example of such physics: 
it describes the universal emergence of maximally random quantum state ensembles on a local subsystem upon measuring the complementary subsystem, which can be viewed as the ``bath''~\cite{cotler_emergent_2023, mark_maximum_2024, lucas_fermions_2023, liu_cv_2024, bejan_matchgate_2025}. Specifically, each state in the ensemble is a post-measurement state on the local subsystem $A$ upon projectively measuring the bath $B$ in some simple local basis, and deep thermalization posits that these states are distributed 
in a maximally entropic fashion shaped only by global physical constraints~\cite{mark_maximum_2024, yu_mixed_2025, Sherry_deep_2026, lucas_fermions_2023, liu_cv_2024, bejan_matchgate_2025, chang_deep_2025, liu_coherence-induced_2025, mcginley_scrooge_2025, mok_nature_2026}.
Deep thermalization represents a novel concept of equilibration going beyond regular thermalization of local observables as it captures the local equilibration {\it conditioned} on certain information of the bath~\footnote{Indeed, the ensemble of post-measurement states can be understood as a physically motivated unraveling of the local reduced density matrix where each state is tied to knowledge of a particular classical configuration of the bath; the reduced density matrix is simply the average state of the ensemble where such knowledge is erased.},
and the emergent universal randomness has found myriad applications in quantum information science, including
ancilla-assisted shadow tomography~\cite{tran_measuring_2023, McGinley_shadow_2023}, benchmarking of quantum simulators using chaos~\cite{mark_benchmarking_2023, shaw_benchmarking_2024}, as well as generation, characterization, and certification of quantum resources~\cite{mok_optimal_2025, Varikuti_resources_2025, Du_certifying_2026, Feng_resource_2026}.

In this work, we identify a new form of universality exhibited in the \textit{dynamical maps} of quantum chaotic many-body systems themselves, focusing on one-dimensional (1D) local quantum circuits with no conserved quantities. Specifically, we consider the quantum channel $\mathcal{N}$ acting on the local subsystem $A$ generating its dynamics, which is induced by a global unitary time-evolution $U_{AB}$ on $A\cup B$. We study the ensemble $\mathcal{K}$ of Kraus operators $K_\mathbf{z}$ underlying $\mathcal{N}$, which we refer to as the \textit{projected Kraus ensemble}~\footnote{By analogy with the notion of \textit{projected ensemble} of quantum states, here we collect the Kraus operators (or, by a slight abuse of terminology, dynamical maps/Kraus maps, which we use interchangeably in this work).},
where each Kraus map is associated with knowledge of a particular classical configuration $\mathbf{z}$ of the bath $B$, see Fig.~\ref{fig:concept}(a).
This is a setting closely related to and motivated by deep thermalization, which concerns universality within local regions of quantum many-body {\it states} conditioned on knowledge of the bath; here we are instead interested in whether the {\it dynamical processes} on local regions themselves may also exhibit universal structure, upon conditioning too on the additional knowledge of the bath~\footnote{We note that there have been many studies on the properties of the dynamical map on a local subsystem induced by global unitary evolution, termed the ``influence matrix'' or ``process tensor''; see~\cite{ODonovan_diagnosing_2026} for a recent work on projected ensembles of process tensors as fine-grained probes of quantum chaos.
While the conceptual parallel is interesting, we do not explore this connection in this paper.}.
Physically, the projected Kraus ensemble $\mathcal{K}$ dictates how local quantum information is scrambled and transmitted over time, given certain classical side information~\footnote{Here, the classical side information refers to the bath measurement outcome that is made available to a future observer, allowing them to act on the local subsystem conditioned on that knowledge.} of its many-body environment. 

Our central result is a conjecture that the (appropriately rescaled) projected Kraus ensemble $\mathcal{K}$ of a local subsystem $A$, under certain joint scaling limits $N_B, t\rightarrow\infty$, is statistically described by a universal one-parameter Ansatz (Conjecture~\ref{conj:1-informal}, informal, Sec.~\ref{sec:overview-main-results}):
\begin{equation}
\label{eq:1} \mathbf{GinUE}\times\mathrm{LogNormal}_{\sigma^2}.
\end{equation}
Here, the Ansatz is composed of a random matrix multiplied by an independent random scalar whose distribution is set by a single parameter $\sigma^2$: $\mathbf{GinUE}$ is the complex Ginibre random matrix of size $d_A \times d_A$ ($d_A$: Hilbert space dimension on $A$), 
and 
$\text{LogNormal}_{\sigma^2}$ denotes the log-normal distribution of a real random scalar $Y$, where $\log(Y)\sim\mathcal{N}(\mu=-\sigma^2,\sigma^2)$ is normally distributed with variance $\sigma^2$ and mean $\mu = -\sigma^2$.
For the 1D local circuits considered here, we predict a family of joint scaling limits $N_B, t\rightarrow\infty$ along contours $\lambda N_B / \gamma^t=\sigma^2$ for fixed $\sigma^2 > 0$, where $\lambda\ge0$ and $\gamma > 1$ are model-specific constants, under which the projected Kraus ensemble $\mathcal{K}$ converges to the corresponding stable limiting form Eq.~\eqref{eq:1}.
Equivalently, the family of ``critical'' joint scaling limits $N_B, t\rightarrow\infty$ (see Fig.~\ref{fig:concept}(b)) follows $t = \log_\gamma(N_B) + \log_\gamma(\lambda / \sigma^2)$. For a more complete statement and conditions under which our conjecture is expected to hold, see Conjecture~\ref{conj:1-precise} (precise) in Sec.~\ref{sec:ruc}.

One implication of our conjecture is that we may consider other scaling limits on either side of the family of critical scalings. For $t$ scaling faster than critical, $\sigma^2\rightarrow 0 $ and $\mathcal{K}$ will be described solely by the Ginibre ensemble; for slower scalings, $\sigma^2\rightarrow\infty$ and $\mathcal{K}$ will not converge to a nontrivial limiting stable distribution: there are substantial fluctuations in the probabilities of measurement outcomes spanning increasingly many orders of magnitude, with their relative spread becoming unbounded in the joint limit. Our result is consistent with the emergence of a crossover from log-normal to Porter--Thomas distributions in global bitstring probabilities as 1D chaotic circuits evolve from shallow to deep~\cite{christopoulos_universal_2025, sauliere_universality_2025, sauliere_universality_noisy_2025}, though distinct from a similar phenomenology found in weakly monitored dynamics~\cite{bulchandani_random-matrix_2024}.

We analytically establish our Ansatz Eq.~\eqref{eq:1} in two special cases, also illustrated in Fig.~\ref{fig:concept}(b): (i) when the global time-evolution is taken to be a global Haar random unitary, corresponding to the limit $t \to \infty$ then $N_B \to \infty$ in that order, so that $\sigma^2=0$, proved in Theorem~\ref{thm:1-informal}.  (ii) Within a special class of exactly solvable models called dual-unitary (DU) circuits~\cite{Bertini_entanglement_2019, Bertini_exact_2019, Piroli_exact_2020, Zhou_maximal_2022, Bertini_exactly_2026} taken in the opposite limit $N_B \to \infty$ then $t \to \infty$. 
Here, $\sigma^2=0$ too because of the absence of {\it dynamical purification}~\cite{Gullans_dynamical_2020, fidkowski_how_2021, ippoliti_dynamical_2023} within the circuit when viewed ``sideways'', i.e., upon performing a spacetime rotation~\cite{ho_exact_2022, Claeys_emergent_2022, ippoliti_dynamical_2023, Stephen_universal_2024}: 
because of dual-unitarity enforcing that sideways evolution is unitary, the model-specific constant 
$\lambda = 0$, hence $\sigma^2=0$ too.
We prove this in  Theorem~\ref{thm:2-informal} for the kicked Ising model (KIM) at the DU point, but we explain why the proof strategy extends to generic DU cases. For more general 1D quantum circuits (e.g., random or Floquet local circuits) which do generically exhibit dynamical purification in their sideways evolution, we argue for the form of the Ansatz based on the multiplicative ergodic theorem~\cite{Oseledets_1968} applied to long products of non-unitary random transfer matrices in the spacetime rotated version of the system, such that generally $\sigma^2 > 0$ (the randomness is provided for by the measurements on the bath, even if the circuit is non-random). We further numerically verify our Ansatz in various microscopic models, in particular finding that the convergence to the postulated universal form Eq.~\eqref{eq:1} occurs exponentially quickly in time. 

    \begin{figure}[t]
        \centering
        \includegraphics[width=0.95\linewidth]{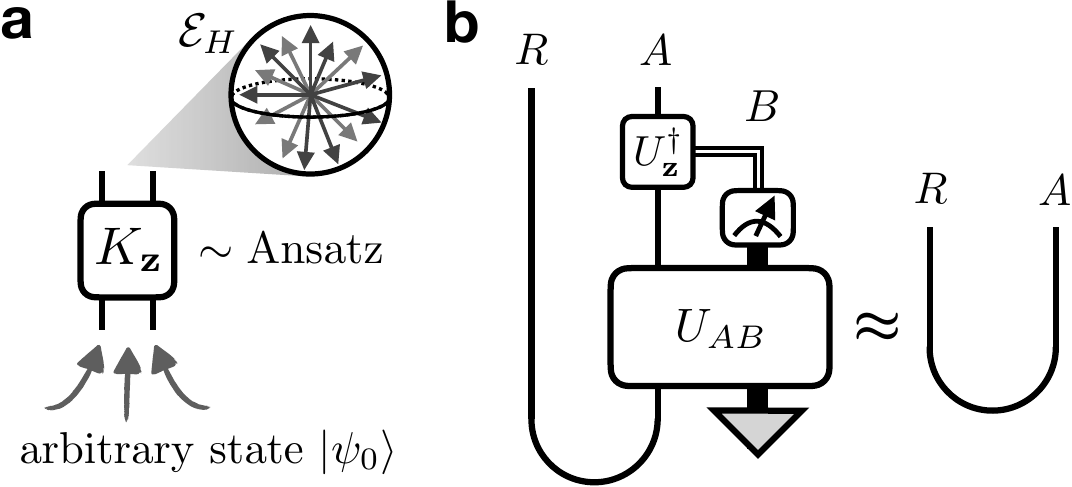}
        \caption{Two physical consequences of the Ansatz. (a)~Deep thermalization for an arbitrary input pure state $|\psi_0\rangle$ on $A$. With the Kraus operators $K_\mathbf{z}$ following the Ansatz Eq.~\eqref{eq:1}, the (unnormalized) projected states $|\widetilde\psi_\mathbf{z}\rangle = K_\mathbf{z}|\psi_0\rangle$
        have no preferred direction and hence the (normalized) projected states $|\psi_\mathbf{z}\rangle$ 
        can be shown to be distributed according to the Haar state ensemble $\mathcal{E}_H$, see Sec.~\ref{sec:consequence-ansatz-deep-thermalization}. (b)~Quantum information recovery assisted with classical side information. Here, $R$ is a reference system maximally entangled with the initial subsystem $A$, and the measurement outcome $\mathbf{z}$ of $B$ is the classical side information fed forward to a decoder. The decoder could then perform a conditional unitary action $U_\mathbf{z}^\dagger$
        to recover quantum information of $R$, shown schematically as the distillation of Bell pairs. See Sec.~\ref{sec:consequence-ansatz-qi-recovery} for details.
        }
        \label{fig:consequences}
    \end{figure}

Besides being of conceptual interest, our results unveil a new kind of universality in quantum chaotic many-body dynamics and have the following physical consequences. 
First, they provide a microscopic mechanism to understand the emergence of deep thermalization (to the Haar ensemble) in generic quantum chaotic many-body dynamics without conservation laws, {\it regardless} of the input state on $A$, see Fig.~\ref{fig:consequences}(a). 
Existing analytical results deriving the emergence of deep thermalization to the Haar ensemble have previously been limited to the special cases of Haar random~\cite{cotler_emergent_2023} and computationally pseudorandom~\cite{chakraborty_fast_2025} global initial states, non-local random Hamiltonian evolved states~\cite{Ghosh_design_2025}, or quantum dynamics involving DU circuits beginning from the so-called ``solvable initial states''~\cite{ho_exact_2022, Claeys_emergent_2022, ippoliti_dynamical_2023}; here we put forth a unifying framework underpinning the ubiquity of deep thermalization involving a general minimal random matrix Ansatz of the measured bath. 
Technically, this arises from the Ginibre component of the Ansatz, which encodes the behavior of quantum information scrambling within subsystem $A$: the rotational invariance of the Ginibre measure means that unnormalized post-measurement states $|\widetilde{\psi}_\mathbf{z}\rangle=K_\mathbf{z}|\psi_0\rangle$ on $A$ have no preferred direction, regardless of the input state $|\psi_0\rangle$ on $A$, leading to the normalized post-measurement states $|{\psi}_\mathbf{z}\rangle$ being distributed isotropically in the Hilbert space. Intriguingly, what Eq.~\eqref{eq:1} additionally unveils is that even when deep thermalization occurs (i.e., when distribution of post-measurement states is maximally random and featureless), there can still be additional information about the global system retained within  the distribution of probabilities of bath outcomes $\mathbf{z}$, 
captured by the  log-normal random scalar component with parameter $\sigma^2$.

Second, within this framework, we study quantum information recovery on local subsystems with a process-level resolution (namely, measurement-resolved branches of $K_\mathbf{z}$) implemented by the classical side information ($\mathbf{z}$). Specifically, we consider quantum information initially localized on subsystem $A$, which is scrambled into the large many-body bath $B$ through the dynamics. By measuring $B$, the classical outcome $\mathbf{z}$ identifies a conditional dynamical map, namely the Kraus operator $K_\mathbf{z}$, that acts on $A$ effectively as a subprocess. We then ask whether the initial quantum information can be retrieved given knowledge of $\mathbf{z}$, as illustrated in Fig.~\ref{fig:consequences}(b). This setting of measurement-conditioned dynamical maps can be seen as a \textit{quantum instrument} on the subsystem---$\mathbf{z}$ serves as the classical side information, and our question constitutes a variant of the seminal Hayden--Preskill thought experiment~\cite{hayden_mirrors_2007}. Using our Ansatz Eq.~\eqref{eq:1}, we demonstrate that such side information is a powerful resource for information recovery, enabling arbitrarily high recovery fidelity with only a small physical overhead. On the contrary, when such side information is absent, the local dynamics $\mathcal{N}$ is effectively irreversible and local quantum information is irretrievably lost into the bath $B$.

The rest of this paper is structured as follows:
in Sec.~\ref{sec:overview}, we introduce the object of interest---the projected Kraus ensemble (Sec.~\ref{sec:overview-setup})---before stating our main results informally (Sec.~\ref{sec:overview-main-results}).
The following sections expand on our main results: in Sec.~\ref{sec:global-haar}, we discuss our first analytical result (Theorem~\ref{thm:1-informal}) of the projected Kraus ensemble of global Haar random unitary dynamics, providing a random-matrix-theoretic baseline for generic physical dynamics. In Sec.~\ref{sec:DU-Floquet}, we discuss our second analytical result (Theorem~\ref{thm:2-informal}) for DU circuits, a special class of 1D circuits enjoying certain analytical properties. In Sec.~\ref{sec:ruc}, we elaborate on our conjecture for generic 1D circuits
(Conjecture~\ref{conj:1-informal}) and provide physical motivation for the family of one-parameter Ansatz (Sec.~\ref{sec:origin-ansatz}). The conjecture is supported by extensive numerics on a brickwork random unitary circuit (Sec.~\ref{sec:numerics}) and Floquet KIM (Appendix~\ref{app:extra-numerics}). In Sec.~\ref{sec:consequence-ansatz}, we discuss the physical consequences of our Ansatz, including a dynamical mechanism for deep thermalization (Sec.~\ref{sec:consequence-ansatz-deep-thermalization}, with a brief recap of deep thermalization) and recovery of quantum information in local conditional dynamics (Sec.~\ref{sec:consequence-ansatz-qi-recovery}). We conclude with our discussion and outlook in Sec.~\ref{sec:discussion}.

\section{Setup and overview of main results}
\label{sec:overview} 

    \subsection{Motivation and setup}
    \label{sec:overview-setup}
    In this work, we are interested in universal aspects of chaotic quantum many-body dynamics under a global unitary evolution $U_{AB}(t)$, generated, for example, by a (time-dependent) many-body Hamiltonian or a spatially extended quantum circuit. There are a variety of probes known to uncover universality of such dynamics, including spectral correlations diagnosed by the spectral form factor~\cite{Chan_spectral_2018, Bertini_random_2021, chan_spectral_2021} and its generalizations~\cite{li_dissipative_2021, shivam_ginibre_2023}, operator spreading characterized by the out-of-time-ordered correlators (OTOCs)~\cite{Nahum_operator_2018, vonKeyserlink_operator_2018, Khemani_lyapunov_2018}, and statistics of bitstring measurement probabilities~\cite{mark_maximum_2024, christopoulos_universal_2025, sauliere_universality_2025, sauliere_universality_noisy_2025, shaw_experimental_2025}. Here, we instead focus on the universal structure that emerges at the level of the \textit{local quantum process} itself. Specifically, we consider a bipartition of the global system $A\cup B$ into a small subregion $A$ and a large complementary region $B$ serving as its ``bath''. Dynamics on $A$ is generated by the \textit{local quantum channel}
    \begin{equation}
        \label{eq:q-channel}
            \mathcal{N}^{A\rightarrow A}_t(\cdot) = \mathrm{Tr}_B\left[
                U_{AB}(t) \left(\,\cdot\,\otimes |B_0\rangle \langle B_0|\right) U_{AB}^{\dagger}(t)
            \right],
    \end{equation}
    where we have assumed $A$ and $B$ are prepared in a factorized state, with the initial state $|B_0\rangle$ on $B$ some simple lowly-entangled state (e.g., a product state). Given an input state $\rho_{A,0}$ on $A$, the quantum channel $\mathcal{N}_t^{A\rightarrow A}$ evolves it to a new state:
    \begin{align}
        \rho_{A,t} = \mathcal{N}^{A\rightarrow A}_t(\rho_{A,0}).
    \end{align}
    For generic quantum chaotic systems, it is expected that the state at late times tends to a universal steady state, described by a thermal Gibbs state, with the effective temperature set by the large bath---this is known as quantum thermalization~\cite{deutsch_quantum_1991, srednicki_chaos_1994, Rigol_thermalization_2008, Dalessio_from_2016}. In particular, if there are no global conservation laws, then independently of the initial state $\rho_{A,0}$,  the resulting state is maximally mixed, leading to the expectation that the channel at late times is itself described by a universal form: the maximally depolarizing channel.

    Crucially, in the formulation above, the channel $\mathcal{N}_t^{A\rightarrow A}$ is obtained by tracing out the bath $B$, thereby discarding any knowledge of it. Inspired by recent works on deep thermalization~\cite{cotler_emergent_2023, mark_maximum_2024, lucas_fermions_2023, liu_cv_2024, bejan_matchgate_2025, yu_mixed_2025, Sherry_deep_2026, chang_deep_2025, liu_coherence-induced_2025, mcginley_scrooge_2025, mok_nature_2026}, which found that local states on subregion $A$ {\it conditioned} on some knowledge of the bath $B$ (e.g., obtained by measuring it in a simple local basis like the computational basis $\{ |\mathbf{z}\rangle_B\}$) are described universally by random wavefunction distributions, here we similarly ask whether the channel $\mathcal{N}^{A\rightarrow A}_t$ may also exhibit universality \textit{conditioned} on some knowledge of the bath.
    Specifically, we probe this ``fine-grained'' structure by writing the local channel on $A$ (dropping $t$ and superscript for brevity) in Kraus decomposition as
    \begin{equation}
    \label{eq:q-channel-kraus-decomposition}
    \begin{cases}
        \mathcal{N}(\cdot) = \sum_{\mathbf{z}}K_\mathbf{z}(\cdot)K_\mathbf{z}^\dagger \vspace{2pt}
        \\
        K_\mathbf{z} :=(\mathbbm{1}_A\otimes \langle\mathbf{z}|_B)U_{AB}(\mathbbm{1}_A\otimes |B_0\rangle_B),
    \end{cases}
    \end{equation}
     and investigate the statistical properties of the ensemble 
     \begin{align}
     \label{eqn:PKE}
        \mathcal{K} := \{ K_{\mathbf{z}}\}_{ \mathbf{z}=1}^{d_B},
     \end{align}
    of Kraus maps $K_\mathbf{z}$, illustrated by Fig.~\ref{fig:concept}(a), where $\{|\mathbf{z}\rangle_B\}$ denotes the computational basis of $B$ and $d_B$ ($d_A$) is the Hilbert space dimension of $B$ ($A$). We dub this the {\it projected Kraus ensemble},  in analogy to the \textit{projected ensemble} in deep thermalization (a high-level review of which can be found at Sec.~\ref{sec:consequence-ansatz-deep-thermalization}). 
    We will also refer to $\mathcal{K}$ as simply the \textit{Kraus ensemble} for brevity.
    Physically, each Kraus map $K_\mathbf{z}$ is a dynamical process acting on $A$ associated with a measurement outcome $\mathbf{z}$, i.e., a particular \textit{subprocess} obtained by unraveling the channel $\mathcal{N}$ using information of the classical configuration $|\mathbf{z}\rangle_B$. As $K_\mathbf{z}$ is a $d_A\times d_A$ submatrix of the global unitary $U_{AB}$,  it  is generally neither Hermitian nor proportional to a unitary.

    Our goal is to identify whether there are particular scalings of space and time where the Kraus ensemble $\mathcal{K}$ converges to universal forms. Below, we focus on 1D quantum circuits of qubits without explicit conservation laws (thus, if a subsystem thermalizes, it does so to infinite temperature) and fix subregion $A$ to be a constant number of qubits $N_A$. Concretely, a joint scaling of space and time is  specified by $f(N_B, t) = \mathrm{constant}$, relating the number of qubits $N_B$ in $B$ (the bath size) to the time $t$ (or circuit depth) as $N_B, t\rightarrow\infty$. Given a particular scaling form, we aim to identify a theoretical target random matrix ensemble $\mathcal{G}_\mathrm{th}:=\{\mathrm{d}\mu_\mathrm{th}(V), V\}$, characterized by an appropriate probability measure $\mu_\mathrm{th}$ on $d_A\times d_A$ matrices, such that the empirical ensemble $\mathcal{K}$ tends towards upon proper rescaling, namely
    \begin{equation}
    \label{eq:rescaled-convergence-schematic}
        \sqrt{d}\mathcal{K} \to \mathcal{G_\mathrm{th}},\quad N_B, t\rightarrow\infty~\text{with}~f(N_B,t)=\mathrm{const.},
    \end{equation}
    where $d=d_A d_B$ is the dimension of the full system $A\cup B$. The factor $\sqrt{d}$ is needed so that each rescaled matrix $\sqrt{d} K_\mathbf{z}$ has entries that are normalized to $\mathcal{O}(1)$
    even in the thermodynamic limit~\footnote{This rescaling is standard in random matrix theory for obtaining a large $d$ limit. For example, a fixed-sized submatrix of a Haar random ($d\times d$) unitary has entries of order $\mathcal{O}(1/\sqrt{d})$, and multiplying the submatrix by $\sqrt{d}$ yields a Ginibre matrix as $d\rightarrow\infty$ with entries defined as i.i.d. standard complex Gaussian variables. See discussion before Theorem~\ref{thm:1-precise} (precise).}.
    To quantify such convergence, we will compare moment-by-moment the rescaled empirical ensemble $\sqrt{d}\mathcal{K}$ with the theoretical ensemble $\mathcal{G}_\mathrm{th}$~\footnote{The theoretical ensemble $\mathcal{G}_\mathrm{th}$ will contain a log-normal component for this work. The log-normal distribution is a long-tailed distribution and is intrinsically moment-indeterminate, i.e., there exists another distribution which shares exactly the same finite-order moments as the log-normal. Nevertheless, we proceed with the physically-motivated measure of evaluating finite-moment distances, as did~\cite{christopoulos_universal_2025, sauliere_universality_2025, sauliere_universality_noisy_2025}.}.
    Defining the $k$-th moment of the former
    \begin{equation}
        d^k M_\mathrm{emp}^{(k)} = \frac{d^k}{d_B}\sum_\mathbf{z} K_\mathbf{z}^{\otimes k}\otimes K_\mathbf{z}^{*\,\otimes k},
    \end{equation}
            and the $k$-th moment of the latter as
    \begin{align}
    M_\mathrm{th}^{(k)}=\int\mathrm{d}\mu_\mathrm{th}(V)\,V^{\otimes k}\otimes V^{*\otimes k},
    \end{align}
    then their $k$-th moment distance is defined to be
    \begin{equation}
    \label{eq:distance-th}
        \Delta_\mathrm{th}^{(k)}:= \frac{\left\| d^k M_\mathrm{emp}^{(k)} - M_\mathrm{th}^{(k)} \right\|_2}{\left\|M_\mathrm{th}^{(k)}\right\|_2},
    \end{equation}
    where the denominator provides normalization to the distance compared with the magnitude of the target itself.
    Here, $\| \cdot \|_2$ is the Hilbert-Schmidt norm (Frobenius norm) chosen for analytical convenience~\footnote{Note that since we mainly consider cases with the size of $A$ fixed, all norms are equivalent.}.

    In practice, we often work with the vectorized version of matrix ensembles. For a Kraus operator $K_\mathbf{z}$, its vectorized form is defined as:
    \begin{equation}
    \label{eq:vectorized-kraus}
        |K_\mathbf{z}) \equiv \sum_{ij} K_{\mathbf{z}, ij}|i\rangle|j\rangle,\quad\text{where}~K_{\mathbf{z}, ij}=\langle i|K_\mathbf{z}|j\rangle.
    \end{equation}
    Here, $\{|i\rangle\}_{i=1}^{d_A}$ is an orthonormal basis of subsystem $A$. Thus, $|K_\mathbf{z}) = \left(K_\mathbf{z}\otimes \mathbbm{1}_A\right)|\widetilde\Phi\rangle$, where $|\widetilde\Phi\rangle\equiv\sum_{i=1}^{d_A}|i\rangle|i\rangle$ is the unnormalized maximally entangled state. Hence, the Kraus ensemble in its vectorized form $\{|K_\mathbf{z})\}$ can be thought of as a so-called \textit{unnormalized quantum state ensemble} defined on a system twice the size of subsystem $A$, where each element is simply a $d_A\times d_A$ matrix \textit{flattened} as a $(d_A^2)$-dimensional column vector. This change of perspective does not alter the definition of distance in Eq.~\eqref{eq:distance-th} under the Hilbert-Schmidt norm. The definition of unnormalized state ensembles is deferred to Sec.~\ref{sec:consequence-ansatz-deep-thermalization} along with the review of deep thermalization, and more details on matrix ensembles and their distance measures can be found in Appendix~\ref{app:review}. In the discussion and outlook Sec.~\ref{sec:discussion}, we also note that the projected Kraus ensemble has a direct physical correspondence to an experimentally preparable quantum state ensemble.

    \subsection{Summary of main results}
    \label{sec:overview-main-results}
        Equipped with the preliminaries, we now state our main results informally. We present these results in increasing order of spacetime structure and thus physical generality, and we organize their scaling regimes in a spacetime diagram in Fig.~\ref{fig:concept}(b). First, following a common philosophy in studies of quantum chaotic dynamics, we consider the prototypical case where $U_{AB}$ is sampled from the \textit{unitary Haar measure}.
        This is a random-matrix-theoretic simplification of the late-time structure of dynamics generated by a physical Hamiltonian or circuit,
        which has been widely used successfully in modeling chaos and fast scrambling~\cite{hayden_mirrors_2007, Sekino_fast_2008, Hosur_chaos_2016, Roberts_chaos_2017}. We show:

        \begin{itheorem}[informal]
        \label{thm:1-informal}
            A large global Haar random unitary dynamics $U_{AB}$ typically generates a Kraus ensemble $\mathcal{K}$ on a small subsystem $A$ whose statistics follows that of the complex Ginibre ensemble, i.e., $\sqrt{d} K\sim\mathbf{GinUE}_{d_A\times d_A}$.
        \end{itheorem}

        A formal version of the theorem, which entails a probabilistic statement on the closeness of the moments of the ensembles, is postponed to Sec.~\ref{sec:global-haar} and proved in Appendix~\ref{app:global-haar-proof}. 
        Above, $\mathbf{GinUE}_{d_A\times d_A}$ corresponds to a $d_A\times d_A$ random matrix ensemble with i.i.d.~elements following a standard complex Gaussian distribution with unit variance and zero mean, and $K$ represents samples from the Kraus ensemble $\mathcal{K}$ taken with equal weights, with the factor $\sqrt{d}$ normalizing its entries to $\mathcal{O}(1)$.  
        This result sets a random matrix theory (RMT) baseline for the form of Kraus ensembles under more physical scenarios.

        Next, we reintroduce 1D spatial locality, moving one step towards physically generic quantum many-body dynamics while retaining certain analytical tractability. We consider   a special class of 1D circuit dynamics known as DU circuits~\cite{Bertini_entanglement_2019, Bertini_exact_2019, Piroli_exact_2020, Zhou_maximal_2022, Bertini_exactly_2026}, which are made up of local gates that are unitary in both space and time directions. Such circuits enjoy properties of maximal chaos~\cite{Zhou_maximal_2022} and exact solvability under special initial states~\cite{Bertini_exactly_2026}. We show: 
        
        \begin{itheorem}[informal]
        \label{thm:2-informal}
            A 1D dual-unitary circuit $U_{AB}(t)$ initialized in a solvable state and measured in a solvable basis generates a Kraus ensemble $\mathcal{K}$ whose statistics follows that of the complex Ginibre ensemble, i.e., $\sqrt{d} K\sim\mathbf{GinUE}_{d_A\times d_A}$, in the thermodynamic limit and at late times (in that order). 
        \end{itheorem}
        Here, solvable product states (bases) refer to a set of states (bases) on the bath $B$ such that the DU circuit with measurements is also unitary when viewed in the space direction, similar to the solvable MPS states defined in~\cite{Piroli_exact_2020}. This theorem follows from perfect quantum information transmission (or the lack of \textit{dynamical purification}) in the sideways evolution through the DU circuit bulk~\cite{ippoliti_dynamical_2023}. We prove this result explicitly for a DU Floquet KIM, following an argument similar to that of~\cite{ho_exact_2022}; its precise formulation is given in Sec.~\ref{sec:DU-Floquet}. Nevertheless, we expect analogous proofs to hold for generic DU circuits under similar solvability constraints on initial state and measurement basis of $B$.

        Finally, we study the most physically relevant case of  generic 1D quantum circuits, which do not enjoy such nice analytically tractable properties like in DU circuits. Accordingly, the universal Ginibre form of the Kraus ensembles has to be modified to account for dynamical purification~\cite{Gullans_dynamical_2020, fidkowski_how_2021, ippoliti_dynamical_2023} in the spatial direction of the circuit bulk. Here, we identify specific joint scaling limits of $N_B, t\rightarrow\infty$ along which we conjecture the Kraus ensemble $\mathcal{K}$ approaches a stable limiting form:
        
        \begin{iconjecture}[informal]
        \label{conj:1-informal}
            A generic 1D local quantum circuit 
            generates a 
            Kraus ensemble described statistically by a universal one-parameter Ansatz:
            \begin{equation}
            \label{eq:univeral-one-parameter-family}
                \sqrt{d} K \sim\mathbf{GinUE}_{d_A\times d_A}\times\mathrm{LogNormal}_{\sigma^2},
            \end{equation}
            in the joint limit $N_B, t \rightarrow\infty$ with fixed parameter set by 
            \begin{equation}
            \label{eq:sigma2-definition}
                \sigma^2=\frac{\lambda N_B}{ \gamma^t }.
            \end{equation}
            Here the constants $\lambda\ge0$ and $\gamma > 1$ are model-dependent but are otherwise independent of $N_A$, $N_B$ or $t$.
        \end{iconjecture}
        The distribution in Eq.~\eqref{eq:univeral-one-parameter-family} consists of two independent parts: a Ginibre matrix part identical to that in the previous theorems, and an independent scalar part following a random log-normal distribution which contributes additional fluctuations. Notably, the distribution $\mathrm{LogNormal}_{\sigma^2}\equiv\mathrm{LogNormal}(\mu=-\sigma^2, \sigma^2)$ has only one free parameter, and describes a positive random variable $Y$ with its logarithm $\log(Y)$ distributed according to $\mathcal{N}(\mu=-\sigma^2, \sigma^2)$,    a real Gaussian distribution with mean $-\sigma^2$ and variance $\sigma^2$, set jointly by spacetime sizes $N_B$ and $t$ through Eq.~\eqref{eq:sigma2-definition}. 
        Note this emergent RMT universality is different from the Ginibre universality emerging in the spectral statistics of chaotic dynamics, previously studied in~\cite{li_dissipative_2021, shivam_ginibre_2023, sa_spectral_2020}.

        Our conjecture highlights a running family of universal Ansatz that the Kraus ensembles of 1D circuits approach in the joint spacetime  limits $N_B, t\to \infty$. In particular, it implies a family of critical scalings $t = \log_\gamma(N_B) + \log_\gamma(\lambda / \sigma^2)$ with $\sigma^2 > 0$, which is in agreement with recent works on random circuit models forming unitary designs in logarithmic depth~\cite{schuster_science_2025, chen_incompressibility_2024}.
        One can verify that for times scaling strictly faster than the critical family $t=\log_\gamma(N_B) + \omega (1)$, $\sigma^2\rightarrow 0$ and the log-normal part becomes trivial. This limit is in agreement with Theorem~\ref{thm:1-informal} where global Haar random unitary dynamics can be effectively understood as taking the long-time limit $t\rightarrow\infty$ first before $N_B\rightarrow\infty$. We state the conjecture   precisely in Sec.~\ref{sec:ruc}, where we also argue for the rationale of the Ansatz as well as the emergence of this critical circuit depth scaling in Sec.~\ref{sec:origin-ansatz}.  We present extensive numerical evidence for a brickwork random unitary circuit (RUC) in Sec.~\ref{sec:numerics} supporting our conjecture, with additional numerical evidence involving an off-DU kicked Ising model (KIM) presented in  Appendix~\ref{app:extra-numerics}.

        We further expound upon the physical interpretations and implications of our Ansatz.
        As mentioned, the Ginibre part in Eq.~\eqref{eq:univeral-one-parameter-family}, common to Theorems~\ref{thm:1-informal} and~\ref{thm:2-informal}, signifies information scrambling within the conditional dynamics on subsystem $A$. Its consequences are twofold:
        \begin{itemize}[leftmargin=1.5em]
            \item Left unitary invariance of the Ansatz gives rise to a state-agnostic microscopic mechanism for deep thermalization: beginning from {\it any} input pure state $|\psi_0\rangle$ on $A$, the output projected ensemble of states $|\psi_\mathbf{z}\rangle$ is also unitarily rotationally invariant, i.e., distributed according to the Haar measure, precisely the statement of deep thermalization. This is  illustrated in Fig.~\ref{fig:consequences}(a),  and discussed  in more detail in Sec.~\ref{sec:consequence-ansatz-deep-thermalization}. 
            \item Measurement-resolved dynamics allow for quantum information recovery conditioned on the classical side information, as illustrated by Fig.~\ref{fig:consequences}(b).
            Specifically, we show that quantum information input on $A$ is recoverable on the output of $A$ up to arbitrarily high fidelity when $\mathbf{z}$ is known, 
            subject to a constant qubit overhead depending on the error threshold. We show this in Sec.~\ref{sec:consequence-ansatz-qi-recovery} through a calculation of the coherent information based on the Ansatz. Remarkably, this applies even when the bath $B$ is thermodynamically large and with logarithmic depth for 1D circuits. This constitutes a variant of the Hayden--Preskill thought experiment~\cite{hayden_mirrors_2007}, where they achieved arbitrarily high recoverability by supplementing the decoder with the quantum side information that is the maximally entangled reference to the initial black hole.
        \end{itemize}
        On the other hand, the log-normal part in Eq.~\eqref{eq:univeral-one-parameter-family}, arising from non-unitary sideways evolution in generic circuits, contributes to the fluctuations in the norms $\|K_\mathbf{z}\|_2$ of the Kraus operators. 
        The norms carry direct physical meaning as the \textit{thermally averaged} Born probabilities $\overline p_\mathbf{z}$ of measuring an outcome $\mathbf{z}$ when the input state on $A$ is taken to be the maximally mixed state $\rho_{A,0} = \pi_A\equiv \mathbbm{1}_A/d_A$. Indeed, 
        \begin{equation}
        \label{eq:norm-prob-relation}
            d_A \overline{p}_\mathbf{z} = \mathrm{Tr}\left(K_\mathbf{z}K_\mathbf{z}^\dagger \right) = \left\| K_\mathbf{z} \right\|_2^2. 
        \end{equation} 
        This part of our Ansatz is in agreement with previous studies on the universal distribution of bitstring probabilities in 1D chaotic circuits~\cite{christopoulos_universal_2025, sauliere_universality_2025, sauliere_universality_noisy_2025}.

\section{Global Haar unitary: random matrix baseline}
\label{sec:global-haar}
    We begin by studying the Kraus ensemble generated by a global Haar random unitary, a standard way of modeling the late-time behavior of the time-evolution operator of a Hamiltonian or circuit. 
    Recall that a Haar random unitary can be obtained by successively Gram--Schmidt orthonormalizing the columns of a complex Ginibre matrix of the same dimension.
    Now, if we consider any {\it finite} submatrix of a large Haar random unitary matrix, the effect of this orthonormalization procedure is weak, and so we should intuitively expect the submatrix (up to rescaling) to itself approximately follow the  Ginibre distribution on matrices equal to the size of the submatrix. 
    This expectation was indeed proven rigorously in the mathematical literature: concretely, in~\cite{mastrodonato_elementary_2007} it was shown that a finite square submatrix of a Haar random unitary defines the \textit{truncated unitary ensemble}, which, up to rescaling by $\sqrt{d}$, converges in distribution to Ginibre in the limit of large underlying $d\times d$ unitary matrix. 
    Importantly, our statement below is conceptually different in that we consider only a {\it single} realization of a global Haar random unitary, and study the Kraus ensemble induced by distinct measurement outcomes (hence the induced source of randomness is different). Nevertheless, we expect an analogous result from the intuitive argument above to hold. Indeed, we have:

    \begin{theorem}
    \label{thm:1-precise}
        Let the global unitary $U_{AB}$ be Haar random in $\mathrm{U}(d_Ad_B)$. Then for any integer $k\ge 1$ and any $0 < \delta <1$, with probability at least $1-\delta$, the $k$-th moment distance of the Kraus ensemble to Ginibre satisfies
        \begin{equation}
            \Delta_\mathrm{Gin}^{(k)} \le \sqrt{\frac{d_A^{2k}}{\delta\, k!\, d_B}\left(1+\mathcal{O}\left(\frac{k^2}{d_A^2}\right)\right)},
        \end{equation}
        for all sufficiently large $d_B$, where the $\mathcal{O}(\cdot)$ term is independent of $d_B$.
    \end{theorem}
Above, the $k$-th moment distance is defined to be 
        \begin{equation}
            \Delta_\mathrm{Gin}^{(k)} := \frac{\left\| d^k M_\mathrm{emp}^{(k)} - M_\mathrm{Gin}^{(k)} \right\|_2}{\left\|M_\mathrm{Gin}^{(k)}\right\|_2}.
        \end{equation}
    See Appendix~\ref{app:review} for more detailed properties of the Ginibre ensemble and the distance measure.

    Note that this is a probabilistic statement on the concentration of moments for any finite $k$ in the limit of $d_B\rightarrow\infty$: it says that, with high probability, a single typical Haar random unitary produces a Kraus ensemble distributed according to Ginibre with high precision. Hence, this confirms our previous intuition that Haar random global unitary leads to Ginibre-distributed Kraus ensemble, and forms an RMT baseline for understanding of physical dynamics later.

    We briefly outline the proof, whose details are presented in Appendix~\ref{app:global-haar-proof}. First, we compute the Haar-averaged squared distance $\mathbb{E}[\Delta_\mathrm{Gin}^{(k)\,2}]$ with $U_{AB}$ sampled from Haar measure using Weingarten calculus; then, we employ a useful Lemma (Lemma~1 from~\cite{mok_optimal_2025}) to control off-diagonal contributions from the Weingarten functions; finally, using Markov inequality, we show the probabilistic statement in the Theorem by leveraging the averaged quantity.

\section{Dual-unitary circuits: exactly solvable result}
\label{sec:DU-Floquet}

    We next take a step away from the RMT baseline by introducing locality in dynamics, while still retaining a certain degree of analytical control.
    To this end, we consider dynamics generated by 1D circuits built from local entangling gates which possess {\it dual-unitarity}~\cite{Bertini_entanglement_2019, Bertini_exact_2019, Piroli_exact_2020, Zhou_maximal_2022, Bertini_exactly_2026}: 
    gates whose actions are unitary both along the time and space directions.
    As a paradigmatic example, we focus below on the Floquet KIM, 
    but we expect the resulting statement and proof 
    to hold analogously for generic DU circuits under analogous conditions of solvable initial states and final measurement bases.

    We follow the setup in~\cite{ho_exact_2022}. Consider a 1D open chain of $N=N_A + N_B$ qubits with Floquet unitary
    \begin{equation}
        U_F = U_h e^{-i H_\mathrm{Ising}},
    \end{equation}
    consisting of an Ising interaction $e^{-iH_\mathrm{Ising}}$ for unit time, followed by a kick $U_h$. Here, the Ising interaction takes the form
    \begin{equation}
        H_\mathrm{Ising} = J\sum_{i=1}^{N-1} \sigma_i^z\sigma_{i+1}^z + \sum_{i=1}^Ng_i\sigma_i^z + (b_1 \sigma_1^z + b_N\sigma_N^z),
    \end{equation}
    where the longitudinal field $g_i\notin\frac{\mathbb{Z}\pi}{8}$ is allowed to have disorder, and $b_1 = b_N = \frac{\pi}{4}$ are introduced for convenience, while the kick is a transverse field rotation
    \begin{equation}
        U_h = \exp(-ih\sum_{i=1}^N\sigma_i^y).
    \end{equation}
    Therefore, a depth-$t$ global unitary is
    \begin{equation}
        U_{AB}(t)\equiv U_F^t = \left(U_h e^{-iH_\mathrm{Ising}}\right)^t.
    \end{equation}
    This Floquet model can be represented as a 1D circuit, also conveniently depicted as a tensor network~\cite{ho_exact_2022}. In Fig.~\ref{fig:kim}, we illustrate the tensor network representation of the global unitary modulo a global phase, where the single-qubit red gates are
    \begin{equation}
    \label{eq:red-square}
        \mathsf{Red}(h) = \begin{pmatrix}
            \cos(h) & \sin(h)
            \\
            \sin(h) & -\cos(h)
        \end{pmatrix},
    \end{equation}
    the single-qubit black gates are
    \begin{equation}
    \label{eq:black-dot}
        R_z(2g_i) = \begin{pmatrix}
            e^{-ig_i} & ~
            \\
            ~ & e^{ig_i}
        \end{pmatrix},
    \end{equation}
    and the two-qubit blue gates are
    \begin{equation}
    \label{eq:blue-square}
        \mathsf{Blue}(J) = \begin{pmatrix}
            1 & ~ & ~ & ~
            \\
            ~ & e^{i(2J - \pi/2)} & ~ & ~
            \\
            ~ & ~ & e^{i(2J - \pi/2)} & ~
            \\
            ~ & ~ & ~ & -1
        \end{pmatrix}.
    \end{equation}
    The detailed construction of the tensor network elements is explained in Appendix~\ref{app:DU-KIM-proof}.

    There are spacetime self-dual points for this model at $|J| = |h| = \frac{\pi}{4}$ known as the DU points, but we fix $J = h = \frac{\pi}{4}$ for concreteness in this Section.
    Upon such parameter setting, the red gate becomes a Hadamard gate, and the blue gate becomes a controlled-$Z$ gate, which is proportional to replacing the blue square with a Hadamard. It has been shown that the generator state $|\Psi_{AB}(t)\rangle = U_F^t|+\rangle^{\otimes N}$ deep thermalizes as soon as $t\ge N_A$ in the thermodynamic limit $N_B\rightarrow\infty$~\cite{ho_exact_2022}.

    \begin{figure}
        \centering
        \includegraphics[width=\linewidth]{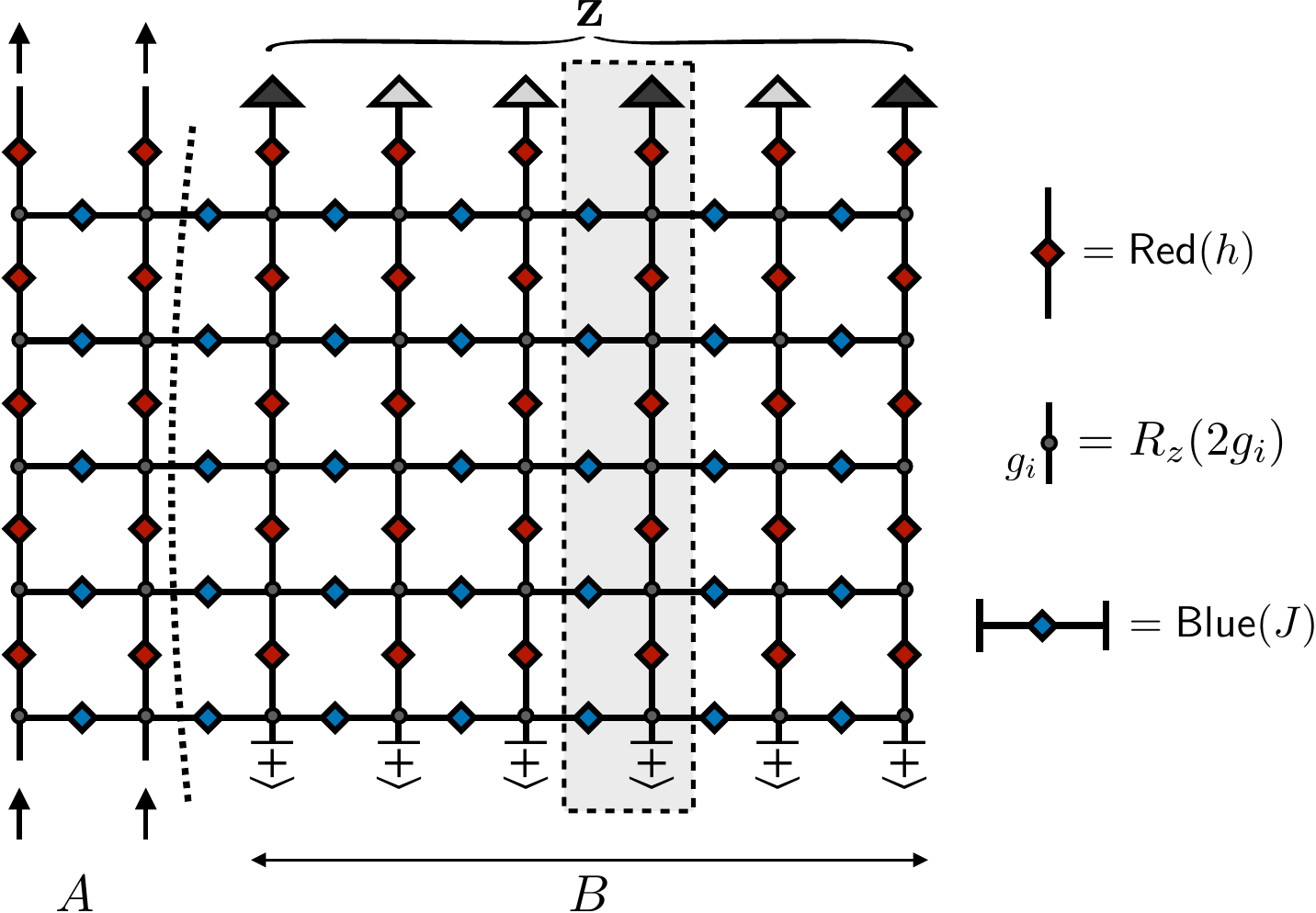}
        \caption{Schematics of kicked Ising model (KIM) represented as a tensor network modulo a global phase, illustrated with $N_A = 2$, $N_B=6$ and $t = 5$. Here, the elementary gates (red gates, black nodes, and blue gates) are defined in
        Eqs.~(\ref{eq:red-square}--\ref{eq:blue-square}) respectively. Initial state on $B$ is set as a solvable state $|B_0\rangle = |+\rangle^{\otimes N_B}$ and the measurement basis on $B$, also solvable, is fixed to be the computational basis $|\mathbf{z}\rangle\equiv\bigotimes_{i=1}^{N_B}|z_i\rangle$, where $|z_i\rangle\in\{|0\rangle, |1\rangle\}$. Upon tensor contraction, this diagram represents a Kraus operator $K_\mathbf{z}$. The shadowed region denotes a \textit{transfer matrix} $V_{z_i}$ of KIM, which propagates information in the spatial direction in the circuit bulk. The spatial cut between $A$ and $B$ is shown with a dashed curve.}
        \label{fig:kim}
    \end{figure}

    For our study of Kraus ensemble of the DU KIM, let us set the initial bath state as a solvable state $|B_0\rangle = |+\rangle^{\otimes N_B}$ and assume a uniform longitudinal field $g_i = g \notin\mathbb{Z}\pi/8$, and finally measure $B$ in the computational basis (a solvable basis). The Floquet dynamics is thus deterministic, and we have the following result:
    \begin{theorem}
    \label{thm:2-precise}
        For the dual-unitary kicked Ising model with uniform $g_i = g \notin \frac{\mathbb{Z}\pi}{8}$ in the thermodynamic limit $N_B\rightarrow\infty$, for any integer $k\ge 1$ and any $t\ge 2N_A$, the Kraus ensemble $\mathcal{K}(t)$ satisfies
        \begin{equation}
            \lim_{N_B\rightarrow\infty} \Delta_\mathrm{Gin}^{(k)}(t) = 1 - \frac{2^{kt}}{(2^t)_k},
        \end{equation}
        where $(x)_k:= x(x+1)\cdots(x+k-1)$ is the rising factorial. For fixed $k$, this distance decays exponentially in time as
        \begin{equation}
            \lim_{N_B\rightarrow\infty}\Delta_\mathrm{Gin}^{(k)}(t) = \mathcal{O}\left(\frac{k^2}{2^t}\right),\qquad(t\rightarrow\infty).
        \end{equation}
    \end{theorem}
    Note that we take the thermodynamic limit (TDL) first, so the limiting $k$-th moment distance has no $d_B$ dependence. 
    
    Our proof closely follows that of~\cite{ho_exact_2022}, which we sketch as follows. Instead of viewing the quantum circuit as effecting (unitary) evolution in time, i.e., reading its action vertically in Fig.~\ref{fig:kim}, we can view it as effecting evolution {\it sideways} in space, i.e., read its action horizontally from right to left to compute the Kraus map $K_{\mathbf{z}}$. 
    In other words, we perform a {\it spacetime duality transformation}.
    Crucially, because the local gates of the circuit are DU and the initial state and final measurement basis on $B$ are solvable, the sideways evolution is also unitary, i.e., the transfer matrices defined by the gray box in Fig.~\ref{fig:kim} are proportional to $t$-qubit unitaries.
    Note that there are two different transfer matrices $V_{z_i}$, related to the measurement outcomes $z_i\in\{0, 1\}$ at site~$i$.
    It was further shown by~\cite{ho_exact_2022} that the set of transfer matrices $\{ V_0, V_1\}$ for the DU KIM with parameters as specified in the theorem   forms a universal gate set (meaning that any $t$-qubit unitary can be formed by some long products of them), which further implies that the states living on the $t$ qubits of the spatial cut between $A$ and $B$ (dashed curve of Fig.~\ref{fig:kim}), when considered over all measurement outcomes, are {\it Haar-randomly distributed} in the large $N_B$ limit. 
    Finally, one can derive that the remaining portion of the circuit on $A$ (i.e., the part of the circuit to the left of the dashed line in Fig.~\ref{fig:kim}) behaves proportionally to an isometry projecting the $t$-qubit Haar random  state to a $2N_A$-qubit state (note half of these qubits are the input of $A$ and half of these are the output of $A$). 
    For details of the proof, see Appendix~\ref{app:DU-KIM-proof}. 
    For generic DU circuits, we simply replace the initial state and measurement basis with a solvable initial state and measurement scheme, and all statements follow analogously.

\section{Generic 1D circuits and universal Ansatz}
\label{sec:ruc}

We now turn to the physically relevant scenarios of {\it generic} 1D local quantum circuits and their emergent Kraus ensembles. 
    To concretize discussions below, consider a minimally structured model for dynamics of a brickwork unitary circuit acting on qudits with local dimension $q$. Such circuits constitute a wide class of spatially extended many-body quantum dynamics~\cite{fisher_random_2023, chalker_random_2026}. Precisely, let a 1D array of $N$ (taken to be even for convenience) qudits interact through layers of nearest-neighbor 2-qudit gates, arranged in a brickwork architecture with open boundary condition as shown in Fig.~\ref{fig:ruc}. Dynamics is built up by successively applying unitaries
    \begin{equation}
    U_\tau=
    \begin{cases}
        \bigotimes_{i=1}^{N/2} u_{2i-1,2i}^{(\tau)} &\text{odd }\tau
        \vspace{5pt}
        \\
        \bigotimes_{i=1}^{N/2 - 1}u_{2i,2i+1}^{(\tau)} &\text{even }\tau,
    \end{cases}
    \end{equation}
    for the odd $\tau$ and even $\tau$ circuit layers, where each local $u_{i,j}^{(\tau)}$ gate acts on the $i$-th and $j$-th site and is selected from $\mathrm{U}(q^2)$, the space of 2-qudit gates. The choice of each local gate may or may not be random. A depth-$t$ circuit is thus described by
    \begin{equation}
        U_{AB}(t)\equiv U_t U_{t-1}\cdots U_2 U_1.
    \end{equation}
    We consider a finite $N_A$ number of qudits as $A$ and fix the initial state on $B$ as $|B_0\rangle = |0\rangle^{\otimes N_B}$, and measure $B$ in the qudit computational basis following dynamics. 
    
    \begin{figure}
        \centering
        \includegraphics[width=\linewidth]{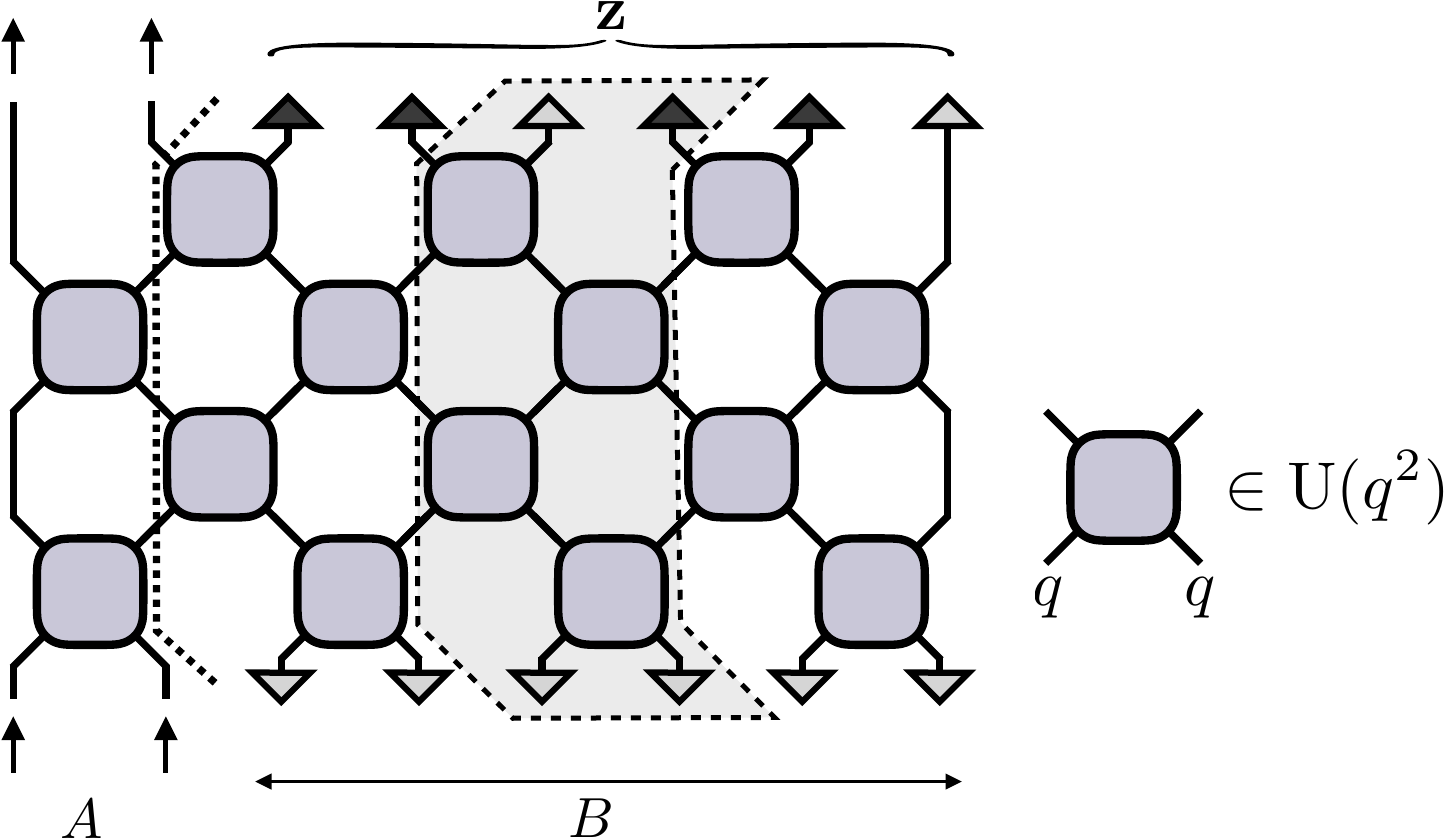}
        \caption{Schematic of a 1D brickwork unitary circuit of qudits, illustrated here with $N_A=2$, $N_B=6$, and $t=4$. Here, the local gates are selected from $\mathrm{U}(q^2)$, and the final measurement on $B$ is projective in local computational basis $|\mathbf{z}\rangle\equiv\bigotimes_{i=1}^{N_B}|z_i\rangle$, where $|z_i\rangle\in\{|a\rangle\}_{a=0}^{q-1}$ is the qudit computational basis. Upon tensor contraction, this diagram represents a Kraus operator $K_\mathbf{z}$. The shadowed region denotes a \textit{transfer matrix} of the brickwork unitary circuit, which propagates information in the spatial direction in the circuit bulk. The spatial cut between $A$ and $B$ is shown with a dashed curve.}
        \label{fig:ruc}
    \end{figure}
    
    To analyze the Kraus ensembles arising from this system, we can similarly perform a spacetime dual transformation like in the DU case, and view the circuit as effecting evolution sideways through space to compute the Kraus maps $K_{\mathbf{z}}$. However, a crucial difference is that the transfer matrices $V_{z_i}$ which effect evolution in space are no longer unitary.
    Thus, tracing the argument of the emergence of Ginibre distribution in the DU circuit case, we see that a $t$-qudit temporal state living on the spatial slice between $A$ and $B$ (dashed curve of Fig.~\ref{fig:ruc}), which is the result of a long product of transfer matrices $V_{z_i}$,
    is no longer Haar randomly distributed. In particular, while one may expect its orientation to still be close to uniformly distributed, its norm is not unity but fluctuates depending on the measurement outcomes. This motivates a modification to the previous idealized Ginibre distribution. 
    As we argue below in Sec.~\ref{sec:origin-ansatz} and verify numerically in Sec.~\ref{sec:numerics}, we identify this as a $\mathrm{LogNormal}_{\sigma^2}$ component, where the parameter $\sigma^2$   encodes a nontrivial  scaling of space and time along which the conjectured target ensemble $\mathbf{GinUE}\times\mathrm{LogNormal}_{\sigma^2}$ arises: 

        \begin{conjecture}
        \label{conj:1-precise}
            For a generic local 1D circuit model with no conserved quantities, let $U_{AB}(t)$ be its depth-$t$ circuit with fixed subsystem size $N_A$ and variable bath size $N_B$. We take $B$ to be initialized in a simple product state and measured in the computational basis. Then, in the joint limit $N_B,t\rightarrow\infty$ with
            \begin{equation}
            \label{eq:sigma2}
                \sigma^2 = \frac{\lambda N_B}{\gamma^t},
            \end{equation}
            held fixed, the $k$-th moment distance between the Kraus ensemble $\mathcal{K}(t)$ and target ensemble $\mathbf{GinUE}\times\mathrm{LogNormal}_{\sigma^2}$ decays exponentially in time as
            \begin{equation}
                \Delta_{\mathrm{th}}^{(k)} \sim e^{-c_k t},\qquad(N_B,t\rightarrow\infty\text{ with }\sigma^2\text{ fixed}),
            \end{equation}
            with decay rates $c_k > 0$. Here, $\lambda\ge 0$ and $\gamma > 1$ are model-dependent constants but are otherwise independent of $N_A$, $N_B$, or $t$. 
        \end{conjecture}
Above, the $k$-th moment distance is defined as
    \begin{equation}
    \label{eq:distance-th-conj-1}
        \Delta_\mathrm{th}^{(k)} = \frac{\left\| d^k M_\mathrm{emp}^{(k)} - M_\mathrm{th}^{(k)}(\sigma^2)\right\|_2}{\left\| M_\mathrm{th}^{(k)}(\sigma^2) \right\|_2},
    \end{equation}
where $k$-th moment $M_\mathrm{th}^{(k)}(\sigma^2)$ of the theoretical ensemble $\mathbf{GinUE}\times\mathrm{LogNormal}_{\sigma^2}$ has the form 
    \begin{equation}
        \label{eq:kth-moment-family}
            M_\mathrm{th}^{(k)}(\sigma^2) = \exp\left(2k(k-1)\sigma^2\right)M_\mathrm{Gin}^{(k)},
        \end{equation}
see Appendix~\ref{app:review}. Note that this reduces to the $k$-th moment of the Ginibre distribution when $\sigma^2=0$. 
        
    This conjecture unveils a  universal spacetime scaling Eq.~\eqref{eq:sigma2} in generic 1D quantum circuits along which the Kraus ensemble converges to a one-parameter universal form, specified by non-universal, model-dependent parameters $(\lambda, \gamma)$.
    More specifically, the joint scalings follow $t = \log_\gamma(N_B) + \log_\gamma(\lambda / \sigma^2)$, which describe a family of vertically-shifted curves which foliate the spacetime diagram, shown pictorially in Fig.~\ref{fig:concept}(b).

    The family of critical joint scalings further separates the limiting behavior of the Kraus ensembles under two distinct scaling regimes. For deep circuits, i.e., circuit depth $t=\log_\gamma(N_B)+\omega(1)$ growing faster than $\log_\gamma(N_B) + c$, where $c$ is an arbitrary real constant (for example, $t \sim \text{poly}(N_B)$), we have $\sigma^2\rightarrow 0$. Thus, we can expect the emergent Kraus ensemble to be described by a plain $\mathbf{GinUE}$ without the fluctuation corrections. 
    This scaling is consistent with recent results~\cite{schuster_science_2025, chen_incompressibility_2024} showing that random 1D circuits can form approximate unitary designs in logarithmic depth, and are ``practically indistinguishable'' from Haar random unitaries. 

    Conversely, for shallow circuits, i.e., circuit depth $t = \log_\gamma(N_B) - \omega(1)>0$ growing slower than $\log_\gamma(N_B) + c$,  $\sigma^2\rightarrow\infty$ and so there is no nontrivial limiting stable distribution that the Kraus ensemble $\mathcal{K}$ converges to. In this case, there are substantial fluctuations in the norms of the Kraus operators across orders of magnitude that become divergent in this limit. This is a prediction consistent with the universal log-normal distribution emerging in global bitstring probabilities found in shallow 1D chaotic circuits~\cite{christopoulos_universal_2025, sauliere_universality_2025, sauliere_universality_noisy_2025}.
    We include a discussion on self-consistency of the Ansatz, particularly in the diverging limit, in Appendix~\ref{app:ansatz-more-details}.

    \subsection{Origin of Ansatz}
    \label{sec:origin-ansatz}

        We now motivate the form of our Ansatz by invoking the intuition of dynamical purification~\cite{Gullans_dynamical_2020, fidkowski_how_2021} and long products of random matrices~\cite{bulchandani_random-matrix_2024,de_luca_universality_2023, Gerbino_dyson-brownian_2024, mochizuki_measurement-induced_2025}.

        As illustrated in Fig.~\ref{fig:oseledets}, the Kraus operators of any 1D local circuit can be obtained by sideways contracting the tensor network representation of the circuit. Concretely, each gray box (viewed from right to left) is a transfer matrix $V_{z_i}$ consisting of a vertical slice of the circuit on $B$ together with the initial state and final measurement outcome $z_i$ on that slice. The transfer matrices may be unitary (special DU solvable case) or non-unitary (generic case), corresponding to discussions in Secs.~\ref{sec:DU-Floquet}~and~\ref{sec:ruc}, respectively. Here, we focus on the latter.

        \begin{figure}
            \centering
            \includegraphics[width=1\linewidth]{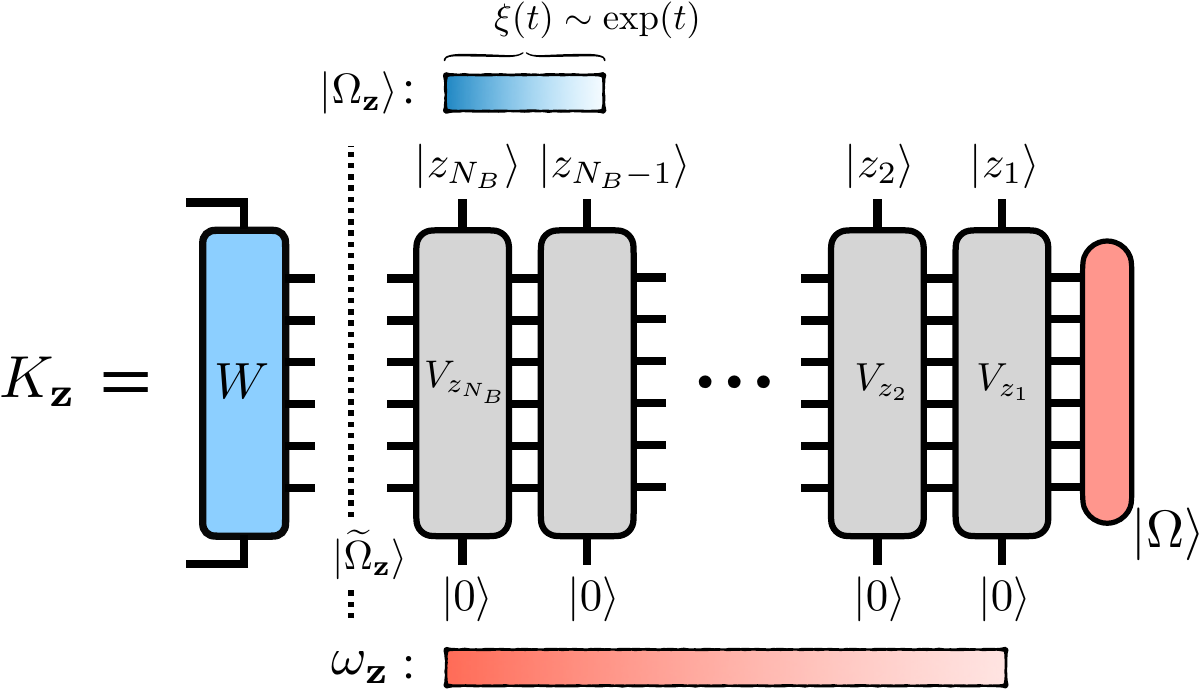}
            \caption{ 
            Schematics of a 1D circuit viewed as tensor network. $K_\mathbf{z}$ can be viewed as the circuit contracted from right to left, equivalently a $t$-qubit state $|\Omega\rangle$ evolving through transfer matrices 
            followed by left multiplying the edge operator $W$. The unnormalized state can be decomposed into $|\widetilde{\Omega}_\mathbf{z}\rangle = \exp(\omega_\mathbf{z})|\Omega_\mathbf{z}\rangle$, namely its length and orientation. The memory span of $|\Omega_\mathbf{z}\rangle$ is shown as the blue bar with a typical length $\xi(t)\sim\exp(t)$ corresponding to the purification length~\cite{ippoliti_dynamical_2023}, while $\omega_\mathbf{z}$ depends on the entire history $\mathbf{z}$ with length $N_B$ shown as the red bar. Thus the distributions of $|\Omega_\mathbf{z}\rangle$ and $\omega_\mathbf{z}$ are expected to be decoupled when $N_B\gg\xi(t)$.
            }
            \label{fig:oseledets}
        \end{figure}

        Concretely, a Kraus operator is represented in Fig.~\ref{fig:oseledets} as a $t$-qubit state $|\Omega\rangle$ (shown as a red box) passing through the non-unitary circuit bulk $B$, resulting in a $t$-qubit unnormalized state $|\widetilde{\Omega}_\mathbf{z}\rangle=\overset{\leftarrow}{\prod}{}^{N_B}_{l=1}V_{z_l}|\Omega\rangle$, which is then acted on by the edge of the circuit $W$ (the blue box):
        \begin{equation}
            |K_\mathbf{z}) = W|\widetilde{\Omega}_\mathbf{z}\rangle.
        \end{equation}
        Here, $|K_\mathbf{z})$ represents the vectorized Kraus operator, and $W$ is an operator that maps a $t$-qubit state to a $(2N_A)$-qubit state. The unnormalized state can be written as $|\widetilde{\Omega}_\mathbf{z}\rangle = \exp({\omega_\mathbf{z}})|\Omega_\mathbf{z}\rangle$, where $|\Omega_\mathbf{z}\rangle$ is a unit vector and $\exp({\omega_{\mathbf{z}}}) = \| |\widetilde{\Omega}_\mathbf{z}\rangle\|$ is the norm. This allows us to separate the effect of bulk and edge of the circuit by
        \begin{equation}
        \label{eq:kraus-from-tensor}
            |K_\mathbf{z}) = W|\Omega_\mathbf{z}\rangle\times \exp(\omega_\mathbf{z}).
        \end{equation}
        
        In the following, we argue that $W|\Omega_\mathbf{z}\rangle$ reproduces the Ginibre part, while $\exp(\omega_\mathbf{z})$ follows a log-normal distribution. We further argue that the two parts are asymptotically independent in the $N_B\gg\xi(t)$ limit, where $\xi(t)\sim\exp(t)$ is known as the \textit{purification length}~\cite{ippoliti_dynamical_2023}, hence motivating the factorized form of our Ansatz.

        \paragraph{Oseledets} First, we explain the origin of the log-normal term. Oseledets' multiplicative ergodic theorem~\cite{Oseledets_1968} states that the product of a long sequence of $L$ i.i.d. random matrices  $V_i$ has the following limiting behavior:
        \begin{equation}
        \label{eq:oseledets-original}
            \lim_{L\rightarrow\infty}\frac{1}{L}\log\|V_L\cdots V_2V_1\| = \lambda_1^{(\infty)}\quad\text{almost surely,}
        \end{equation}
        where $\{V_i\}_{i=1}^{L}$ are $D\times D$ matrices sampled from the same distribution with a mild regularity assumption~\footnote{Oseledet's theorem assumes the logarithmic integrability condition $\mathbb E\left[\max(\log(\|V\|), 0)\right]<\infty$ for the underlying distribution of $V_i$'s. This prevents excessively heavy upper tails for the single-step matrix norm $\| V\|$}.
        For finite matrices, the particular choice of norm $\|\cdot\|$ does not matter here as they are all equivalent and would lead to the same {\it fixed, non-random} $\lambda_1^{(\infty)}$, known as the leading \textit{Lyapunov exponent}.
        For our purposes, consider the long product of transfer matrices in the bulk acting on $|\Omega\rangle$:
        \begin{equation}
            V_{z_{N_B}} \cdots V_{z_2}V_{z_1}|\Omega\rangle = \exp(\omega_\mathbf{z})|\Omega_\mathbf{z}\rangle,
        \end{equation}
        with $D=2^t$ for the $t$ qubits. From Oseledet's theorem one would expect that
        \begin{equation}
            \omega_\mathbf{z} = \lambda_1^{(\infty)}N_B + \mathcal{O}(\sqrt{N_B}) 
        \end{equation}
        for some fixed $\lambda_1^{(\infty)}$; fluctuations 
        arise only in the $\mathcal{O}(\sqrt{N_B})$ term, which is asymptotically suppressed relative to the $N_B$ term in the $N_B\rightarrow\infty$ limit. The collection of log norms $\{\omega_\mathbf{z}\}$ thus follows a distribution with a linear drift in $N_B$ and a diffusive spread as $\sqrt{N_B}$. This aligns with the simple intuition of multiplying independent random variables, which can be identified as the sum of their logs, which would then follow the conventional central limit theorem. Hence, it is natural to propose a distribution for the log norms $\{\omega_\mathbf{z}\}$ as:
        \begin{equation}
            \omega \sim \mathcal{N}(\mu =\lambda_1^{(\infty)}N_B, \sigma^2\propto N_B),
        \end{equation}
        which gives rise to the log-normal distribution of $\exp(\omega_\mathbf{z})$. Note that this highlights the $N_B$ dependence, but does not yet address the $t$ dependence of the Ansatz.

        Next, to understand the origin of Ginibre, we point out that
        in the limit of long product, the vector $|\Omega_\mathbf{z}\rangle$ follows a stationary distribution described by the so-called \textit{Furstenberg measure}~\cite{Furstenberg_noncommuting_1963}. It is a distribution of normalized $D$-dimensional vectors that remains invariant under the application of an additional single-step random matrix ($V$) followed by normalizing again to unit vectors. Generically, the Furstenberg measure depends on the distribution of $V$ and can be highly nontrivial.
        Nevertheless, for chaotic local dynamics, we expect a long sequence of random transfer matrices to drive the direction of $|\Omega_\mathbf{z}\rangle$ towards a stationary distribution such that no preferred direction survives after subsequent action of the circuit edge $W$.
        For the DU KIM, a stronger statement holds analytically:
        the stationary distribution itself is exactly the Haar distribution on $t$-qubit states, a key result shown by~\cite{ho_exact_2022} that we discuss further in Appendix~\ref{app:DU-KIM-proof}.

        To connect this back to the Kraus ensemble, recall in Eq.~\eqref{eq:kraus-from-tensor} that the edge operator $W$ maps the $D=2^t$-dimensional space to a $2^{2N_A}$-dimensional space. For $t\gg 2N_A$, we expect $W$ to become increasingly well approximated as an isometry that partially projects out $(t-2N_A)$ qubits, up to a rescaling factor. Thus, we expect $W|\Omega_\mathbf{z}\rangle$ to produce an isotropic Gaussian vector, i.e., a Ginibre matrix up to proper reshaping and rescaling. This is verified for the DU KIM in Appendix~\ref{app:DU-KIM-proof} by showing that the edge $W$ is proportional to a perfect isometry whenever $t\ge 2N_A$, and tested numerically for generic 1D circuits in Sec.~\ref{sec:numerics}.

        At this point, however, there is no guarantee that the distribution of the Ginibre part $W|\Omega_\mathbf{z}\rangle$ is decoupled from the log-normal part $\exp(\omega_\mathbf{z})$. To motivate this, we recall the physics of \textit{dynamical purification} in monitored quantum circuits~\cite{Gullans_dynamical_2020, fidkowski_how_2021, bulchandani_random-matrix_2024,de_luca_universality_2023, Gerbino_dyson-brownian_2024, mochizuki_measurement-induced_2025}, which is the phenomenon of how quantum systems gradually ``forget'' about their initial states in non-unitary dynamics. Originally studied for unitary circuits interspersed with local measurements at some rate, it is shown that in the low-rate regime (or weakly monitored regime) a \textit{mixed dynamical phase} emerges---a system retains its memory for a characteristic time $t\sim\exp(N)$ known as the \textit{purification time} scaling exponentially in system size (conversely in the high-rate regime this time becomes $O(\log N)$). In~\cite{ippoliti_dynamical_2023}, it was established that the non-unitary dynamics in the space direction of a generic 1D unitary circuit maps to the  mixed dynamical phase of monitored quantum circuits and
        thus has 
        information propagating within the circuit bulk with an exponential characteristic decay length $\xi(t)\sim\exp(t)$.

        Consequently, this suggests that $|\Omega_\mathbf{z}\rangle$ in the evaluation of the Kraus map $K_{\mathbf{z}}$ has only memory from the nearest $\sim\xi(t)$ transfer matrices of the long product, i.e., those immediately to the right of the cut between $A$ and $B$. On the other hand, the scalar part $\exp(\omega_\mathbf{z})$ contains contributions from the entire history of the $N_B$ transfer matrices in the long product. Therefore, we expect the two parts to be asymptotically independent when $N_B\gg\xi(t)\sim\exp(t)$. This is illustrated by the colored bars in Fig.~\ref{fig:oseledets}.
        
        \begin{figure*}[!t]
            \centering
            \includegraphics[width=0.9\linewidth]{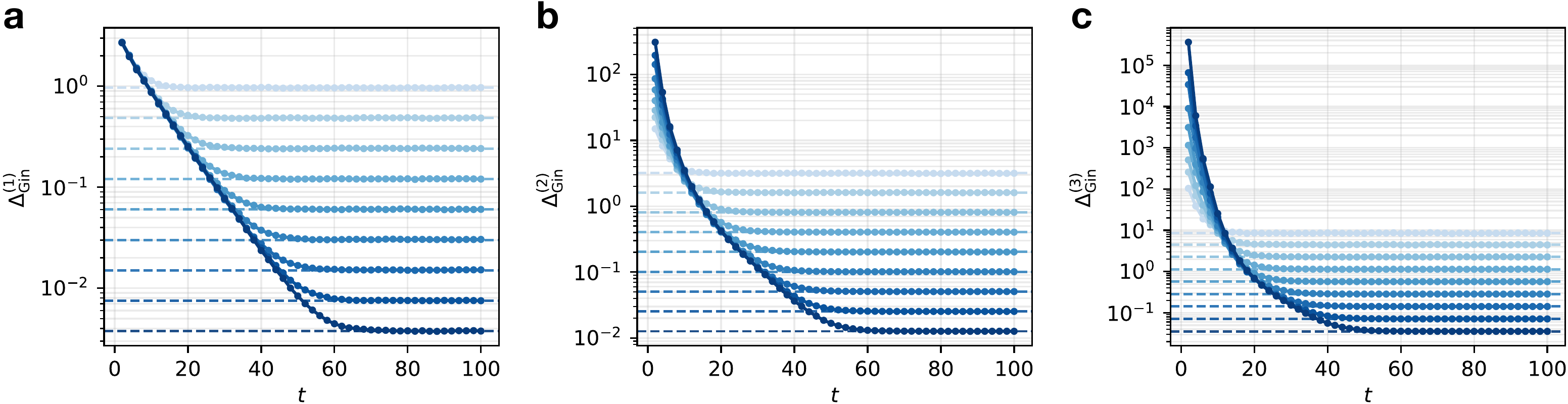}
            \caption{Distance between empirical Kraus ensembles of brickwork RUC and the fixed Ginibre ensemble. Here we fix $N_A=2$ and show $N_B=4,6,8,10,12,14,16,18,20$ (light-dark). (a-c): Distance $\Delta^{(k)}_\mathrm{Gin}$ shown for $k=1,2,3$. Each data point is averaged over $100$ circuit realizations. The dashed lines represent the averaged distance $\Delta_\mathrm{Gin}^{(k)}$ for global Haar random unitary (also over $100$ realizations) with corresponding bath sizes $N_B$, and they scale as $1/\sqrt{d_B}$ (light-dark). For $k=2,3$, early-time curves diverge and do not follow a common envelope, hinting at the failure of the Ginibre ensemble at capturing early-time behaviors.
            }
            \label{fig:dist-Gin}
        \end{figure*}

        \paragraph{Full Ansatz} While the physical motivation laid out above applies only to the $N_B\gg\xi(t)$ scenario, we nevertheless assume factorization of the Ginibre and log-normal parts when extending to a full Ansatz for arbitrary $N_B$ and $t$~\footnote{In the deep circuit $\sigma^2\rightarrow 0$ limit, the log-normal part of the Ansatz vanishes and asymptotic statistical independence between the two parts is trivially satisfied. Thus, the only nontrivial extension of the factorized Ansatz concerns scaling regimes with $\sigma^2 > 0$ but not yet $N_B\gg\xi(t)$, which we will verify numerically.}.
        This leads us to the form for the rescaled Kraus ensemble as
        \begin{equation}
            \sqrt{d}K \sim\mathbf{GinUE}\times\mathrm{LogNormal}(\mu, \sigma^2),
        \end{equation}
        where $\sigma^2 = a(t) N_B$ is proportional to $N_B$, and the normalization constraint on Kraus maps enforces that $\mu = - \sigma^2 = -a(t) N_B$~\footnote{The normalization condition of Kraus operators is written as $\sum_\mathbf{z}K_\mathbf{z}^\dagger K_\mathbf{z} = \mathbbm{1}_A$. When one constructs the Ansatz, it must satisfy $\mathbb{E}[K^\dagger K]=\frac{1}{d_B}\mathbbm{1}_A$. The only allowable log-normal settings are $Y\sim\mathrm{LogNormal}(\mu=-\sigma^2,\sigma^2)$ such that $\mathbb{E}[Y^2]=1$. See Appendix~\ref{app:review} for more details.}. 
        Our next step is to then determine the functional form of $a(t)$.
        Note as a  consistency check, at  late times $t\rightarrow\infty$ the Ansatz should reproduce the global Haar scenario, and thus the log-normal part should vanish. This implies  $\lim_{t\rightarrow\infty}a(t) = 0$.

        Now, for 1D circuits, we expect the $t$ dependence of $a(t)$  to  be related to the average shrink rate (i.e., negative growth rate) of a typical $t$-qubit normalized state $|\Omega_F\rangle$ (sampled from the \textit{Furstenberg measure}) passing through a single-step transfer matrix $V$ within the bulk. Concretely, we consider the asymptotic behavior of the average growth rate $\mathbb{E}\left[\log( \left\| V|\Omega_F\rangle\right\|)\right]$. If $V$ were unitary, this growth rate would be $0$ independently of $t$. However, for non-unitary transfer matrices $V$ of generic 1D circuits~\footnote{Here, we impose a suitable rescaling on $V$ such that $\mathbb{E}\left[\langle\Omega_F|V^\dagger V|\Omega_F\rangle\right] = 1$ without loss of generality. This is to isolate the drift caused by non-unitarity of the transfer matrices from that caused trivially by the $\sqrt{d}=\sqrt{2^{N_B}}$ rescaling factor for the Kraus ensemble.}, we expect
        \begin{equation}
        \label{eq:average-growth-main-text}
            \mathbb{E}\left[\log( \left\| V|\Omega_F\rangle\right\|)\right] \approx - \frac{\lambda_\mathrm{eff} }{\gamma_\mathrm{eff}^t}.
        \end{equation}
        Here, $\lambda_\mathrm{eff}\ge 0$ and $\gamma_\mathrm{eff} > 1$ are model-dependent constants, and for long products of random matrices, this average shrink in log norm sets the leading Lyapunov exponent $\lambda_1^{(\infty)} = -{\lambda_\mathrm{eff}}/{\gamma_\mathrm{eff}^t}$ in Eq.~\eqref{eq:oseledets-original}. This result can be analytically derived in a toy example of a $t$-qudit Haar random state passing through a single layer of dual Haar random gates (i.e., spacetime dual transformed $2$-qudit Haar random gates), and we further expect similar behavior to hold for generic transfer matrices from 1D circuits. Details of this toy example and numerical verification of Eq.~\eqref{eq:average-growth-main-text} are presented in Appendix~\ref{app:shrink}.
        
        From Eq.~\eqref{eq:average-growth-main-text}, we see that as $t$ increases, the average shrink rate produced by a single-step transfer matrix $V$ correspondingly decays exponentially with $t$. Connecting back to the Ansatz for Kraus maps, we expect $a(t)$ for 1D circuits to be inverse exponential in $t$ as well, namely
        \begin{equation}
            a_\text{1D-circ}(t) = \frac{\lambda}{\gamma^t},
        \end{equation}
        where $\lambda\ge0$ and $\gamma > 1$ are \textit{dressed} versions of $\lambda_\mathrm{eff}$ and $\gamma_\mathrm{eff}$ due to the inclusion of the circuit edge $W$, but the functional scaling form should be fixed by virtue of the structure of transfer matrices in generic 1D circuits.
        We therefore arrive at our full Ansatz:
        \begin{equation}
            \sqrt{d}K\sim\mathbf{GinUE}\times\mathrm{LogNormal}_{\sigma^2},\quad\sigma^2=\frac{\lambda N_B}{\gamma^t},
        \end{equation}
        as advertised.

    \subsection{Numerical verification}
    \label{sec:numerics}

        In this section, we present detailed numerical results for the Kraus ensembles of circuit models in support of our Ansatz. 
        Specifically, we focus on the brickwork random unitary circuit (RUC)~\cite{Nahum_operator_2018, chan_projected_2024} of qubits ($q=2$), which is described by the circuit architecture shown in Fig.~\ref{fig:ruc} where each gate is i.i.d. sampled from the Haar measure of $\mathrm{U}(4)$. Importantly, we stress that we always generate a Kraus ensemble for a fixed realization of the underlying circuit first, before ensemble-averaging the resulting distance quantities of interest.

        As a first rudimentary investigation,  we plot in  Fig.~\ref{fig:dist-Gin}  the distance $\Delta^{(k)}_\mathrm{Gin}$ between the empirical Kraus ensembles and the plain Ginibre ensemble (i.e., without incorporating the log-normal component in our Ansatz). We see that for any fixed $N_B$ and $k$, the distance $\Delta^{(k)}_\mathrm{Gin}$ indeed plateaus at late $t$  to the same value of $\Delta^{(k)}_\mathrm{Gin}$ obtained from global Haar random dynamics with the same system size; these decay like $1/\sqrt{d_B}$ in accordance with Theorem~\ref{thm:1-precise}.
        However, the numerics for $k=2,3$ also show that the distance $\Delta^{(k)}_\mathrm{Gin}$ diverges with $N_B$ at any fixed $t$, indicating that there is a nontrivial scaling of space and time that needs to be identified, and that the  Ginibre ensemble as solely the target ensemble may not be appropriate, in line with our conjecture.

        We thus compare the empirical Kraus ensembles with a running target (one-parameter Ansatz) instead of a fixed target (Ginibre).
        Given a Kraus ensemble $\mathcal{K}(N_B, t)$ denoted by the spacetime size $(N_B, t)$ of the circuit that generates it, our objective is twofold: firstly, we should identify the member in the family of one-parameter Ansatz that best describes the Kraus ensemble at this size. We extract this through a maximum likelihood estimation, resulting in a collection of estimates $\hat\sigma^2(N_B,t)$ which we could further fit to a theoretical scaling prediction $\sigma^2_\mathrm{th}(N_B, t) = \lambda N_B/\gamma^t$ with respect to $N_B$ and $t$; secondly, we need to analyze how good an approximation the theoretical distribution is to the empirical ensemble. To this end, we extract numerics for the distances $\Delta^{(k)}_\mathrm{th}$ to our Ansatz. These two steps are addressed in Secs.~\ref{sec:probing-norms}~and~\ref{sec:running-target}, respectively.

        \subsubsection{Probing distribution of $\|K_\mathbf{z}\|$}
        \label{sec:probing-norms}
        To extract the best estimates $\hat\sigma^2(N_B,t)$, instead of directly working with the full matrices in $\mathcal{K}(N_B, t)$, we identify a scalar quantity closely related to the norms $\|K_\mathbf{z}\|_2$ of the Kraus maps and fit to those scalars.

        Given the Ansatz $\sqrt{d}K\sim \mathbf{GinUE}\times\mathrm{LogNormal}_{\sigma^2}$, we observe that $\sigma^2$ can be estimated from the norms $\|K\|_2$. To see this, let $G\sim\mathbf{GinUE}$ be a $d_A\times d_A$ random Ginibre matrix and $Y\sim\mathrm{LogNormal}_{\sigma^2}$ an independent log-normally distributed random scalar, then $\sqrt{d}K\overset{d}{=}YG$ in distribution, where $K$ represents an ideal Kraus operator. Let $R = \mathrm{log}(\sqrt{d}\|K\|_2)$ be the rescaled log norm, and $X = \|G\|^2_2$, then
        \begin{equation}
            R \overset{d}{=} \log(Y \cdot \|G\|_2) \overset{d}{=} \frac{1}{2}\log(X) + \log(Y),
        \end{equation}
        namely, $R$ is equal in distribution to the sum of two independent random scalars. Here, $X = \|G\|_2^2 = \sum_{ij}|G_{ij}|^2$ is the sum of squares of independent Gaussians, and it follows that $X\sim\Gamma(d_A^2,1)$, an Erlang distribution. $\log(Y)$ follows a normal distribution $\mathcal{N}(\mu=-\sigma^2, \sigma^2)$ by definition. Therefore, the theoretical distribution of $R$ can be obtained by a convolution, and we denote its probability density function (PDF) as $p_{\sigma^2}(r)$, parametrized by $\sigma^2$. The detailed form of $p_{\sigma^2}(r)$ can be found in Appendix~\ref{app:ansatz-more-details}.

        \begin{figure}[!t]
            \centering
            \includegraphics[width=1.00\linewidth]{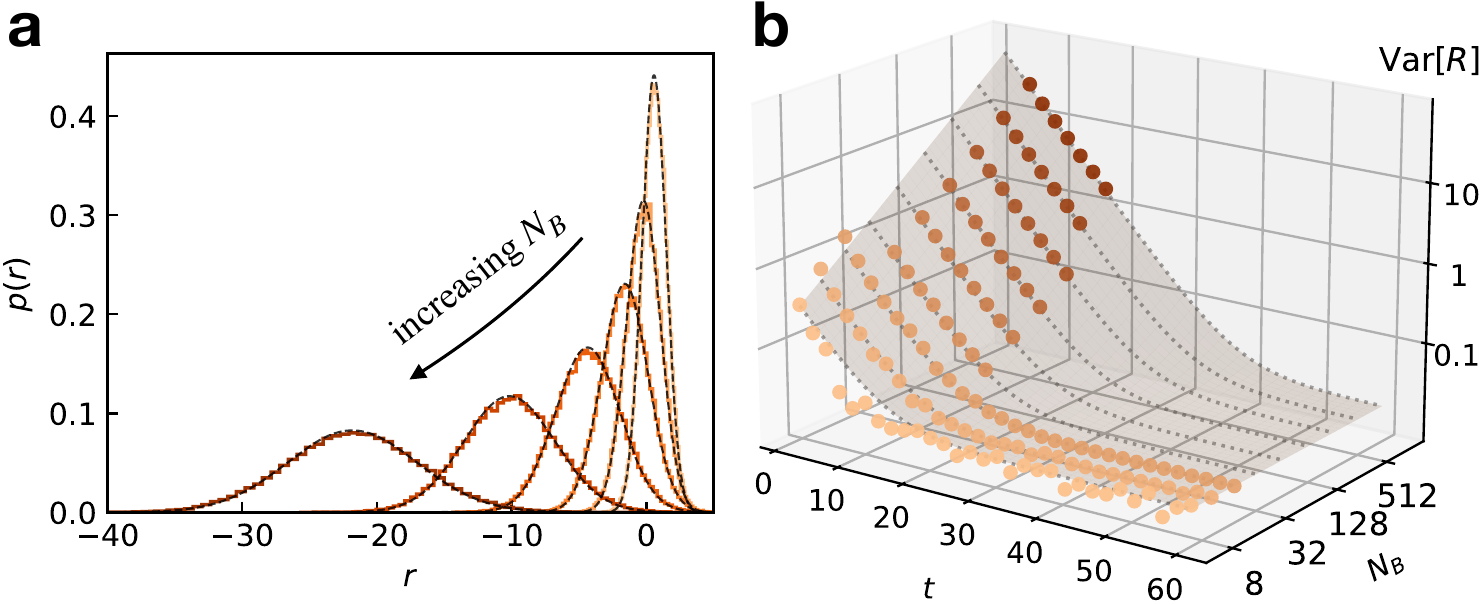}
            \caption{Fitting data shown for RUC. (a) Histogram of rescaled log norms $R_\mathbf{z}=\mathrm{log}(\sqrt{d}\|K_\mathbf{z}\|_2)$ for fixed $N_A=2$, $t = 6$, and $N_B = 32,64,128,256,512,1024$ (light-dark). Each histogram contains $2^{16}$ random samples from the Kraus ensemble. The dashed curves represent theoretical distributions $p_{\hat\sigma^2}(r)$ demonstrating excellent fits. (b) Variance of $\{R_\mathbf{z}\}$ in log scale, with $N_B = 6, 10, 18, 32, 64, 128, 256, 512, 1024$ (light-dark). The light brown surface outlines theoretical prediction $\mathrm{Var}[R]=\sigma^2_\mathrm{th}(N_B, t) + \frac{1}{4}\psi_1(2^{2N_A})$, where $\sigma^2_\mathrm{th} = \lambda N_B / \gamma^t$ is fitted from the best fits $\hat\sigma^2(N_B, t)$. Here for RUC, $\lambda\approx0.0858$ and $\gamma\approx 1.249$. The dashed curves represent slices of that surface with fixed $N_B$ in accordance with the plotted data.}
            \label{fig:sigma2_fit}
        \end{figure}

        Our task of extracting $\hat\sigma^2(N_B,t)$ can therefore be cast as a maximum likelihood estimation (MLE) procedure on the rescaled log norms $R$. Given an empirical Kraus ensemble $\mathcal{K}=\{K_\mathbf{z}\}_{\mathbf{z}=1}^{d_B}$, we compute the discrete set $\{R_\mathbf{z}\equiv\mathrm{log}(\sqrt{d}\|K_\mathbf{z}\|_2)\}_{\mathbf{z}=1}^{d_B}$ of rescaled log norms with PDF $p_\mathcal{K}(r)=\frac{1}{d_B}\sum_\mathbf{z}\delta(r - R_\mathbf{z})$. To determine the optimal fit of $\sigma^2$, we minimize the empirical cross-entropy of the theoretical distribution $p_{\sigma^2}$ under the empirical distribution $p_\mathcal{K}$, namely        
        \begin{equation}
        \begin{aligned}
            H\left(p_\mathcal{K}, p_{\sigma^2}\right) &\equiv -\int\mathrm{d}r\,p_\mathcal{K}(r)\log{p_{\sigma^2}(r)}
            \\
            &= -\frac{1}{d_B}\sum_\mathbf{z}\log p_{\sigma^2}(R_\mathbf{z}),
        \end{aligned}
        \end{equation}
        which reduces to maximizing the log-likelihood by
        \begin{equation}
            \hat\sigma^2 = \arg\max_{\sigma^2\ge 0}\sum_\mathbf{z}\log(p_{\sigma^2}(R_\mathbf{z})),
        \end{equation}
        standard in MLE. This allows us to extract $\hat\sigma^2(N_B, t)$ for any $N_B$ and $t$ from only the norm data $\|K_\mathbf{z}\|_2$, shown in Fig.~\ref{fig:sigma2_fit}(a) for the brickwork RUC. Here, for 1D circuits, we are numerically prohibited from accessing both large bath sizes and long times. Therefore, we either simulate the full Kraus ensemble for $N_B$ up to $20$ and $t$ up to $100$, or perform a uniform Monte Carlo sampling~\footnote{Numerically, we randomly sample a bitstring and calculate the action of a long sequence of transfer matrices on the temporal state $|\Omega\rangle$ prescribed by that bitstring. For numerical stability (especially in the large $N_B$ scenario where the norms of the Kraus operators span a wide range of orders of magnitude), we store the norm $\exp(\omega_\mathbf{z})$ of the temporal vector in log scale separately from its orientation $|\Omega_\mathbf{z}\rangle$, and progressively update them for every action of the transfer matrix.} from the Kraus ensemble for $N_B$ up to $1024$ and $t$ up to $20$.

        Next, we analyze the spacetime scalings that the data $\{\hat\sigma^2(N_B,t)\}$ uncover. To achieve this, we fit them against our theoretical functional form $\sigma^2_\mathrm{th}(N_B, t) = \lambda N_B / \gamma^t$ (as predicted in Sec.~\ref{sec:origin-ansatz}) on the available spacetime points $(N_B, t)$, obtaining $\lambda\approx 0.0858$ and $\gamma\approx1.249$ for brickwork qubit RUC. To visualize the quality of this fit, consider comparing an ideal signal with the empirical signal as follows: the ideal variance of the rescaled log norm $R$ is (see Appendix~\ref{app:ansatz-more-details})
        \begin{equation}
        \label{eq:ideal-var-R}
            \mathrm{Var}[R] = \sigma^2 + \frac{1}{4}\psi_1(d_A^2),
        \end{equation}
        where $\psi_1$ is the trigamma function. In Fig.~\ref{fig:sigma2_fit}(b), we show the empirical variance of $\{R_\mathbf{z}\}$ as scattered points, and overlay with the theoretical surface by substituting $\sigma^2_\mathrm{th}(N_B, t)$ into Eq.~\eqref{eq:ideal-var-R} with the parameters $(\lambda, \gamma)$ extracted from the previous fit for brickwork RUC. The scattered points show excellent agreement with the fitted surface, even though we are numerically restricted to accessing either large bath size or long time, but not both.

        At this point, we have extracted the empirical estimates $\hat\sigma^2(N_B, t)$ for brickwork qubit RUC and confirmed their consistency with the functional form $\sigma^2_\mathrm{th}(N_B, t) = \lambda N_B / \gamma^t$. 
        The same numerical procedure for an off-DU KIM~\footnote{We numerically study the off-DU KIM by setting $J = 1.000$, $h=0.4488$ away from DU point, with $g_i \sim \mathrm{Uniform}\left[\frac{\pi}{7}-\frac{\pi}{12}, \frac{\pi}{7}+\frac{\pi}{12}\right]$ sampled uniformly randomly and independently for each site.} is presented in Appendix~\ref{app:extra-numerics}, where we confirm the identical scaling form of $\sigma_\mathrm{th}^2$, with distinct numerical result for $\lambda\approx0.0364$, $\gamma\approx1.388$. These are system-specific, non-universal parameters.

        \subsubsection{Distance to running target}
        \label{sec:running-target}

        Having identified the spatiotemporal dependence of $\sigma^2$, we turn to specifying how well the running theoretical distribution $\mathbf{GinUE}\times\mathrm{LogNormal}_{\sigma^2}$ (the \textit{target} ensemble) captures the empirical Kraus ensemble statistically. As outlined in Sec.~\ref{sec:overview-setup}, we compare the empirical ensemble with the theoretical ensemble moment-by-moment.

        Recall that the $k$-th moment of the Ansatz is given by Eq.~\eqref{eq:kth-moment-family}. Thus, we aim to show that (here, we drop all explicit $\sigma^2$ dependence associated with subscript ``th'')
        \begin{equation}
            d^k M_\mathrm{emp}^{(k)}\rightarrow M_\mathrm{th}^{(k)} \equiv \underbrace{\exp(2k(k-1)\sigma^2)}_{x_\mathrm{th}^{(k)}}M_\mathrm{Gin}^{(k)},
        \end{equation}
        along any scaling contour $N_B, t\rightarrow\infty$ with $\sigma^2 = \lambda N_B/\gamma^t$ fixed, with the distance $\Delta^{(k)}_\mathrm{th}$ (defined in Eq.~\eqref{eq:distance-th-conj-1}) decaying exponentially in $t$.
        Note that $M_\mathrm{th}^{(k)}$ is always proportional to $M_\mathrm{Gin}^{(k)}$, thus it is helpful to consider the following unique decomposition for the empirical moment
        \begin{equation}
        \label{eq:signal-noise-decomposition}
            d^kM_\mathrm{emp}^{(k)} = x_\mathrm{emp}^{(k)}\left(M^{(k)}_\mathrm{Gin} + M_\perp^{(k)}\right),
        \end{equation}
        where we decompose the empirical moment into a part aligned with $M_\mathrm{Gin}^{(k)}$, and the rest $M_\perp^{(k)}$ orthogonal to $M_\mathrm{Gin}^{(k)}$ in the Hilbert-Schmidt sense: $\mathrm{Tr}(M_\perp^{(k)}M_\mathrm{Gin}^{(k)}) = 0$. This allows us to treat the sources of error in $x_\mathrm{emp}^{(k)}$ and in $M_\perp^{(k)}$ separately.
        Specifically, the coefficient can be found by
        \begin{equation}
            x_\mathrm{emp}^{(k)} = \frac{d^k\mathrm{Tr}\left(M_\mathrm{emp}^{(k)}\right)}{\mathrm{Tr}\left(M_\mathrm{Gin}^{(k)}\right)} = \frac{\mathcal{P}^{(k)}_\mathrm{emp}}{(d_A^2)_k},
        \end{equation}
        where we define the \textit{generalized purity} of the Kraus ensemble as
        \begin{equation}
            \mathcal{P}^{(k)}_\mathrm{emp} := \frac{d^k}{d_B}\sum_\mathbf{z}(K_\mathbf{z}|K_\mathbf{z})^k = d^k\mathrm{Tr}\left( M_\mathrm{emp}^{(k)} \right),
        \end{equation}
        directly related to the moments of the norms $\|K_\mathbf{z}\|_2$.

        Next, to capture the error $\Delta_\mathrm{th}^{(k)}$ with respect to the target ensemble, we introduce two helpful quantities.
        First, we define the \textit{running distance}:
        \begin{equation}
        \label{eq:running-distance}
            \Delta_\mathrm{run}^{(k)} := \frac{\left\| d^k M_\mathrm{emp}^{(k)} - x_\mathrm{emp}^{(k)}M_\mathrm{Gin}^{(k)} \right\|_2}{\left\| x_\mathrm{emp}^{(k)}M_\mathrm{Gin}^{(k)} \right\|_2}
            = \frac{\left\|M_\perp^{(k)} \right\|_2}{\left\| M_\mathrm{Gin}^{(k)}\right\|_2},
        \end{equation}
        as a measure of closeness in terms of alignment in the direction of $M_\mathrm{Gin}^{(k)}$. Second, we define the residue error in generalized purity compared to the theoretical distribution at $\sigma_\mathrm{th}^2(N_B, t)=\lambda N_B/\gamma^t$ extracted previously:
        \begin{equation}
            \delta^{(k)}:= \frac{\left| \mathcal{P}_\mathrm{emp}^{(k)} - \mathcal{P}_\mathrm{th}^{(k)}\right|}{\mathcal{P}_\mathrm{th}^{(k)}}=\frac{\left|x_\mathrm{emp}^{(k)} - x_\mathrm{th}^{(k)}\right|}{x_\mathrm{th}^{(k)}},
        \end{equation}
        where $\mathcal{P}_\mathrm{th}^{(k)} = \exp(2k(k-1)\sigma^2_\mathrm{th}) (d_A^2)_k$ is the theoretical generalized purity set by the ideal spacetime scaling. $\delta^{(k)}$ captures the error in the $k$-th moment of the norm of Kraus compared with our Ansatz at $\sigma_\mathrm{th}^{2}(N_B,t)$.

        \begin{figure*}[!t]
            \centering
            \includegraphics[width=0.85\linewidth]{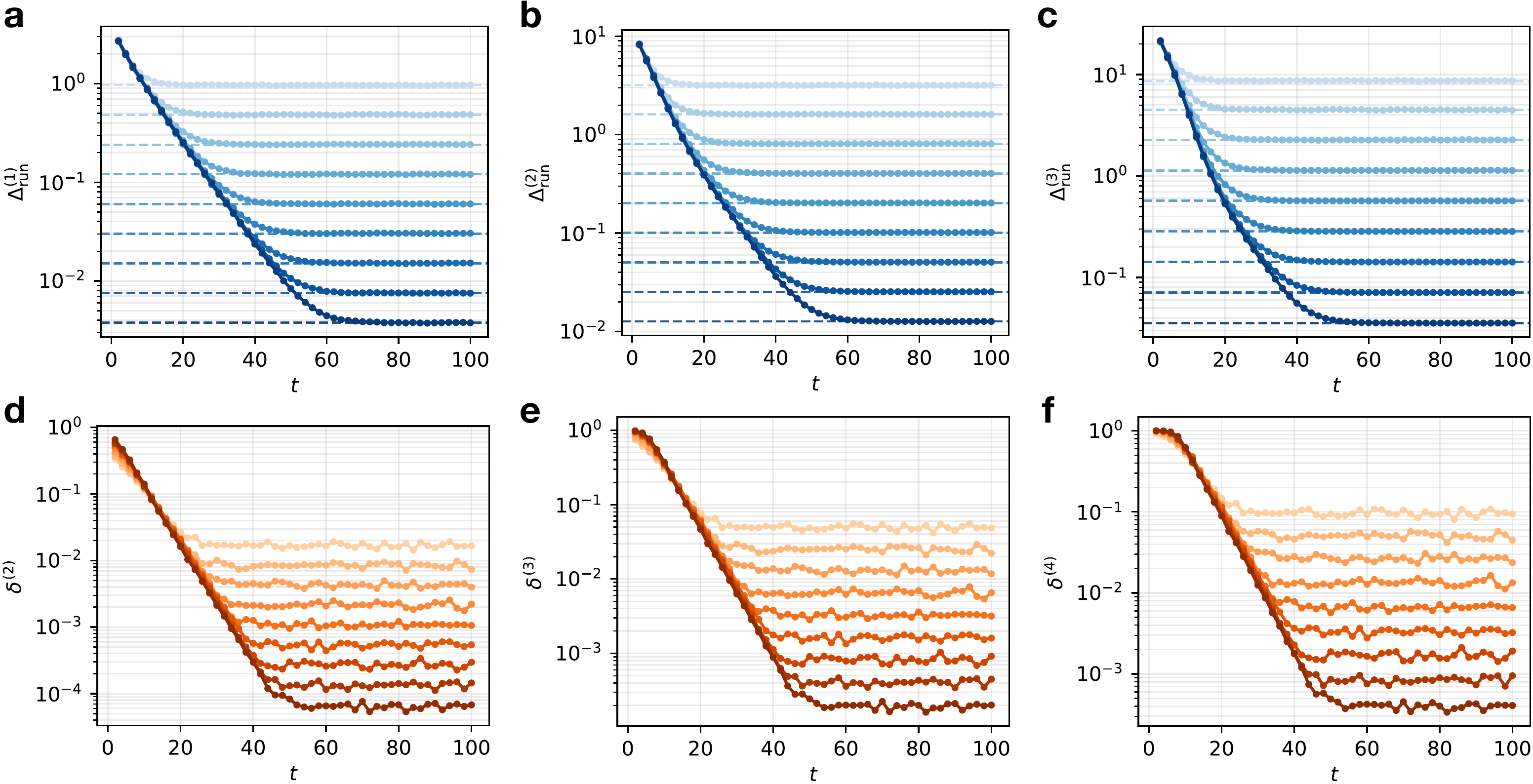}
            \caption{Numerical tests for RUC. Here we fix $N_A=2$ and show $N_B=4,6,8,10,12,14,16,18,20$ (light-dark). (a-c) Running distance $\Delta^{(k)}_\mathrm{run}$ shown for $k=1,2,3$. Each data point is averaged over 100 circuit realizations. These curves collapse in early times and follow an envelope showing a multi-rate exponential decay in $t$ before the curves plateau at a turning time $t_c$ scaling linearly in $N_B$.
            We attribute the multi-rate decay to early-time transient of local circuits, and assume that the envelope crosses over to a simple exponential decay in $t$ at later times. The dashed lines represent the averaged distance $\Delta_\mathrm{Gin}^{(k)}$ for global Haar random unitary (also over $100$ realizations) with corresponding bath size $N_B$. (d-f) Purity relative error  $\delta^{(k)}$ compared to theoretical value, shown for $k=2,3,4$ (as $k=1$ is trivial by normalization constraint of the Kraus ensemble). Each data point is averaged over 100 circuit realizations. The envelope of the curves also follows exponential decay in $t$ at sufficiently large $t$. 
            }
            \label{fig:dist-purity}
        \end{figure*}
        
        We claim that $\Delta_\mathrm{run}^{(k)}$ and $\mathcal{P}_\mathrm{emp}^{(k)}$ (or equivalently $x_\mathrm{emp}^{(k)}$) capture the complete error profile of $\Delta_\mathrm{th}^{(k)}$. In particular, the distance $\Delta_\mathrm{th}^{(k)}$ can be bounded as

        \begin{equation}
        \label{eq:distance-bound}
        \begin{aligned}
            \Delta_{\mathrm{th}}^{(k)} &= 
            \frac{\left\| d^k M_\mathrm{emp}^{(k)} - M_{\mathrm{th}}^{(k)} \right\|_2}{\left\| M_{\mathrm{th}}^{(k)} \right\|_2}
            \\
            &= \frac{\left \| (x_\mathrm{emp}^{(k)} - x_\mathrm{th}^{(k)}) M_\mathrm{Gin}^{(k)} + x_\mathrm{emp}^{(k)} M_\perp^{(k)} \right\|_2}{\left\| x_\mathrm{th}^{(k)}M_\mathrm{Gin}^{(k)}\right\|_2}
            \\
            &\le \sqrt{\delta^{(k)\,2} + (1+\delta^{(k)})^2\Delta_\mathrm{run}^{(k)\,2 
            }},
        \end{aligned}
        \end{equation}
        which also holds for $\Delta_\mathrm{Gin}^{(k)}$ when we set $x_\mathrm{th}^{(k)}=1$.
        This allows us to bound the distance $\Delta_\mathrm{th}^{(k)}$ through separate sources of error. With Eq.~\eqref{eq:distance-bound}, it suffices to show exponential decay of $\Delta_\mathrm{run}^{(k)}$ and $\delta^{(k)}$ separately with respect to $t$ in the joint scaling limit.

        We show the numerical results of $\Delta_\mathrm{run}^{(k)}$ and $\delta^{(k)}$ for brickwork RUC in Fig.~\ref{fig:dist-purity}, where we plot the errors separately for different $N_B$---for a given color within a panel, $N_B$ is fixed---and the late-time plateau values within each panel decay exponentially with $N_B$. However, this does not yet reveal the decay behavior of errors along a particular scaling contour $N_B, t\rightarrow\infty$ with $t=\log_\gamma N_B + c$.

        To understand errors under joint scaling, first observe that for any $k$, both $\Delta_\mathrm{run}^{(k)}$ and $\delta^{(k)}$ show excellent collapse of curves along an envelope of exponential decay in $t$, until a time $t_c\sim\mathcal{O}(N_B)$ when the curve turns to a plateau set by the finite-size effect of $N_B$. This is in sharp contrast to Fig.~\ref{fig:dist-Gin} where a fixed Ginibre target fails to capture the universal behavior of Kraus ensembles.

        If one follows the joint limit of $N_B, t\rightarrow\infty$ with $\sigma^2=\lambda N_B/\gamma^t$ held fixed, circuit depth $t$ scales like $\mathcal{O}\left(\log{N_B}\right)$---much slower than the turning time $t_c$ scaling linearly as $\mathcal{O}(N_B)$. Therefore, we extrapolate and argue that in this joint scaling limit, both $\Delta_\mathrm{run}^{(k)}$ and $\delta^{(k)}$ follow the envelopes without ever hitting a finite-size plateau, and therefore decay exponentially in $t$ indefinitely. Hence, by Eq.~\eqref{eq:distance-bound}, the distance $\Delta_\mathrm{th}^{(k)}$ also decays exponentially in time indefinitely, and we arrive at the convergence behavior stated in Conjecture~\ref{conj:1-precise}.
        
        To recapitulate, we have argued for the exponential decay of $\Delta_\mathrm{th}^{(k)}$ along the critical family of joint scalings as $t=\log_\gamma N_B + c$. In Appendix~\ref{app:extra-numerics}, we carry out the same numerical procedures for the off-DU KIM. We also study the spectral density of the fixed-trace Kraus ensemble, i.e., the spectrum of $K^\dagger K/\mathrm{Tr}(K^\dagger K)$. This is to isolate and probe only the Ginibre part of the Ansatz, and we show that the empirical fixed-trace Kraus ensemble follows the spectral properties of the fixed-trace Wishart--Laguerre ensemble, consistent with the Ansatz. 

\section{Physical consequences of the Ansatz}
\label{sec:consequence-ansatz}

    Having stated and substantiated our conjecture of the universal form of Kraus ensembles emergent in generic 1D quantum circuits, we now move to discussing its physical consequences~\footnote{In this section, motivated by Conjecture~\ref{conj:1-informal}, we make a stronger physical assumption that, in the limiting universal regime, averages over Kraus operators may be replaced by integrals with respect to the full Ansatz measure.}. We derive implications for: (i) deep thermalization, 
    the emergence of maximally random {\it quantum state ensembles} in quantum many-body dynamics from measurements, 
    and (ii) quantum information recovery, where we treat the measurement outcomes $\mathbf{z}$ as  classical side information which can be used as a resource to retrieve initially local quantum information that is scrambled into the global system.

    \subsection{Deep thermalization}
    \label{sec:consequence-ansatz-deep-thermalization}
    Here, we provide a brief recap of \textit{deep thermalization} and argue that the universal behavior of Kraus ensembles provides a microscopic mechanism for deep thermalization to Haar ensembles in generic 1D circuits,
    shown schematically in Fig.~\ref{fig:consequences}(a). Readers familiar with the background may skip the following recap.

    \paragraph{Recap}
    Consider a bipartite many-body state $|\Psi_{AB}\rangle$ with subsystem dimensions $d_A$ and $d_B$, respectively. Projective measurements on $B$ in the computational basis $\{|\mathbf{z}\rangle_B\}$ generate a probabilistic ensemble of pure states on $A$ called the \textit{projected ensemble} (PE):
    \begin{equation}
        \mathcal{E}_A = \{p_\mathbf{z}, |\psi_\mathbf{z}\rangle_A\},
    \end{equation}
    where $p_\mathbf{z} = \left\|\left(\mathbbm{1}_A\otimes\langle\mathbf{z}|_B\right)|\Psi_{AB}\rangle\right\|^2$ is the Born probability of finding $B$ in the state $|\mathbf{z}\rangle_B$, and $|\psi_\mathbf{z}\rangle_A = \left(\mathbbm{1}_A\otimes\langle\mathbf{z}|_B\right)|\Psi_{AB}\rangle/\sqrt{p_\mathbf{z}}$ is the state on $A$ conditioned on the state of $B$ being $|\mathbf{z}\rangle_B$.
    One can understand the PE as a particular physically motivated \textit{unraveling} of the reduced density matrix $\rho_A \equiv \mathrm{Tr}_B|\Psi_{AB}\rangle\langle\Psi_{AB}|$, which is obtained by mixing the PE and discarding the outcome~$\mathbf{z}$: $\rho_A = \sum_\mathbf{z}p_\mathbf{z}|\psi_\mathbf{z}\rangle\langle\psi_\mathbf{z}|$. 

    Suppose now $|\Psi_{AB}\rangle$ arises from quantum many-body dynamics. 
    As discussed in Sec.~\ref{sec:overview-setup}, the local dynamics of subsystem $A$ interacting with its bath $B$ can be described by a quantum channel $\mathcal{N}^{A\rightarrow A}_t$. Under generic chaotic dynamics of a small subsystem $A$ interacting with a large bath $B$, the reduced density matrix $\rho_{A,t} = \mathcal{N}^{A\rightarrow A}_t(\rho_{A,0})$ is expected to approach a Gibbs state at late times---a process known as (regular) {\it quantum thermalization}. 
    Microscopically, this can be understood from the production of entanglement between $A$ and $B$, such that initial quantum information in $A$ gets scrambled into global degrees of freedom and becomes locally irretrievable~\cite{deutsch_quantum_1991, srednicki_chaos_1994, Rigol_thermalization_2008, Dalessio_from_2016}. 
    
    \textit{Deep} thermalization refers instead to a deeper form of equilibration, captured by the PE on $A$ approaching a maximally entropic universal distribution subject to physical constraints. The most generic universal state ensemble is termed the (generalized) \textit{Scrooge} ensemble and is related to a generalized maximal entropy principle on an information-theoretic level, previously studied in~\cite{mark_maximum_2024}. The (generalized) Scrooge ensembles have been investigated in spin systems~\cite{chang_deep_2025, liu_coherence-induced_2025, Feng_resource_2026, mcginley_scrooge_2025, mok_nature_2026}, as well as fermionic and bosonic systems~\cite{lucas_fermions_2023, liu_cv_2024, bejan_matchgate_2025}. Under special cases like in systems without explicit conservation laws, the generalized Scrooge ensemble reduces to the Haar ensemble $\mathcal{E}_H$.

    Previously, analytical results demonstrating the emergence of Haar ensembles in deep thermalization have been limited to very special states $|\Psi_{AB}\rangle$~\cite{cotler_emergent_2023, Ghosh_design_2025, ho_exact_2022, Claeys_emergent_2022, ippoliti_dynamical_2023}. This naturally raises the question we address here: how does Kraus universality imply deep thermalization to Haar ensembles for an arbitrary pure input state on $A$?
   
    \paragraph{Implication of Kraus universality} Given an empirical Kraus ensemble $\{K_\mathbf{z}\}$ and an {\it arbitrary} pure input state $|\psi_0\rangle$ on $A$, the unnormalized PE---a collection of unnormalized states assuming equal weights among the ensemble---can be written as
    \begin{equation}
        \widetilde{\mathcal{E}}_A = \{ |\widetilde\psi_\mathbf{z}\rangle\},\quad\text{where }|\widetilde\psi_\mathbf{z}\rangle=K_\mathbf{z}|\psi_0\rangle.
    \end{equation}
    Here, each element contains information of its Born probability $p_\mathbf{z} = \langle\widetilde\psi_\mathbf{z}|\widetilde\psi_\mathbf{z}\rangle$ and its corresponding normalized state $|\psi_\mathbf{z}\rangle = |\widetilde\psi_\mathbf{z}\rangle/\sqrt{p_\mathbf{z}}$. Hence, information of the normalized PE is fully derivable from unnormalized PE.

    To apply the universal form of Kraus ensembles, note that our Ansatz enjoys both left and right unitary invariance, i.e., for $\sqrt{d}K\sim\mathbf{GinUE}\times\mathrm{LogNormal}_{\sigma^2}$ and any unitary $U_L$ and $U_R$, there is
    \begin{equation}
        K\overset{d}{=}U_L K U_R^\dagger.
    \end{equation}
    Intuitively, right unitary invariance implies that the unnormalized PE is independent of the input state $|\psi_0\rangle$, as expected for a small subsystem thermalized by a big bath. Furthermore, left unitary invariance implies that the unnormalized PE is rotationally invariant and, upon normalization, produces the Haar ensemble $\mathcal{E}_H$.

    Therefore, under the limiting form of the Kraus ensemble specified by the conjecture and independently of the input state $|\psi_0\rangle$, the unnormalized PE follows:
    \begin{equation}
    \label{eq:ansatz-unnormalized-states}
        \sqrt{d}|\widetilde\psi\rangle\sim\mathbf{Gauss}_{d_A}\times\mathrm{LogNormal}_{\sigma^2},
    \end{equation}
    where $\mathbf{Gauss}_{d_A}$ represents a $d_A$-dimensional standard complex Gaussian distribution. This is nothing but a column of the Ginibre matrix, and the log-normal part provides additional fluctuations in the Born probabilities.
    
    It thus follows that the $k$-th moment of the PE is
    \begin{equation}
    \begin{aligned}
        \rho^{(k)} = \sum_\mathbf{z}\frac{|\widetilde\psi_\mathbf{z}\rangle\langle\widetilde\psi_\mathbf{z}|^{\otimes k}}{\langle\widetilde\psi_\mathbf{z}|\widetilde\psi_\mathbf{z}\rangle^{k-1}}\approx d_B\underset{\widetilde\psi}{\mathbb{E}}\left[ \frac{|\widetilde\psi\rangle\langle\widetilde\psi|^{\otimes k}}{\langle\widetilde\psi|\widetilde\psi\rangle^{k-1}}\right] = \rho^{(k)}_H,
    \end{aligned}
    \end{equation}
    precisely the phenomenon of deep thermalization to a Haar ensemble. Here, the expectation ${\mathbb{E}}$ is over the Ansatz in Eq.~\eqref{eq:ansatz-unnormalized-states} for unnormalized states (see Appendix~\ref{app:ansatz-more-details} for details). 
    Notably, this result is independent of $\sigma^2$ while the Born probabilities of the PE retain the nontrivial $\sigma^2$-dependent log-normal statistics from the Ansatz, indicating there is physics still characterizing quantum many-body dynamics even after the PE on a local subsystem has deeply thermalized. 
    
    \subsection{Quantum information recovery}
    \label{sec:consequence-ansatz-qi-recovery}
        While the local chaotic dynamics $\mathcal{N}$ thermalizes any input state to a maximally mixed state, the projected Kraus ensemble $\mathcal{K}$ provides a process-level unraveling of the dynamics. Here, we show that the universal Ansatz of Kraus maps underpins the possibility of recovering quantum information locally in chaotic dynamics with the assistance of classical side information.

        Conventionally, quantum information stored within a quantum ``memory'' is contaminated through uncontrolled interactions with an environment that decoheres the system. In the quantum error correction setup, quantum information is encoded in logical states, which are embedded in a larger physical space with redundancies to protect the quantum information from noise. In order for the intended information to remain recoverable, one needs to make sure that the environment stays decoupled from the encoded logical information~\cite{Nielsen_Chuang_2010, Schumacher_quantum_1996}.
        
        The dynamics of such a quantum ``memory'' can be modeled by the local subsystem dynamics discussed in this work. Suppose subsystem $A$ (the ``memory'') hosts the quantum information we wish to protect, and its bath $B$ (the ``environment'') interacts with $A$ and may corrupt the information stored. It is thus natural to ask: how can we restore the initial state of $A$ by acting on the final state of $A$? What resources does this task require?
        
        Let the logical information be represented by a reference system $R$ maximally entangled with a subspace of $A$, as illustrated in Fig.~\ref{fig:recovery}. Specifically, the maximally entangled state prior to any dynamics can be written as
        \begin{equation}
            |\Phi^{RA}\rangle = \frac{1}{\sqrt{d_R}}\sum_{\alpha=1}^{d_R} |\alpha\rangle_R\otimes |\alpha\rangle_A,
        \end{equation}
        where $d_R$ is the Hilbert space dimension of the reference system $R$, and $\{|\alpha\rangle\}_{\alpha=1}^{d_R}$ is an orthonormal basis of $R$, with $d_R\le d_A$. Correspondingly, $\{|\alpha\rangle_A\}$ spans only a $d_R$-dimensional subspace of $A$, i.e., the code space. For convenience, we denote the size of subsystems by their number of qubits. The number of logical qubits $N_R$ is no larger than the number of physical qubits $N_A$ in $A$, and we refer to $n=(N_A - N_R)$ as the qubit overhead of the encoding. Through interaction with a thermodynamically large bath ($N_B\gg 1$), the logical information in $R$ is scrambled across $A\cup B$ and can no longer be retrieved locally from $A$, as illustrated in Fig.~\ref{fig:recovery}(a). In order to recover the original quantum information in $R$, a decoder requires additional resources.

        In the seminal work by Hayden and Preskill~\cite{hayden_mirrors_2007}, they considered the case where the decoder is supplemented with a maximally entangled reference $Q$ to the initial bath $B$, and is allowed to perform joint operations on $Q$ and $A$ after a global Haar random unitary on $A$ and $B$. They concluded that quantum information on $R$ can be recovered provided that there is a constant number of qubit overhead. Their setting is illustrated in Fig.~\ref{fig:recovery}(b), where $Q$ represents a quantum register holding the quantum side information in assistance of the decoder. Their protocol falls into the category of environment-assisted quantum error correction, studied previously in~\cite{Gregoratti_quantum_2003, Hayden_correcting_2005, Winter_environment_2007}.

        \begin{figure}[!t]
            \centering
            \includegraphics[width=\linewidth]{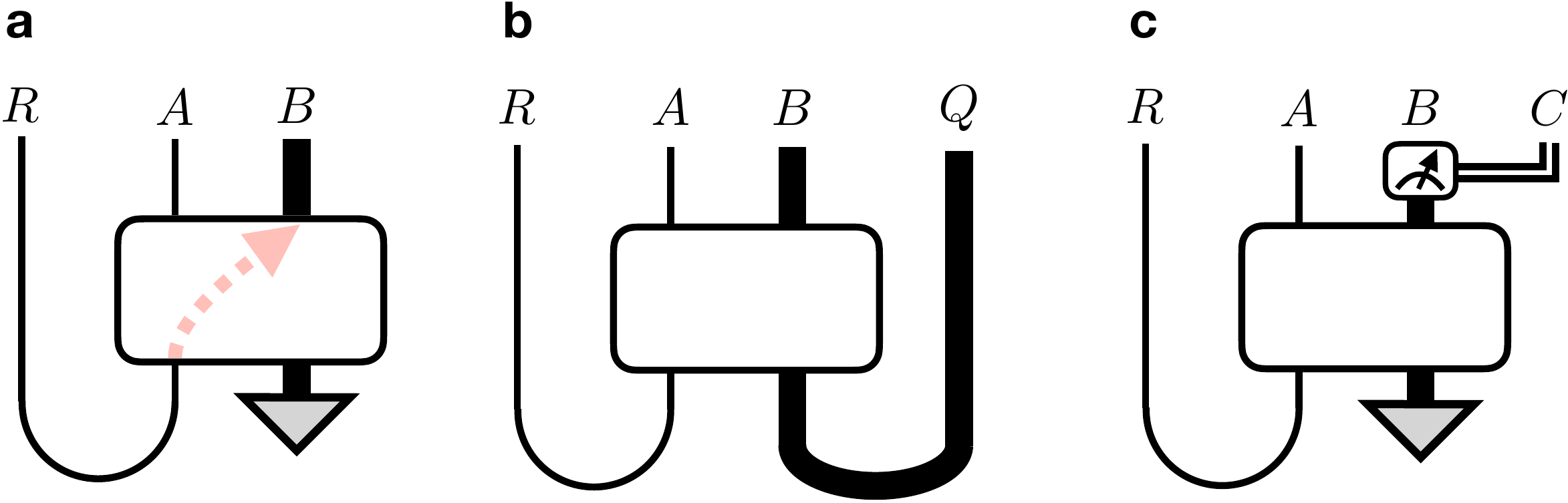}
            \caption{Schematics of quantum information recovery setups. (a) No side information supplemented as additional resource, therefore quantum information on $R$ is leaked into $B$ and becomes locally irretrievable on $A$. (b) With quantum side information in the form of a maximally entangled reference $Q$ to the initial state of the bath~\cite{hayden_mirrors_2007}. Quantum information in $R$ is recoverable in $A$ and $Q$ with a constant qubit overhead $n$. (c) With classical side information in the form of measurement outcomes stored in $C$. Quantum information in $R$ is recoverable in $A$ and $C$ with a constant qubit overhead $n$.}
            \label{fig:recovery}
        \end{figure}

        Here, we propose a modified setup where the decoder no longer has access to the quantum side information of the bath. Assume once again that the bath is initialized in a fixed pure state. Instead of discarding $B$, suppose we measure $B$ in the computational basis and keep the measurement outcome in a classical register $C$, and we allow the decoder to act jointly on $C$ and $A$ to retrieve quantum information on $R$. Our setup is illustrated in Fig.~\ref{fig:recovery}(c), where $C$ represents a classical register holding the classical side information in assistance of the decoder.

        Physically, the projected Kraus ensemble provides a description of the local dynamics on $A$ conditioned on the measurement outcome of $B$. These are the characteristics of a \textit{quantum instrument}, described by
        \begin{equation}
        \label{eq:q-instrument}
            \mathcal{M}^{A\rightarrow AC}(\cdot) = \sum_\mathbf{z}K_\mathbf{z}(\cdot)K_\mathbf{z}^\dagger\otimes|\mathbf{z}\rangle\langle\mathbf{z}|_C,
        \end{equation}
        which takes the form of a \textit{flagged channel}, where $C$ stands for the classical register.  
        Unlike in Eqs.~\eqref{eq:q-channel} and~\eqref{eq:q-channel-kraus-decomposition} where the bath $B$ is traced out and inaccessible to a future observer, here, the local dynamics on $A$ is flagged by a classical register holding the $\mathbf{z}$ information. As we will show, this additional classical side information allows for recoverability of quantum information that is otherwise lost in the local thermalizing channel $\mathcal{N}^{A\rightarrow A}$.

        To derive this, consider that the flagged channel has a Choi state representation:
        \begin{equation}
        \begin{aligned}
            |\Psi^{RABC}\rangle 
            &= \frac{1}{\sqrt{d_R}}\sum_{\alpha,\mathbf{z}}\underbrace{|\alpha\rangle_R}_{\text{logical}}\otimes \underbrace{K_\mathbf{z}|\alpha\rangle_A}_{\text{physical}}\otimes\underbrace{|\mathbf{z}\rangle_C}_{\text{side}}\otimes|\mathbf{z}\rangle_B,
        \end{aligned}
        \end{equation}
        where we note that $B$ and $C$ are symmetric by virtue of the projective measurement. To capture quantum information recoverability in this setting, we consider the coherent information from the reference $R$ to the systems $AC$ that are accessible to the decoder:
        \begin{equation}
        \label{eq:coherent-info}
        \begin{aligned}
            I_\mathrm{coh}(R\rangle AC)&:= S(AC) - S(RAC)
            \\
            &= S(RB) - S(B) = S(R|B)
            \\
            &= \sum_\mathbf{z}p(\mathbf{z})S(\rho_\mathbf{z}^{R})\le\log(d_R),
        \end{aligned}
        \end{equation}
        where we make use of the fact that the Choi state is a pure state on the second line. To explain the third line, we first define the unnormalized conditional Choi state:
        \begin{equation}
            |\widetilde{\Psi}_\mathbf{z}^{RA}\rangle = \frac{1}{\sqrt{d_R}}\sum_\alpha|\alpha\rangle_R\otimes K_\mathbf{z}|\alpha\rangle_A,
        \end{equation}
        with probability $p(\mathbf{z})=\langle\widetilde{\Psi}_\mathbf{z}^{RA}|\widetilde{\Psi}_\mathbf{z}^{RA}\rangle$, and
        \begin{equation}
            \rho_\mathbf{z}^{R} = \mathrm{Tr}_A|\Psi_{\mathbf{z}}^{RA}\rangle\langle\Psi_{\mathbf{z}}^{RA}|,
        \end{equation}
        where $|\Psi_\mathbf{z}^{RA}\rangle\equiv |\widetilde{\Psi}_\mathbf{z}^{RA}\rangle/\sqrt{p(\mathbf{z})}$ is the normalized conditional Choi state. Therefore, the coherent information is nothing but the averaged entanglement entropy between $R$ and $A$ in the conditional Choi states. In~\cite{Schumacher_quantum_1996}, it is proved that there exists a perfect quantum information recovery protocol if and only if $I_\mathrm{coh}$ is maximal, in this case $\log(d_R)$. Later in~\cite{Schumacher_approximate_2001} it is shown that an approximate quantum error correction protocol exists if $I_\mathrm{coh} = \log(d_R) - \epsilon$ for small $\epsilon > 0$. More precisely, there exists a recovery protocol that acts only on the system accessible to the decoder (in this case $AC$) and outputs a state $\sigma^{RA}$ such that the fidelity $F(\sigma^{RA}, |\Phi^{RA}\rangle)\ge 1 - \sqrt{\epsilon}$. In our case specifically, we show that a recovery protocol that acts on $A$ with a unitary $U_\mathbf{z}^\dagger$ conditioned on the measurement outcome $\mathbf{z}$ suffices to achieve the recovery fidelity, as illustrated in Fig.~\ref{fig:consequences}(b). Here, $U_\mathbf{z}$ is the unitary arising from polar decomposition $K_\mathbf{z}\mathbb{P}^{(R)} = U_\mathbf{z}\sqrt{\mathbb{P}^{(R)}K_\mathbf{z}^\dagger K_\mathbf{z}\mathbb{P}^{(R)}}$, where $\mathbb{P}^{(R)} = \sum_\alpha |\alpha\rangle\langle\alpha|_A$ is the projector on $A$ to the code space, and application of $U_\mathbf{z}^\dagger$ can be understood as the action to counter the subprocess $K_\mathbf{z}$. We prove this in Appendix~\ref{app:quantum-information-recovery}.

        As a special case discussed in~\cite{cheng_emergent_2025}, for a mixed unitary channel with Kraus maps proportional to unitaries, i.e., $K_\mathbf{z}^\dagger K_\mathbf{z}\propto\mathbbm{1}_A$, one can show that the coherent information is maximal with $N_R = N_A$. This is because the Schmidt spectra of these Kraus operators are exactly flat, giving rise to maximal entanglement entropies in the conditional Choi states $|\Psi_\mathbf{z}^{RA}\rangle$ between $R$ and $A$. For such cases, initial quantum information is perfectly recoverable with the assistance of the classical side information. This is in alignment with the fact that Born probabilities $p_\mathbf{z}$ of the event $\mathbf{z}$ happening do not depend on the input state on $A$ that is fed into the system~\cite{cheng_emergent_2025}, and therefore the measurement on $B$ has no way of intercepting quantum information passed from the past to the future in $A$.

        For local dynamics induced by generic 1D circuit dynamics on $A\cup B$, consider substituting our universal Ansatz for the Kraus maps $K_\mathbf{z}$. This corresponds to taking the suitable joint limit $N_B, t\rightarrow\infty$, and we have
        \begin{equation}
        \label{eq:coherent-info-ansatz}
            I_\mathrm{coh}(R\rangle AC) \ge \log(d_R) - \frac{d_R}{2d_A} = \log(d_R)-2^{-n-1},
        \end{equation}
        recall that $n=(N_A - N_R)$ is the qubit overhead. 
        This is obtained by substituting the Ansatz into Eq.~\eqref{eq:coherent-info} and replacing the sum with an integral. Based on this calculation, we may then conclude that under the Ansatz, it takes only a constant qubit overhead $n=\mathcal{O}(\log({1}/{\epsilon}))$ for an arbitrary error threshold $\epsilon > 0$ such that $I_\mathrm{coh}(R\rangle AC)\ge \log(d_R)-\epsilon$ is satisfied. For a detailed discussion of the calculations and the qubit overhead, see Appendix~\ref{app:quantum-information-recovery}.

        This calculation suggests that even with a thermodynamically large bath, projective measurement on the bath reads out only a minuscule amount of quantum information in $A$. Importantly, Eq.~\eqref{eq:coherent-info-ansatz} contains no dependence of $\sigma^2$. Thus, the recoverability condition is independent of the additional log-normal fluctuations in the norm of the Kraus operators. Consequently, even for a thermodynamically large bath, a generic 1D logarithmic-depth circuit provides the same level of conditional recoverability as the global Haar case. We note here that a similar result concerning random Haar codes was discussed in~\cite{Lee_randomly_2026} based on a typical decoupling argument, while  our argument is based on Kraus universality and thus applicable for generic physical dynamics.

\section{Discussion and outlook}
\label{sec:discussion}

In this work, we have identified a new form of random-matrix universality in generic 1D local quantum circuit dynamics. The central object we study is the \textit{projected Kraus ensemble}, the collection of effective subprocesses $K_\mathbf{z}$ of the bath on a local subsystem resolved by classical information $\mathbf{z}$ from measurements,
which provides a finer fingerprint of chaotic quantum many-body dynamics.

Concretely, we proposed a one-parameter Ansatz consisting of two components: $\mathbf{GinUE}\times\mathrm{LogNormal}_{\sigma^2}$. The Ginibre part captures information scrambling within the subsystem due to its rotational invariance properties, while the log-normal part introduces another layer of fluctuations in the norms of the matrices, originating from the effective non-unitary dynamics of transfer matrices in the spatial direction. For generic 1D local quantum circuits, we identified and argued for a family of critical spacetime joint scaling limits $N_B, t\rightarrow\infty$ parametrized by fixed $\sigma^2 = \lambda N_B / \gamma^t$ for each contour, and we provided evidence that the empirical projected Kraus ensemble (up to proper rescaling) converges exponentially quickly in $t$ towards the Ansatz at the corresponding $\sigma^2$. Our Ansatz was also established in special cases of global Haar random unitary dynamics and solvable DU circuits.

Physically, we showed that this universality underlies deep thermalization in generic 1D circuit dynamics. In particular, the projected ensemble of states inherits the rotational invariance of the projected Kraus ensemble and hence approaches Haar randomness independently of the input state. It is worth re-emphasizing that locality-induced log-normal features in the Born probabilities may persist even after deep thermalization has occurred.

Finally, we exploited this universal structure to study conditional quantum information recovery within chaotic many-body dynamics. We proposed a variant of the Hayden--Preskill recovery protocol in which local quantum information can be recovered with the assistance of classical side information rather than quantum side information. Thus, the same measurement outcomes that label the projected Kraus ensemble also serve as a resource for retrieving quantum information leaked into the bath.

One may inquire how our predicted Kraus universality could be tested experimentally. The projected Kraus ensemble can be realized directly as an unnormalized projected ensemble of states by introducing a reference system $A'$ maximally entangled with $A$, evolving $AB$, and subsequently measuring $B$. Specifically, 
\begin{equation}
\label{eq:PKE-prepare}
    \frac{1}{\sqrt{d_A}}|K_\mathbf{z}) = \frac{1}{\sqrt{d_A}}\,\vcenter{\hbox{\includegraphics[height=1.25cm]{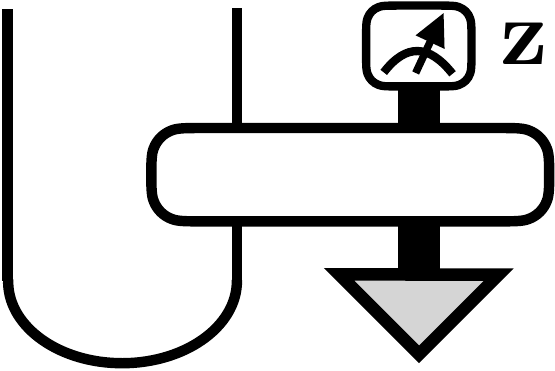}}}\,.
\end{equation}
where the generator state is $(\mathbbm{1}_{A'}\otimes U_{AB})|\Phi_{A'A}\rangle\otimes |B_0\rangle_B$, with $|\Phi_{A'A}\rangle=d_A^{-1/2}\sum_i |i\rangle_{A'}|i\rangle_A$ the maximally entangled state between the reference $A'$ and subsystem $A$. This construction translates the two components of Kraus universality into experimentally accessible statistics: the Born probabilities of the state ensemble encode the norms of Kraus operators, while the projected states encode the normalized Kraus structure. To avoid issues of post-selection, the ensemble statistics may be probed using established techniques based on quantum--classical correlators~\cite{McGinley_postselection-free_2024, Garratt_probing_2024}.

Several directions remain open for future studies. First, it would be useful to understand how higher-dimensional local circuits approach the late-time Ginibre regime. There, the spatial transfer matrix intuition used here should be modified, and the spatiotemporal scaling features would correspondingly change. For example, in generic 2D local circuits, $\mathcal{O}(1)$ critical timescale is expected due to teleportation phenomenology~\cite{McGinley_measurement-induced_2025, liu_conditional_2026}. Second, it would be interesting to study dynamics with conserved quantities, e.g., Hamiltonian dynamics. One could expect that, in such a setting, the Ginibre intuition should be modified to comply with the observation of Scrooge ensembles~\cite{mark_maximum_2024, mcginley_scrooge_2025, mok_nature_2026}. Intuitively, the Kraus ensemble could potentially approach a universal form $\sqrt{d}K\sim \sqrt{d_A\rho_A}\,\mathbf{GinUE}$ that remains right (but no longer left) unitary invariant, and is consistent with the Scrooge ensemble by construction. Additionally, the early time critical spatiotemporal scalings would be different from that of 1D circuit dynamics. Preliminary numerical results suggest scalings in the form of $\sigma^2\propto N_B / t^\alpha$ for a mixed-field Ising model Hamiltonian with $B$ initialized in a zero energy-density product state. It could also be interesting to study conditional dynamics with noise, similar to the setup in~\cite{yu_mixed_2025, Sherry_deep_2026}.

\acknowledgments
We thank Vir Bulchandani, John Chalker, Soonwon Choi, Daniel Mark, Max McGinley, and Tian-Hua Yang for insightful discussions.
Numerical experiments of this work were supported by the National Supercomputing Centre (NSCC) of Singapore. W.~W.~H.~is supported by the National Research Foundation (NRF), Singapore, through the NRF Fellowship NRF-NRFF15-2023-0008, and through the National Quantum Office, hosted in A*STAR, under its Centre for Quantum Technologies Funding Initiative (S24Q2d0009). The Institute for Quantum Information and Matter is an NSF Physics Frontiers Center (PHY-2317110).

\nocite{ODonovan_diagnosing_2026}
\bibliography{ref}

\newpage
\appendix
\onecolumngrid
\clearpage

\begin{center}
{\large\bfseries
Supplemental Material for: \\Emergent universality in Kraus maps of quantum chaotic many-body dynamics
\par}

\vspace{0.8em}

{\normalsize
Qi Camm Huang,$^{1,2}$
Wai-Keong Mok,$^{3}$
Tobias Haug,$^{4}$
and Wen Wei Ho$^{1,2}$
\par}

\vspace{0.4em}

{\small{\itshape
$^{1}$Department of Physics, National University of Singapore, Singapore 117551\\
$^{2}$Centre for Quantum Technologies, National University of Singapore, Singapore 117543\\
$^{3}$Institute for Quantum Information and Matter,\\
California Institute of Technology, Pasadena, CA 91125, USA\\
$^{4}$Quantum Research Center, Technology Innovation Institute, Abu Dhabi, UAE\\}
(Dated: \today)
\par}

\end{center}
\vspace{0.8em}

In this supplemental material, we provide details on the background information, derivation of the theorems and other statements, and extended discussions of this work. The Appendices are organized as follows: Appendix~\ref{app:review} reviews the basic tools used throughout the paper, including Weingarten calculus for Haar random unitaries, Wick calculus for Ginibre matrices, elementary properties of the log-normal distribution, and the interpretation of Kraus ensembles as unnormalized state ensembles. Appendix~\ref{app:global-haar-proof} contains the proof of Theorem~\ref{thm:1-precise}, establishing the emergence of Ginibre statistics from global Haar random unitary dynamics. Appendix~\ref{app:DU-KIM-proof} presents the tensor network formulation of the kicked Ising model and proves the dual-unitary (DU) result stated in Theorem~\ref{thm:2-precise}. Appendix~\ref{app:shrink} explains the argument of the functional dependence $\sigma^2\propto 1/\gamma^t$ appearing in the Ansatz. Appendix~\ref{app:ansatz-more-details} provides further details on the Ansatz, including norm statistics for fitting procedures, self-consistency checks, and implications on deep thermalization. Appendix~\ref{app:extra-numerics} contains additional numerical evidence, including results for kicked Ising dynamics away from the dual-unitary point and spectral tests of the fixed-trace Kraus ensemble. Finally, Appendix~\ref{app:quantum-information-recovery} presents the details of quantum information recovery analysis.

\section{Review of the basics}
\label{app:review}

In this section, we review the basics of Weingarten and Wick calculus for Haar random unitary ensemble and Ginibre ensemble, respectively. We then refresh some basic properties of the log-normal distribution, and explain our measure of closeness when comparing matrix ensembles.

\subsection{Weingarten and Wick calculus}
A Haar random unitary ensemble (also $\mathbf{CUE}$) is characterized by the uniform measure over the unitary group $\mathrm{U}(d)$. It is the unique probability measure that is both left and right unitary invariant. The statistics of the Haar unitary ensemble is given by Weingarten calculus. Diagrammatically, the $k$-th moment is
\begin{equation}
\begin{aligned}
\label{eq:weingarten-calculus}
    M_\mathrm{Haar}^{(k)} &:= \underset{U\sim\mathrm{Haar}}{\mathbb{E}}\left[U^{\otimes k} \otimes U^{*\otimes k}\right] = \underset{U\sim\mathrm{Haar}}{\mathbb{E}}\left[
    \vcenter{\hbox{\includegraphics[height=1.2cm]{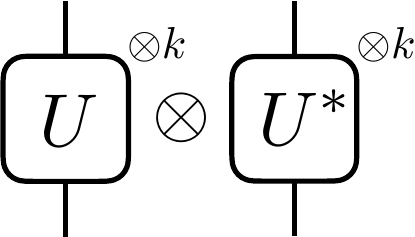}}}
    \right]
    \\
    &= \sum_{\sigma,\tau\in S_k}\mathrm{Wg}^{(k)}_d(\sigma^{-1}\tau)\, \vcenter{\hbox{\includegraphics[height=1.2cm]{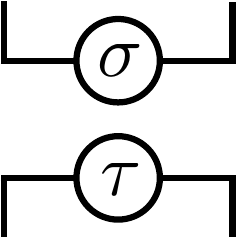}}}~,
\end{aligned}
\end{equation}
where $\mathrm{Wg}^{(k)}_d(\sigma^{-1}\tau)$ is the Weingarten matrix, thoroughly surveyed in~\cite{collins_weingarten_2022, mele_introduction_2024}. Here, $\sigma$ and $\tau$ are elements of the permutation group $S_k$, which also serve as the row and column indices of the Weingarten matrix. The permutation action on the $k$-copy space is defined through
\begin{equation}
    \vcenter{\hbox{\includegraphics[height=0.6cm]{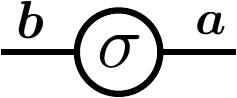}}} \equiv \langle\boldsymbol{b}|W_\sigma|\boldsymbol{a}\rangle = \prod_{i=1}^k \langle b_i|a_{\sigma^{-1}(i)}\rangle= \prod_{i=1}^{k}\langle b_{\sigma(i)}|a_i\rangle,
\end{equation}
where $i\in\{1,2,\cdots k\}$ stands for the replica index, $|\boldsymbol{a}\rangle = \bigotimes_{i=1}^{k}|a_i\rangle$ and  $\langle\boldsymbol{b}| = \bigotimes_{i=1}^{k}\langle b_i|$ are tensor products of local basis in the single-copy $d$-dimensional space, with $|a_i\rangle, |b_i\rangle\in\{|0\rangle, |1\rangle,\cdots|d-1\rangle\}$ for any $i$, and $W_\sigma$ is the unitary representation of $\sigma$ by permuting the $k$ copies of the Hilbert space. In the scaling limit $d\rightarrow\infty$, the Weingarten matrix has the following asymptotic property:
\begin{equation}
    \mathrm{Wg}_d^{(k)}(\mathrm{id}) = \frac{1}{d^k} + \mathcal{O}\left(\frac{1}{d^{k+2}}\right)\quad\text{and}\quad\mathrm{Wg}_d^{(k)}(\sigma) = \mathcal{O}\left(\frac{1}{d^{k+|\sigma|}}\right)\text{~for~}\sigma\ne\mathrm{id},
\end{equation}
where $|\sigma|$ denotes the minimal number of transpositions of the permutation $\sigma$ away from the identity $\mathrm{id}$. This suggests that in the limit $d\rightarrow\infty$, the Weingarten matrix is dominated by the leading terms on the diagonal entries.

Similarly, the statistics of a Ginibre ensemble can be derived through Wick contractions. Consider a Ginibre ensemble of $d_A\times d_A$ matrices, denoted $\mathbf{GinUE}_{d_A\times d_A}$ in the main text. It consists of i.i.d. matrix elements following the standard complex Gaussian distribution with zero mean and unit variance. Diagrammatically, its $k$-th moment is
\begin{equation}
\begin{aligned}
    M_\mathrm{Gin}^{(k)}&:= \underset{G\sim\mathrm{Gin}}{\mathbb{E}}\left[G^{\otimes k} \otimes G^{*\otimes k}\right] = \underset{G\sim\mathrm{Gin}}{\mathbb{E}}\left[\vcenter{\hbox{\includegraphics[height=1.2cm]{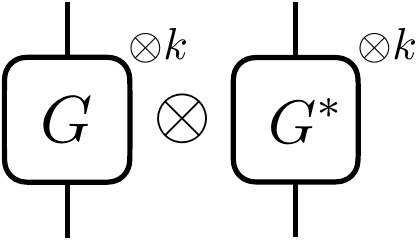}}}
    \right]
    \\
    &= \sum_{\sigma\in S_k}~\vcenter{\hbox{\includegraphics[height=1.2cm]{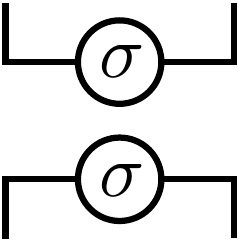}}}~,
\end{aligned}
\end{equation}
which consists of only the diagonal terms. This fact does not rely on $d_A$ being large.

In this work, we find it convenient to view $M_\mathrm{Gin}^{(k)}$ as a matrix from right to left, this is equivalent to treating the matrix ensembles in their vectorized version. Diagrammatically, this is equivalent as
\begin{equation}
\begin{aligned}
\label{eq:M_Gin}
    M_\mathrm{Gin}^{(k)} &= \underset{G\sim\mathrm{Gin}}{\mathbb{E}}\left[|G)(G|^{\otimes k}\right]= \underset{G\sim\mathrm{Gin}}{\mathbb{E}}\left[~\vcenter{\hbox{\includegraphics[height=1.15cm]{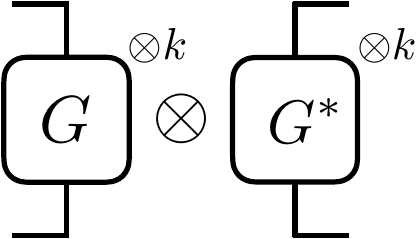}}}~\right]
    \\
    &= \sum_{\sigma\in S_k}~\vcenter{\hbox{\includegraphics[height=1.05cm]{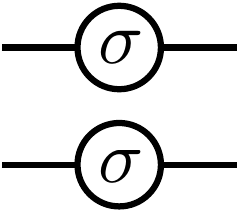}}} \,= (d_A^2)_k\,\rho^{(k)}_{H, d_A^2} = k!\,\mathbb{P}_{\mathrm{sym}, d_A^2}^{(k)},
\end{aligned}
\end{equation}
where we use the notation $(d)_k:=d(d+1)\cdots(d+k-1)$ as the rising factorial, $\rho^{(k)}_{H, d}:= \int \mathrm{d}\mu_\mathrm{Haar}(\psi)|\psi\rangle\langle\psi|^{\otimes k}$ for the $k$-th moment of a $d$-dimensional Haar random state ensemble, and $\mathbb{P}_{\mathrm{sym}, d}^{(k)}$ for the projector onto the totally symmetric subspace of $k$ copies of a $d$-dimensional Hilbert space. We have made use of the fact that
\begin{equation}
\label{eq:rho_Haar_k}
    \rho^{(k)}_{H, d} = \frac{\mathbb{P}_{\mathrm{sym}, d}^{(k)}}{(d)_k/k!} = \frac{\sum_{\sigma\in S_k}W_\sigma}{(d)_k},
\end{equation}
and $\frac{(d)_k}{k!} = \begin{pmatrix} d + k - 1\\k\end{pmatrix}$ is the dimension of the totally symmetric subspace. For later convenience, we also note that
\begin{equation}
    \mathrm{Tr}\left(M_\mathrm{Gin}^{(k)}\right) = (d_A^2)_k, \quad \mathrm{Tr}\left(M_\mathrm{Gin}^{(k)\,2}\right) = \left\| M_\mathrm{Gin}^{(k)}\right\|_2^2 = k!\,(d_A^2)_k.
\end{equation}

\subsection{LogNormal distribution}
\label{app:log-normal-introduction}

Log-normal distributions naturally emerge in multiplication of random variables. For $Y\sim\mathrm{LogNormal}(\mu, \sigma^2)$, the defining probability density function (PDF) is
\begin{equation}
    p_Y(y) = \frac{1}{y\sigma\sqrt{2\pi}}\exp(-\frac{(\log(y) - \mu)^2}{2\sigma^2})\quad\text{for}~y\in(0,+\infty),
\end{equation}
and the logarithm of $Y$ follows a real normal distribution $\log(Y)\sim\mathcal{N}(\mu, \sigma^2)$. All positive moments of the log-normal distribution exist:
\begin{equation}
    \mathbb{E}\left[Y^n\right] = \exp(n\mu+\frac{n^2\sigma^2}{2})\quad\text{for}~n=1,2,3,\cdots.
\end{equation}

An Ansatz for $\sqrt{d}K\sim\mathbf{GinUE}\times\mathrm{LogNormal}(\mu, \sigma^2)$ must satisfy the normalization condition for Kraus operators:
\begin{equation}
    \mathbbm{1}_A = \sum_\mathbf{z}K_\mathbf{z}^\dagger K_\mathbf{z} \leftrightarrow d_B\mathbb{E}\left[K^\dagger K\right] = \frac{1}{d_A}\mathbb{E}\left[Y^2\right]\mathbb{E}\left[G^\dagger G\right] = \exp(2\mu+2\sigma^2)\mathbbm{1}_A,
\end{equation}
where we set $G\sim\mathbf{GinUE}$ and $Y\sim\mathrm{LogNormal}(\mu,\sigma^2)$. For consistency, we must demand $\mu = -\sigma^2$. Hence, we use the notation $\mathrm{LogNormal}_{\sigma^2}\equiv\mathrm{LogNormal(\mu=-\sigma^2, \sigma^2)}$ throughout this work. Then for $Y\sim\mathrm{LogNormal}_{\sigma^2}$, we have
\begin{equation}
\underset{Y\sim\mathrm{LogNormal}_{\sigma^2}}{\mathbb{E}}\left[Y^{2k}\right] = \exp(2k(k-1)\sigma^2),
\end{equation}
and our Ansatz follows with
\begin{equation}
    M_\mathrm{th}^{(k)}(\sigma^2) = \mathbb{E}\left[Y^{2k}\right]\mathbb{E}\left[|G)(G|^{\otimes k}\right] = \exp(2k(k-1)\sigma^2)M_\mathrm{Gin}^{(k)} = x_\mathrm{th}^{(k)}(\sigma^2)M_\mathrm{Gin}^{(k)},
\end{equation}
where we denote $x_\mathrm{th}^{(k)}(\sigma^2)\equiv \exp(2k(k-1)\sigma^2)$.

The log-normal distribution is a long-tailed distribution, which suffers from numerical instability when sample size is insufficient, i.e., the fluctuations in the $k$-th moment can be large. To this end, when we extract a parameter fit $\hat\sigma^2$ to $\sigma^2$, we use the log norms of the rescaled Kraus operators $R_\mathrm{z}\equiv\mathrm{log}(\sqrt{d} \| K_\mathbf{z}\|_2)$ in Section~\ref{sec:numerics} for better stability.

\subsection{Kraus ensemble as unnormalized state ensemble}
\label{app:kraus-as-unnormalized-ensemble}
Assuming the vectorized form of matrix ensembles, we denote the $k$-th moment of the empirical Kraus ensemble as
        \begin{equation}
        \label{eq:kraus-kth-moment}
            M_\mathrm{emp}^{(k)}=\frac{1}{d_B}\sum_{\mathbf{z}}|K_\mathbf{z})(K_\mathbf{z}|^{\otimes k},
        \end{equation}
        and we use $M_\mathrm{th}^{(k)}=\int\mathrm{d}\mu_\mathrm{th}(V)\,|V)(V|^{\otimes k}$ to denote $k$-th moment of a theoretical distribution, where $\mu_\mathrm{th}$ is its probability measure. %
        By Eq.~\eqref{eq:signal-noise-decomposition}, there is a unique decomposition of the empirical moments, which we restate here:
\begin{equation}
    d^kM_\mathrm{emp}^{(k)} = x_\mathrm{emp}^{(k)}\left(M^{(k)}_\mathrm{Gin} + M_\perp^{(k)}\right),
\end{equation}
where $x_\mathrm{emp}^{(k)} = d^k\mathrm{Tr}\left(M_\mathrm{emp}^{(k)}\right)\big/\mathrm{Tr}\left(M_\mathrm{Gin}^{(k)}\right)$ and $\mathrm{Tr}\left(M_\mathrm{Gin}^{(k)}M_\perp^{(k)}\right) = 0$. This is because $M_\mathrm{emp}^{(k)}$ can be interpreted as the $k$-th moment density matrix of an unnormalized state ensemble $\{|K_\mathbf{z})\}$, which lives within the totally symmetric subspace specified by the projector $\mathbb{P}^{(k)}_{\mathrm{sym},d_A^2}$. Such a density matrix can always be decomposed into two parts: one proportional to $\mathbb{P}^{(k)}_{\mathrm{sym},d_A^2}$ (and therefore $M_\mathrm{Gin}^{(k)}$), and a traceless part, illustrated by Fig.~\ref{fig:decompose}.
\begin{figure}[!h]
    \centering
    \includegraphics[width=0.325\linewidth]{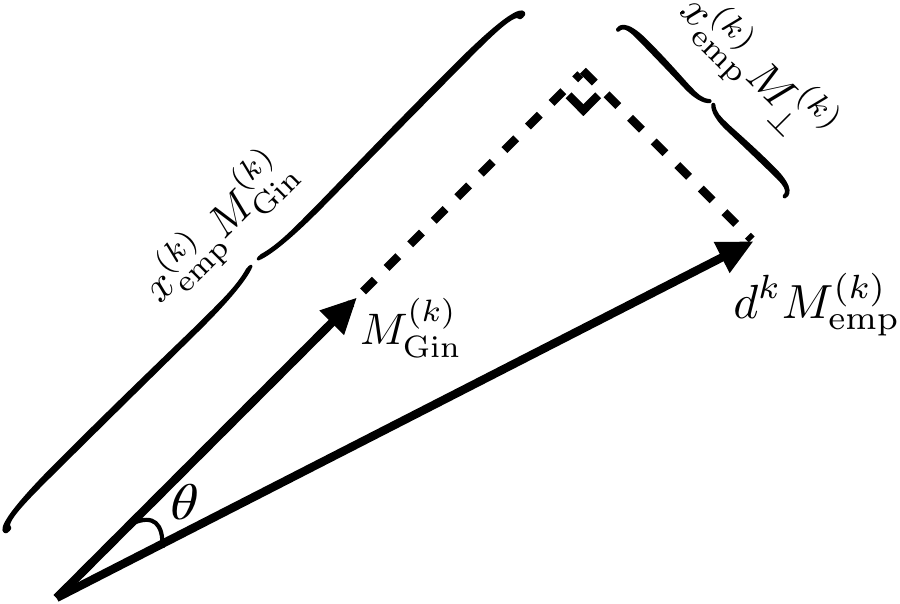}
    \caption{Decomposition of $k$-th moment of the empirical ensemble, where the rescaled empirical $k$-th moment $d^k M_\mathrm{emp}^{(k)}$ is decomposed into a part that is aligned with $M_\mathrm{Gin}^{(k)}$, and a part that is orthogonal (in the Hilbert-Schmidt sense) to $M_\mathrm{Gin}^{(k)}$.}
    \label{fig:decompose}
\end{figure}

It is helpful to define the \textit{generalized purity} of the Kraus ensemble
\begin{equation}
    \mathcal{P}_\mathrm{emp}^{(k)}:= d^k\mathrm{Tr}\left(M_\mathrm{emp}^{(k)}\right) = \frac{d^k}{d_B}\sum_\mathbf{z}(K_\mathbf{z}|K_\mathbf{z})^k,
\end{equation}
which captures the statistical moments of the norms of the empirical Kraus ensemble, with $x_\mathrm{emp}^{(k)} = \mathcal{P}_\mathrm{emp}^{(k)} / (d_A^2)_k$. We also define the \textit{generalized frame potential} for the Kraus ensemble
\begin{equation}
    \mathcal{F}_\mathrm{emp}^{(k)}:= d^{2k}\mathrm{Tr}\left(M_\mathrm{emp}^{(k)\, 2}\right) = \frac{d^{2k}}{d_B^2}\sum_{\mathbf{z},\mathbf{z}'}\left|(K_\mathbf{z}|K_{\mathbf{z}'})\right|^{2k},
\end{equation}
which satisfies
\begin{equation}
\begin{aligned}
\label{eq:inequality-FP}    \mathcal{F}_\mathrm{emp}^{(k)}&=x_\mathrm{emp}^{(k)\,2}\mathrm{Tr}\left( M_\mathrm{Gin}^{(k)\,2}\right) + x_\mathrm{emp}^{(k)\,2}\mathrm{Tr}\left(M_\perp^{(k)\,2}\right)
    \\
    &= \frac{k!\,\mathcal{P}_\mathrm{emp}^{(k)\,2}}{(d_A^2)_k}\left(1 + \Delta_\mathrm{run}^{(k)\,2}\right)\ge\frac{k!\,\mathcal{P}_\mathrm{emp}^{(k)\,2}}{(d_A^2)_k},
\end{aligned}
\end{equation}
where we use the definition of the running distance in Eq.~\eqref{eq:running-distance}, restated as follows:
\begin{equation}
    \Delta_\mathrm{run}^{(k)} := \frac{\left\| d^k M_\mathrm{emp}^{(k)} - x_\mathrm{emp}^{(k)}M_\mathrm{Gin}^{(k)}\right\|_2}{\left\| x_\mathrm{emp}^{(k)} M_\mathrm{Gin}^{(k)}\right\|_2} = \frac{\left\| M_\perp^{(k)}\right\|_2}{\left\| M_\mathrm{Gin}^{(k)}\right\|_2}.
\end{equation}
In Fig.~\ref{fig:decompose}, this corresponds to the tangent of the ``angle'' between $M_\mathrm{Gin}^{(k)}$ and $M_\mathrm{emp}^{(k)}$ as $\Delta_\mathrm{run}^{(k)} = \tan\theta$. Hence, the $k$-th moment distance of the empirical Kraus ensemble to the theoretical target ensemble is
\begin{equation}
\begin{aligned}
    \Delta_\mathrm{th}^{(k)} &= \frac{\left\|d^k M_\mathrm{emp}^{(k)} -M_\mathrm{th}^{(k)}(\sigma^2)  \right\|_2}{\left\| M_\mathrm{th}^{(k)}(\sigma^2) \right\|_2} = \frac{\left\| \left(x^{(k)}_\mathrm{emp}- x_\mathrm{th}^{(k)}(\sigma^2)\right)M_\mathrm{Gin}^{(k)} + x_\mathrm{emp}^{(k)}M_\perp^{(k)}\right\|_2}{\left\| x_\mathrm{th}^{(k)}(\sigma^2) M_\mathrm{Gin}^{(k)} \right\|_2}
    \\
    &= \sqrt{\left(\frac{x_\mathrm{emp}^{(k)}}{x_\mathrm{th}^{(k)}(\sigma^2)}-1\right)^2 + \left(\frac{x_\mathrm{emp}^{(k)}}{x_\mathrm{th}^{(k)}(\sigma^2)}\right)^2\Delta_\mathrm{run}^{(k)\,2}}\le \sqrt{\delta^{(k)\,2} + (1+\delta^{(k)})^2\Delta_\mathrm{run}^{(k)\,2 
            }},
\end{aligned}
\end{equation}
where we used $\delta^{(k)}:= \frac{\left| \mathcal{P}_\mathrm{emp}^{(k)} - \mathcal{P}_\mathrm{th}^{(k)}\right|}{\mathcal{P}_\mathrm{th}^{(k)}}=\frac{\left|x_\mathrm{emp}^{(k)} - x_\mathrm{th}^{(k)}\right|}{x_\mathrm{th}^{(k)}}$ as the residue error in the $k$-th generalized purity of the empirical Kraus ensemble. In particular, when we set $\sigma^2=0$ explicitly with $x_\mathrm{Gin}^{(k)}\equiv x_\mathrm{th}^{(k)}(\sigma^2=0)=1$, this reduces to
\begin{equation}
\begin{aligned}
    \Delta_\mathrm{Gin}^{(k)} &:= \frac{\left\| d^k M_\mathrm{emp}^{(k)} - M_\mathrm{Gin}^{(k)} \right\|_2}{\left\|M_\mathrm{Gin}^{(k)}\right\|_2} = \sqrt{\left(x_\mathrm{emp}^{(k)}-1\right)^2 + x_\mathrm{emp}^{(k)\,2}\Delta_\mathrm{run}^{(k)\,2}}
    \\
    &=\sqrt{1 - \frac{2\mathcal{P}_\mathrm{emp}^{(k)}}{(d_A^2)_k} + \frac{\mathcal{F}_\mathrm{emp}^{(k)}}{k!\,(d_A^2)_k}}\ge\left| 1 - \frac{\mathcal{P}_\mathrm{emp}^{(k)}}{(d_A^2)_k}\right|,
\end{aligned}
\end{equation}
where we used Eq.~\eqref{eq:inequality-FP} in the second line. This inequality is saturated when $\Delta_\mathrm{run}^{(k)}=0$. Effectively, we may conclude that the $k$-th moment error (distance) with respect to a theoretical ensemble $\Delta_\mathrm{th}^{(k)}$ (or $\Delta_\mathrm{Gin}^{(k)}$ when the target is Ginibre) is captured fully by $\mathcal{P}_\mathrm{emp}^{(k)}$ (or equivalently $x_\mathrm{emp}^{(k)}$) and $\Delta_\mathrm{run}^{(k)}$. By Eq.~\eqref{eq:inequality-FP}, we note that $\mathcal{F}_\mathrm{emp}^{(k)}$ can be expressed in terms of $\mathcal{P}_\mathrm{emp}^{(k)}$ and $\Delta_\mathrm{run}^{(k)}$, thus the error profile can be captured by any two elements of the triple $(\mathcal{P}_\mathrm{emp}^{(k)}, \mathcal{F}_\mathrm{emp}^{(k)}, \Delta_\mathrm{run}^{(k)})$.

\section{Proof for global Haar case: emergence of $\mathbf{GinUE}$}
\label{app:global-haar-proof}

In this section, we provide a detailed proof of Theorem~\ref{thm:1-precise} in the main text. Given a Haar random unitary $U_{AB}$ on the composite system $AB$, the empirical $k$-th moment
\begin{equation}
    M^{(k)}_\mathrm{emp}(U_{AB}) = \frac{1}{d_B}\sum_\mathbf{z}|K_\mathbf{z})(K_\mathbf{z}|^{\otimes k},
\end{equation}
is itself a random tensor. Recall that we defined the $k$-th moment distance to Ginibre, the target ensemble, as
\begin{equation}
    \Delta_\mathrm{Gin}^{(k)} := \frac{\left\| d^kM_\mathrm{emp}^{(k)}(U_{AB}) - M_\mathrm{Gin}^{(k)}\right\|_2}{\left\| M_\mathrm{Gin}^{(k)}\right\|_2},
\end{equation}
where $M_\mathrm{Gin}^{(k)}$ is the $k$-th moment tensor given in Eq.~\eqref{eq:M_Gin}. We now restate Theorem~\ref{thm:1-precise} for convenience.
\begin{atheorem}
        Let the global unitary $U_{AB}$ be Haar random in $\mathrm{U}(d_Ad_B)$. Then for any integer $k\ge 1$ and any $0 < \delta <1$, with probability at least $1-\delta$, the $k$-th moment distance of the Kraus ensemble to Ginibre satisfies
        \begin{equation}
            \Delta_\mathrm{Gin}^{(k)} \le \sqrt{\frac{d_A^{2k}}{\delta\, k!\, d_B}\left(1+\mathcal{O}\left(\frac{k^2}{d_A^2}\right)\right)},
        \end{equation}
        for all sufficiently large $d_B$, where the $\mathcal{O}(\cdot)$ term is independent of $d_B$.
\end{atheorem}
We point out that this applies for a single realization of the global unitary, unlike that in the previous result~\cite{mastrodonato_elementary_2007} for the \textit{truncated unitary ensemble} ($\mathbf{TUE}$) generated from a continuous set of global unitary realizations. In order to prove this theorem, we quote Lemma 1 in~\cite{mok_optimal_2025} as follows, which we restate without proving.

\begin{lemma}[Large-$d$ approximation of Weingarten sums~\cite{mok_optimal_2025}]
\label{lemma:weingarten-large-d}
    For any function $a_{\sigma\tau}$ of the permutation elements $\sigma$, $\tau$ of the symmetric group $S_k$ with $k^2 < d$, the error of the leading order approximation is
    \begin{equation}
        \left| \sum_{\tau\in S_k}a_{\sigma\tau}\mathrm{Wg}^{(k)}_d (\sigma^{-1}\tau) - a_{\sigma\sigma}\mathrm{Wg}^{(k)}_d(\sigma^{-1}\sigma=\mathrm{id}) \right| = \mathcal{O}\left( \frac{k^2}{d^{k+1}}\max_{\sigma\ne\tau}|a_{\sigma\tau}| \right).
    \end{equation}
\end{lemma}
This lemma provides a theoretical foundation for approximating a sum that includes Weingarten terms using only the diagonal terms. Equipped with this error estimation, we are now able to show the following lemma:
\begin{lemma}
\label{lemma:Haar-Delta-Gin}
    For global Haar random unitary $U_{AB}$ sampled from $\mathrm{U}(d_Ad_B)$, the squared distance satisfies
    \begin{equation}
        \mathbb{E}\left[\Delta_\mathrm{Gin}^{(k)\,2}\right]\le\frac{d_A^{2k}}{k!\cdot d_B}\left(1 + \mathcal{O}\left(\frac{k^2}{d_A^2}\right)\right),
    \end{equation}
    for sufficiently large $d_B$, where $\mathbb{E}[\cdot]$ denotes expectation of $U_{AB}$ over the Haar measure.
\end{lemma}
\textit{Proof.} First recall the definition of generalized frame potential and purity of the empirical ensemble:
\begin{equation}
    \mathcal{F}^{(k)}_\mathrm{emp} = \frac{d^{2k}}{d_B^2}\sum_{\mathbf{z},\mathbf{z}'}\left| \left(K_\mathbf{z}|K_{\mathbf{z}'}\right)\right|^{2k}=d^{2k}\mathrm{Tr}\left(M^{(k)\,2}_\mathrm{emp}\right),\quad\text{and}\quad     \mathcal{P}^{(k)}_\mathrm{emp}=\frac{d^k}{d_B}\sum_\mathbf{z}(K_\mathbf{z}|K_\mathbf{z})^{k}=d^k\mathrm{Tr}\left(M_\mathrm{emp}^{(k)}\right),
\end{equation}
then by properties of the Frobenius norm, we have
\begin{equation}
\begin{aligned}
\label{eq:RUC-frobenius}
    {\mathbb{E}}\left[
        \left\| d^k M_\mathrm{emp}^{(k)} - M_\mathrm{Gin}^{(k)} \right\|_2^2
        \right] &= d^{2k}{\mathbb{E}}\left[\mathrm{Tr}\left(M_\mathrm{emp}^{(k)\,2}\right)\right] - 2d^k{\mathbb{E}}\left[\mathrm{Tr}\left( M_\mathrm{emp}^{(k)}M_\mathrm{Gin}^{(k)}\right) \right] + \mathrm{Tr}\left(M_\mathrm{Gin}^{(k)\,2}\right)
        \\
        &=\mathbb{E}\left[\mathcal{F}_\mathrm{emp}^{(k)}\right] - 2k!\mathbb{E}\left[\mathcal{P}_\mathrm{emp}^{(k)}\right] + \mathrm{Tr}\left(M_\mathrm{Gin}^{(k)\,2}\right).
\end{aligned}
\end{equation}
We then evaluate these terms separately. For the generalized frame potential, we have
\begin{equation}
\begin{aligned}
\label{eq:RUC-expectation-F}
\mathbb{E}\left[\mathcal{F}_\mathrm{emp}^{(k)}\right] &= \frac{d^{2k}}{d_B^2}\sum_{\mathbf{z},\mathbf{z}'}\mathbb{E}\left[ \vcenter{\hbox{\includegraphics[height=1.7cm]{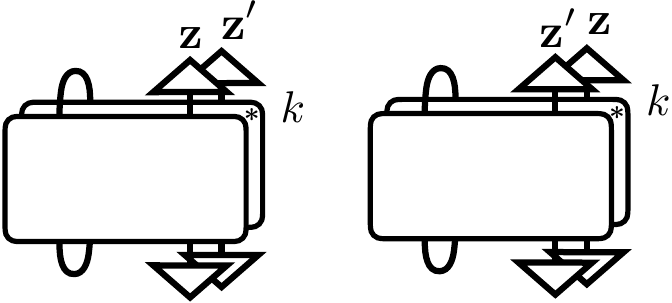}}} \right]
\\
&= \frac{d^{2k}}{d_B}\mathbb{E}\left[ \vcenter{\hbox{\includegraphics[height=1.45cm]{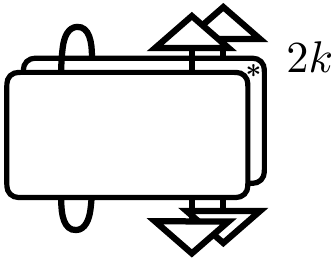}}} \right] + \frac{d^{2k}(d_B-1)}{d_B}\mathbb{E}\left[ \vcenter{\hbox{\includegraphics[height=1.5cm]{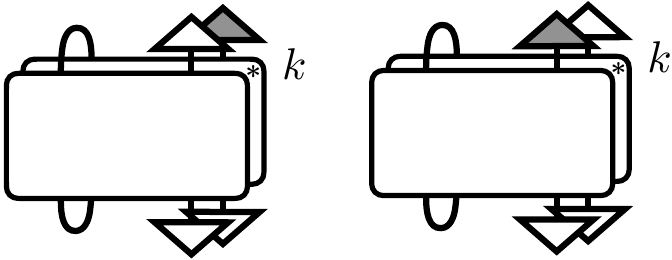}}} \right]
\\
&= \frac{d^{2k}}{d_B}
\underbrace{\sum_{\sigma,\tau\in S_{2k}}\vcenter{\hbox{\includegraphics[height=1.25cm]{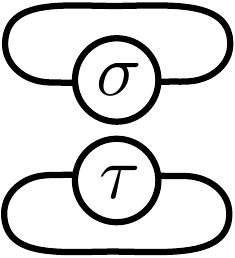}}}\mathrm{Wg}_d^{(2k)}(\sigma^{-1}\tau)}_{(\text{i})} + \frac{d^{2k}(d_B-1)}{d_B}\underbrace{\sum_{\substack{(\sigma_1,\sigma_2)\in S_k^{\times 2} \\ \tau\in S_{2k}}} \vcenter{\hbox{\includegraphics[height=1.3cm]{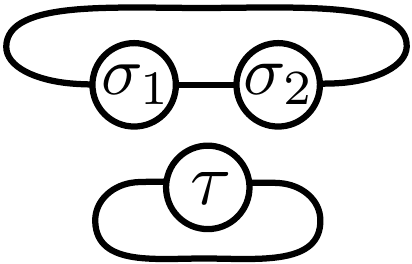}}}\mathrm{Wg}_d^{(2k)}\left((\sigma_1,\sigma_2)^{-1}\tau\right)}_{(\text{ii})}
\\
&\le k!(d_A^2)_k + \frac{1}{d_B}\left( (d_A^2)_{2k} - k! (d_A^2)_k \right) + \mathcal{O}\left(\frac{k! k^2 d_A^{2k} (d_A)_k}{d}\right),
\end{aligned}
\end{equation}
where in the first line we represent the Kraus operators diagrammatically, with the squares in front representing $U_{AB}$, and the squares in the back representing $U_{AB}^*$; in the second line we separate the sum into two parts: the first contains contributions with $\mathbf{z}=\mathbf{z}'$, and the second contains contributions with $\mathbf{z}\ne\mathbf{z}'$, represented with different colors in the top boundary; in the third line we use the Weingarten calculus summarized in Eq.~\eqref{eq:weingarten-calculus}, for the second term, in place of $\sigma$ now constrained by the top boundary condition, only permutation spins in the form of $(\sigma_1, \sigma_2)\in S_k^{\times 2}\subset S_{2k}$ survive, here $\sigma_1, \sigma_2\in S_k$ and $(\sigma_1, \sigma_2)$ is a permutation on $2k$ copies defined by $(\sigma_1, \sigma_2)(i) = k+\sigma_1(i)$ and $(\sigma_1, \sigma_2)(k+i)=\sigma_2(i)$ for $1\le i \le k$. We then use lemma~\ref{lemma:weingarten-large-d} for $(\text{i})$ and $(\text{ii})$ to bound the off-diagonal sum: for sufficiently large $d_B$ such that $(2k)^2 < d = d_Ad_B$ and lemma~\ref{lemma:weingarten-large-d} holds, we have
\begin{equation}
\begin{aligned}
    (\text{i}) &= \sum_{\sigma\in S_{2k}}\left( d_A^{2\#\text{cyc}(\sigma)} \mathrm{Wg}_d^{(2k)}(\mathrm{id}) + \mathcal{O}\left( \frac{(2k)^2}{d^{2k+1}}d_A^{2k}d_A^{\#\text{cyc}(\sigma)
    } \right)\right)
    \\
    &= \frac{(d_A^2)_{2k}}{d^{2k}} + \mathcal{O}\left( \frac{k^2}{d^{2k+1}}d_A^{2k}\underbrace{\sum_{\sigma\in S_{2k}}d_A^{\#\text{cyc}(\sigma)}}_{(d_A)_{2k}} \right) = \frac{(d_A^2)_{2k}}{d^{2k}} + \mathcal{O}\left( \frac{k^2 d_A^{2k}(d_A)_{2k}}{d^{2k+1}} \right),
\end{aligned}
\end{equation}
where $\#\text{cyc}(\sigma)$ is the number of cycles in a permutation $\sigma$, and
\begin{equation}
\begin{aligned}
    (\text{ii}) & = \sum_{(\sigma_1, \sigma_2)\in S_k^{\times 2}} \left(d_A^{2\#\text{cyc}(\sigma_1 \sigma_2)} \mathrm{Wg}_d^{(2k)}(\mathrm{id}) + \mathcal{O}\left( \frac{(2k)^2}{d^{2k+1}}d_A^{2k} d_A^{\#\text{cyc}(\sigma_1\sigma_2)} \right)\right)
    \\
    &= \frac{k! (d_A^2)_k}{d^{2k}} + \mathcal{O}\left( \frac{k^2}{d^{2k+1}}k!\, d_A^{2k} (d_A)_k \right),
\end{aligned}
\end{equation}
where $\sigma_1\sigma_2\in S_k$ is a permutation on $k$ copies, and $k!$ comes from the cardinality of $S_k$. Combining these estimations for $(\text{i})$ and $(\text{ii})$, we arrive at the fourth line in Eq.~\eqref{eq:RUC-expectation-F}. Similarly, for the generalized purity, we have
\begin{equation}
\begin{aligned}
    \mathbb{E}\left[\mathcal{P}_\mathrm{emp}^{(k)}\right] &= d^k\,\mathbb{E}\left[\vcenter{\hbox{\includegraphics[height=1.45cm]{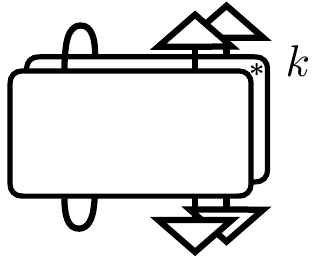}}}\right]
    = d^k \sum_{\sigma,\tau\in S_k} \vcenter{\hbox{\includegraphics[height=1.25cm]{graph_perm_contract_same.pdf}}}\mathrm{Wg}_d^{(k)}(\sigma^{-1}\tau)
    \\
    &= d^k\sum_{\sigma \in S_k}\left( d_A^{2\#\text{cyc}(\sigma)}\mathrm{Wg}_d^{(k)}(\mathrm{id}) + \mathcal{O}\left( \frac{k^2}{d^{k+1}}d_A^k d_A^{\#\text{cyc}(\sigma)} \right)\right)
    =
    (d_A^2)_k + \mathcal{O}\left( \frac{k^2 d_A^k (d_A)_k}{d} \right).
\end{aligned}
\end{equation}
Combining these, we have the following inequality continuing Eq.~\eqref{eq:RUC-frobenius}
\begin{equation}
    \mathbb{E}\left[ \left\| d^k M_\mathrm{emp}^{(k)} - M_\mathrm{Gin}^{(k)} \right\|_2^2 \right] \le \frac{d_A^{4k}}{d_B}\left(1 + \mathcal{O}\left(\frac{k^2}{d_A^2}\right)\right).
\end{equation}
Substitute and we find that
\begin{equation}
    \mathbb{E}\left[\Delta_\mathrm{Gin}^{(k)\,2}\right]\le\frac{d_A^{2k}}{k!\cdot d_B}\left(1 + \mathcal{O}\left(\frac{k^2}{d_A^2}\right)\right)
\end{equation}
which shows the desired result. \quad$\square$

With lemma~\ref{lemma:Haar-Delta-Gin}, we are then ready to prove Theorem~\ref{thm:1-precise}.

\textit{Proof of Theorem~\ref{thm:1-precise}.} Markov inequality guarantees that
\begin{equation}
    \mathbb{P}\left[ \Delta_\mathrm{Gin}^{(k)}\ge \epsilon \right] \le \frac{\mathbb{E}\left[\Delta_\mathrm{Gin}^{(k)\,2}\right]}{\epsilon^2}\le\frac{d_A^{2k}}{\epsilon^2\cdot k!\cdot d_B}\left(1 + \mathcal{O}\left(\frac{k^2}{d_A^2}\right)\right),
\end{equation}
for any $\epsilon > 0$. By setting the RHS to be $\delta$, we can claim that for any $0 < \delta < 1$,
\begin{equation}
    \Delta_\mathrm{Gin}^{(k)} \le \sqrt{\frac{d_A^{2k}}{\delta\cdot k!\cdot d_B}\left(1+\mathcal{O}\left(\frac{k^2}{d_A^2}\right)\right)},
\end{equation}
with probability at least $(1-\delta)$. This concludes the proof. \quad$\blacksquare$

The results in this appendix suggest that for global Haar random unitary, the Kraus ensemble has distance to Ginibre $\Delta_\mathrm{Gin}^{(k)}$ scaling as $1/\sqrt{d_B}$. We comment here that if we consider the \textit{truncated unitary ensemble} ($\mathbf{TUE}$), which is a continuous distribution of a submatrix of a larger global unitary across the realizations of the global unitary, the distance to Ginibre would scale as $1/d_B$.

\section{Kicked Ising model and proof for dual-unitary case}
\label{app:DU-KIM-proof}

In this section, we introduce the tensor network representation of a kicked Ising model (KIM) in Section~\ref{sec:DU-Floquet}, and provide a proof for Theorem~\ref{thm:2-precise}. Recall that the Floquet unitary of KIM is defined as
\begin{equation}
    U_F = U_h e^{-i H_\mathrm{Ising}},\quad\text{where}~H_\mathrm{Ising} = J\sum_{i=1}^{N-1} \sigma_i^z\sigma_{i+1}^z + \sum_{i=1}^Ng_i\sigma_i^z + (b_1 \sigma_1^z + b_N\sigma_N^z)~\text{and}~U_h = \exp(-ih\sum_{i=1}^N\sigma_i^y).
\end{equation}
Here, $g_i\notin\frac{\mathbb{Z}\pi}{8}$ and $b_1=b_N = \frac{\pi}{4}$ on the boundary. First, we represent the single-qubit $Z$-rotation gate with a dot
\begin{equation}
    e^{- i g_i \sigma_i^z} = \,\vcenter{\hbox{\includegraphics[height=0.7cm]{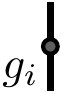}}}\, = R_z(2g_i) = \begin{pmatrix}
            e^{-ig_i} & ~
            \\
            ~ & e^{ig_i}
        \end{pmatrix},
\end{equation}
which satisfies the property: $\vcenter{\hbox{\includegraphics[height=0.9cm]{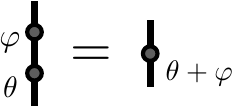}}}$\,. We also make use of the following notation for a multi-leg node:
\begin{equation}
    \vcenter{\hbox{\includegraphics[height=1.1cm]{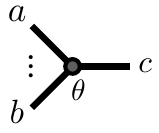}}}\, = \delta_{a\cdots b c} e^{-i\theta(1-2a)},
\end{equation}
where $a,\cdots,b,c\in\{0,1\}$ indicates the computational basis, $\delta_{a\cdots bc}=1$ iff $a=\cdots=b=c$ and $\delta_{a\cdots b c}=0$ otherwise. A node with no dot at the vertex is the special case with $\theta = 0$. By contracting one of the legs with a $|\pm\rangle$ state, we have
\begin{equation}
\label{eq:attach_pm}
    \vcenter{\hbox{\includegraphics[height=0.75cm]{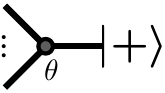}}}\, = \,\vcenter{\hbox{\includegraphics[height=0.75cm]{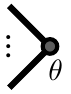}}}\, \big/ \sqrt{2},\quad\text{and}\quad\,\vcenter{\hbox{\includegraphics[height=0.75cm]{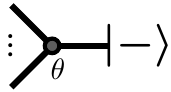}}}\, = \,i\times\vcenter{\hbox{\includegraphics[height=0.75cm]{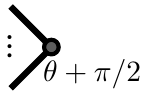}}}\, \big/ \sqrt{2},
\end{equation}
where $|\pm\rangle = \left(|0\rangle \pm |1\rangle\right)/\sqrt{2}$. This allows us to attach or remove a $|\pm\rangle$ state to any such node. The global phase $i$ in the second equality is irrelevant and can be dropped without affecting the physics. Next, the two-qubit $ZZ$-rotation is represented with
\begin{equation}
\label{eq:ZZ-rotation}
    \vcenter{\hbox{\includegraphics[height=0.55cm]{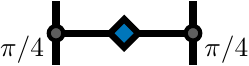}}}\,=\begin{pmatrix}
        -i & ~ & ~ &~
        \\
        ~ & 1 & ~ & ~
        \\
        ~ & ~ & 1 & ~
        \\
        ~ & ~ & ~ & i
    \end{pmatrix}\underbrace{\begin{pmatrix}
            1 & ~ & ~ & ~
            \\
            ~ & e^{i(2J - \pi/2)} & ~ & ~
            \\
            ~ & ~ & e^{i(2J - \pi/2)} & ~
            \\
            ~ & ~ & ~ & -1
        \end{pmatrix}}_{\mathsf{Blue}(J)}= -i e^{iJ} e^{-iJ \sigma^z\otimes \sigma^z},
\end{equation}
where we used the definition of a blue gate $\vcenter{\hbox{\includegraphics[height=0.4cm]{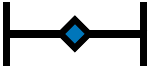}}}\,=\mathsf{Blue}(J)$ as in Eq.~\eqref{eq:blue-square}. The global phase $-i e^{iJ}$ is irrelevant and can be dropped. Finally, the single-qubit $Y$-rotations in $U_h$ can be represented as
\begin{equation}
\label{eq:Y-rotation}
    \vcenter{\hbox{\includegraphics[height=1cm]{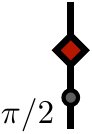}}}\,= \begin{pmatrix}
            \cos(h) & \sin(h)
            \\
            \sin(h) & -\cos(h)
        \end{pmatrix}\begin{pmatrix}
        -i & ~
        \\
        ~ & i
    \end{pmatrix}= -ie^{-i h \sigma^y},
\end{equation}
where we used the definition of a red gate $\vcenter{\hbox{\includegraphics[height=0.7cm]{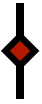}}}\,=\mathsf{Red}(h)$ as in Eq.~\eqref{eq:red-square}. Likewise, the global phase $-i$ can be dropped. Combining these, we represent the Floquet unitary with
\begin{equation}
    U_F = \vcenter{\hbox{\includegraphics[height=1.2cm]{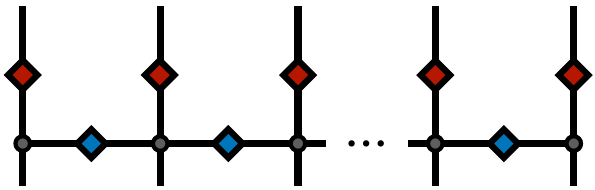}}}\, ,
\end{equation}
up to an irrelevant global phase, where the dot on the $i$-th site is parametrized by $g_i$. Here, we have absorbed the single-qubit $Z$-rotations (dots) arising from Eqs.~\eqref{eq:ZZ-rotation} and~\eqref{eq:Y-rotation}: for each qubit, there are two $\frac{\pi}{4}$-rotations (including the leftmost and rightmost qubit as we added the $b_1$ and $b_N$ terms) from $ZZ$-rotations in Eq.~\eqref{eq:ZZ-rotation}, and there is one $\frac{\pi}{2}$-rotation from $Y$-rotations in Eq.~\eqref{eq:Y-rotation}. These contributions on each qubit add up to a $Z$-rotation of angle $\pi$, which results in a trivial global phase $R_z(2\pi)=-1$ that we dropped.

At the DU point $J=h=\frac{\pi}{4}$, the blue gate reduces to a controlled-$Z$ gate and the red gate to a Hadamard gate. For concreteness, we also set $g_i = g\notin\frac{\mathbb{Z}\pi}{8}$ to be uniform, then our setup is identical to that of~\cite{ho_exact_2022}. We now turn to proving Theorem~\ref{thm:2-precise}, which we restate here for convenience:
\begin{atheorem}
    For the dual-unitary kicked Ising model with uniform $g_i = g \notin \frac{\mathbb{Z}\pi}{8}$ in the thermodynamic limit $N_B\rightarrow\infty$, for any integer $k\ge 1$ and any $t\ge 2N_A$, the Kraus ensemble $\mathcal{K}(t)$ satisfies
    \begin{equation}
        \lim_{N_B\rightarrow\infty} \Delta_\mathrm{Gin}^{(k)}(t) = 1 - \frac{2^{kt}}{(2^t)_k},
    \end{equation}
    where $(x)_k:= x(x+1)\cdots(x+k-1)$ is the rising factorial. For fixed $k$, this distance decays exponentially in time as
    \begin{equation}
        \lim_{N_B\rightarrow\infty}\Delta_\mathrm{Gin}^{(k)}(t) = \mathcal{O}\left(\frac{k^2}{2^t}\right),\qquad(t\rightarrow\infty).
    \end{equation}
\end{atheorem}

To show this, we prove a stronger result about the $k$-th moment of the empirical Kraus ensemble of the DU KIM, which we denote as $M_\mathrm{DU-KIM}^{(k)}$ (in place of $M_\mathrm{emp}^{(k)}$).

\begin{lemma}
\label{lemma:du-kim}
    For the dual-unitary kicked Ising model in the thermodynamic limit $N_B\rightarrow\infty$, the Kraus ensemble at any $t\ge 2N_A$ satisfies
    \begin{equation}
        \lim_{N_B\rightarrow\infty} d^k M^{(k)}_\mathrm{DU-KIM} = \frac{2^{kt}}{(2^t)_k}M_\mathrm{Gin}^{(k)}.
    \end{equation}
\end{lemma}
Then Theorem~\ref{thm:2-precise} follows directly from this lemma. In order to prove this, let us take a closer look at the structure of the DU KIM. A typical Kraus operator can be represented through tensor network contractions in Fig.~\ref{fig:du-kim-structure}, similar to the generic 1D circuit discussed in Section~\ref{sec:origin-ansatz} as shown in Fig.~\ref{fig:oseledets}. Here, the circuit bulk gives rise to a state $|\widetilde\Omega_\mathbf{z}\rangle$ at the cut between $A$ and $B$ upon contraction, which is then acted on by the edge operator $W$, resulting in the Kraus operator in its vectorized form: $|K_\mathbf{z}) = W|\widetilde\Omega_\mathbf{z}\rangle$. For DU KIM specifically, there are two analytically tractable properties: (i) the transfer matrices in the bulk are proportional to some unitary, namely $V_i = \frac{1}{\sqrt{2}}U_{z_i}$ where $\{U_0, U_1\}$ can be shown to form a universal gate set. Random long product of such gates approximates a Haar random unitary. (ii) The edge operator $W$ can be shown to be proportional to an isometry, namely a partial projection that maps a $t$-qubit state to a $(2N_A)$-qubit state. We will make use of these properties to prove the lemma.

\begin{figure}[!h]
    \centering
    \includegraphics[width=0.675\linewidth]{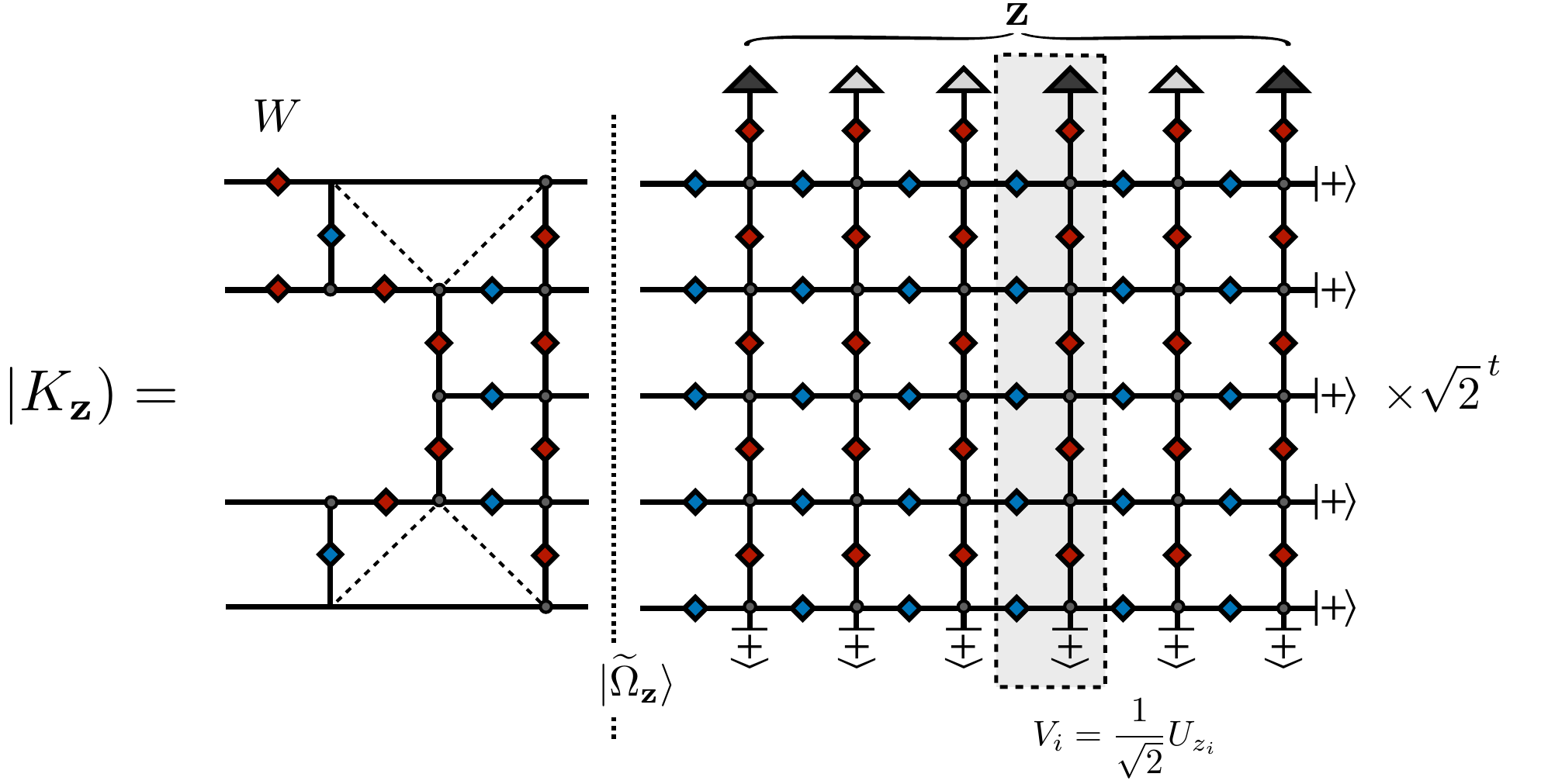}
    \caption{A typical Kraus operator of the dual-unitary (DU) KIM, shown for $N_A=2$, $N_B=6$, and $t=5$. In the bulk, we attach a product state $|\Omega\rangle = |+\rangle^{\otimes t}$ at the rightmost slice by Eq.~\eqref{eq:attach_pm}, thereby introducing a factor of $\sqrt{2}^{\,t}$. For fixed $g_i = g\notin\frac{\mathbb{Z}\pi}{8}$, the transfer matrices (shown in the grey box) are proportional to a unitary, namely $V_i = \frac{1}{\sqrt{2}}U_{z_i}$. This can be seen easily by noting that the blue gates and the red gate switch roles when viewed sideways, and the measurement on the top translates to projections by $|\pm\rangle$. Therefore, the final state at the cut is $|\widetilde\Omega_\mathbf{z}\rangle = \frac{1}{\sqrt{2}^{N_B - t}}U_{z_{N_B}}\cdots U_{z_2} U_{z_1}|+\rangle^{\otimes t}$. The form of the edge operator $W$ is a simple reshaping of the original tensor network. This can be seen by identifying and gluing back together the pairs of dashed lines on top and at the bottom of $W$.}
    \label{fig:du-kim-structure}
\end{figure}

By denoting $\mathcal{U}_\mathbf{z} = U_{z_B}\cdots U_{z_2} U_{z_1}$ and collecting them into an equal-weight unitary ensemble $\mathcal{E}_\mathcal{U}$,~\cite{ho_exact_2022} has shown that they form an exact unitary-design in the thermodynamic limit (TDL). Here, we quote their results without proving.

\begin{lemma}[Theorem 2 in~\cite{ho_exact_2022}]
\label{lemma:ho-exact}
    For $g\notin\mathbb{Z}\pi/8$, the unitary ensemble $\mathcal{E}_\mathcal{U}$ forms an exact unitary-design in the TDL. That is, all moments $k$ of $\mathcal{E}_\mathcal{U}$ and the Haar random unitary ensemble agree:
    \begin{equation}
        \lim_{N_B\rightarrow\infty} \sum_{\mathbf{z}}\frac{1}{d_B}\mathcal{U}_\mathbf{z}^{\otimes k}\otimes \mathcal{U}_\mathbf{z}^{*\otimes k} = \int_{U\sim\mathrm{Haar}}\mathrm{d}U\, U^{\otimes k}\otimes U^{*\otimes k} = M_\mathrm{Haar}^{(k)}.
    \end{equation}
\end{lemma}
This guarantees that the state $|\widetilde\Omega_\mathbf{z}\rangle$ is Haar-randomly distributed in the TDL up to a constant normalization, and also suggests that the Furstenberg measure of the transfer matrices of the DU KIM is the Haar measure. Next, we zoom in on the edge operator $W$ in Fig.~\ref{fig:du-kim-edge} and show that it is proportional to an isometry from $t$-qubit states to $(2N_A)$-qubit states when $t\ge 2N_A$. Here, the edge operator can be seen as an operator proportional to a unitary $\mathcal{V}$ on $t$ qubits, followed by projecting the middle $(t-2N_A)$ qubits to the $|+\rangle$ state. The prefactor can be obtained by counting within the orange trapezoid, where there are $N_A(t-N_A)$ red gates and $(N_A-1)(t-N_A)$ blue gates, and the overall edge operator can be seen as
\begin{equation}
\begin{aligned}
    W &= \left(\mathbbm{1}^{\otimes N_A}\otimes \langle +|^{\otimes (t - 2N_A)}\otimes \mathbbm{1}^{\otimes N_A}\right)\mathcal{V}\times \left( \sqrt{2}^{t-2N_A}\times \left(\frac{1}{\sqrt{2}}\right)^{N_A(t-N_A)} \times \sqrt{2}^{(N_A-1)(t - N_A)}\right)
    \\
    &= \frac{1}{\sqrt{2}^{N_A}} \left(\mathbbm{1}^{\otimes N_A}\otimes \langle +|^{\otimes (t - 2N_A)}\otimes \mathbbm{1}^{\otimes N_A}\right)\mathcal{V},
\end{aligned}
\end{equation}
where $\mathcal{V}$ is a unitary operator on $t$ qubits. Equipped with these properties, we are ready to prove Lemma~\ref{lemma:du-kim}.

\begin{figure}[!h]
    \centering
    \includegraphics[width=0.375\linewidth]{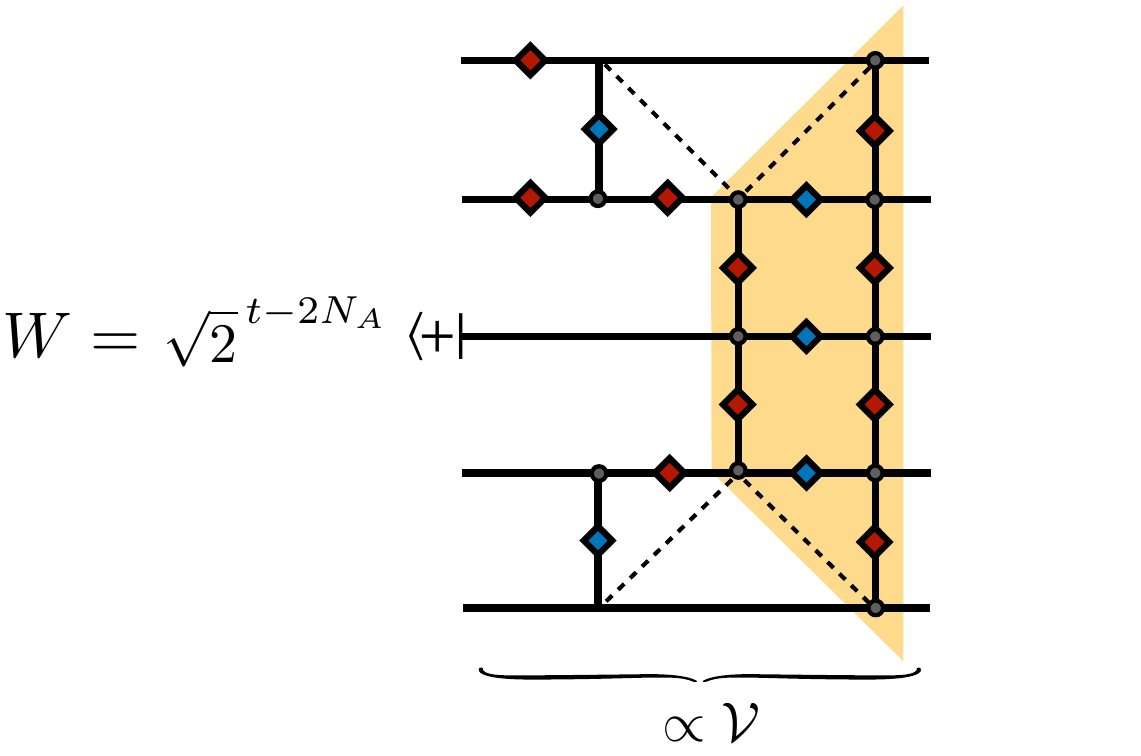}
    \caption{The edge operator $W$ with $t\ge 2N_A$ shown as an operator proportional to an isometry, here shown for $N_A=2$ and $t=5$. Interpreted from right to left, the gates within the orange trapezoid are proportional to a unitary, while the gates outside of the trapezoid are exactly unitary by themselves. Likewise, we have attached $|+\rangle$ states to the middle $(t-2N_A)$ qubits as well as their prefactors according to Eq.~\eqref{eq:attach_pm}.}
    \label{fig:du-kim-edge}
\end{figure}

\textit{Proof of Lemma~\ref{lemma:du-kim}.} Given the tensor network representation detailed in Fig.~\ref{fig:du-kim-structure}, and the structure of the edge operator $W$ when $t\ge 2N_A$ as illustrated in Fig.~\ref{fig:du-kim-edge}, a Kraus operator takes the form
\begin{equation}
\begin{aligned}
    |K_\mathbf{z}) = W|\widetilde\Omega_\mathbf{z}\rangle &= \frac{\sqrt{2}^t}{\sqrt{2}^{N_A + N_B}} \left(\mathbbm{1}^{\otimes N_A}\otimes \langle +|^{\otimes (t - 2N_A)}\otimes \mathbbm{1}^{\otimes N_A}\right)\underbrace{\mathcal{V}\mathcal{U}_\mathbf{z}|+\rangle^{\otimes t}}_{|\Psi_\mathbf{z}^{(t)}\rangle}
    \\
    &= \frac{\sqrt{2^t}}{\sqrt{d}}\langle+|^{\otimes (t-2N_A)}|\Psi_\mathbf{z}^{(t)}\rangle,
\end{aligned}
\end{equation}
where $|\Psi_\mathbf{z}^{(t)}\rangle$ is a normalized $t$-qubit state. The rescaled $k$-th moment of the Kraus ensemble is
\begin{equation}
\begin{aligned}
    d^k M_\mathrm{DU-KIM}^{(k)} &= \frac{d^k}{d_B}\sum_\mathbf{z}|K_\mathbf{z})(K_\mathbf{z}|^{\otimes k} = \frac{2^{kt}}{d_B}\sum_\mathbf{z}\left(\langle +|^{\otimes(t - 2N_A)}\left(|\Psi_\mathbf{z}^{(t)}\rangle\langle\Psi_\mathbf{z}^{(t)}|\right) |+\rangle^{\otimes(t-2N_A)}\right)^{\otimes k}
    \\
    &= 2^{kt}\left(\langle + |^{\otimes (t-2N_A)}\right)^{\otimes k}\left( \frac{1}{d_B}\sum_\mathbf{z}|\Psi_\mathbf{z}^{(t)}\rangle\langle\Psi_\mathbf{z}^{(t)}|^{\otimes k} \right)\left(|+\rangle^{\otimes (t-2N_A)}\right)^{\otimes k}.
\end{aligned}
\end{equation}
By Lemma~\ref{lemma:ho-exact}, the collection $\{U_\mathbf{z}\}$ forms unitary designs in the TDL, thus the state ensemble $\{|\Psi_\mathbf{z}^{(t)}\rangle\}$ also forms state designs, i.e., $\lim_{N_B\rightarrow\infty}\frac{1}{d_B}\sum_\mathbf{z}|\Psi_\mathbf{z}^{(t)}\rangle\langle\Psi_\mathbf{z}^{(t)}|^{\otimes k} = \rho_{H, 2^t}^{(k)}$. We thus have
\begin{equation}
\begin{aligned}
    \lim_{N_B\rightarrow\infty}d^k M_\mathrm{DU-KIM}^{(k)} &= 2^{kt}\left(\langle + |^{\otimes (t-2N_A)}\right)^{\otimes k}\rho_{H,2^t}^{(k)}\left(|+\rangle^{\otimes (t-2N_A)}\right)^{\otimes k}
    \\
    &=2^{kt}\left(\langle + |^{\otimes (t-2N_A)}\right)^{\otimes k}\frac{1}{(2^t)_k}\underbrace{\sum_{\sigma\in S_k} W_\sigma}_{\text{on }k\text{ copies of }t\text{ qubits}}\left(|+\rangle^{\otimes (t-2N_A)}\right)^{\otimes k}
    \\
    &= \frac{2^{kt}}{(2^t)_k}\underbrace{\sum_{\sigma\in S_k} W_\sigma}_{\text{on }k\text{ copies of }2N_A\text{ qubits}}=\frac{2^{kt}}{(2^t)_k}M_\mathrm{Gin}^{(k)},
\end{aligned}
\end{equation}
where $W_\sigma$ represents the permutation action on $k$-copied space. In the second line, we have used Eq.~\eqref{eq:rho_Haar_k} to rewrite $\rho^{(k)}_{H, 2^t}$ as a sum of permutations. In the third line, we observe that for any permutation $W_\sigma$ the inner product of $k$ copies of $|+\rangle^{\otimes (t-2N_A)}$ always equals $1$, and the leftover action on the $k$ copies of the remaining $2N_A$ qubits is simply $W_\sigma$ restricted to those copies of $2N_A$ qubits. We then use Eq.~\eqref{eq:M_Gin} to convert this to the $k$-th moment of the Ginibre ensemble. Hence, we arrive at the desired statement. \quad$\square$

Lemma~\ref{lemma:du-kim} is a mathematical statement of the distribution of the limiting Kraus ensemble in the TDL. Colloquially, the Kraus ensemble in the DU KIM can be seen as an unnormalized state ensemble of $t$-qubit Haar random states projected to $2N_A$-qubit Haar random states, which also inherits unitary invariance. Therefore, following the same idea as in~\cite{ho_exact_2022}, one can show that for any pure input state on $A$, the DU KIM produces exact state designs on the output of $A$ for $t\ge 2N_A$ in the TDL. This generalizes the previous result in~\cite{ho_exact_2022} that exact state designs emerge for $t\ge N_A$ for a fixed input $|+\rangle^{\otimes N_A}$ on $A$.

With Lemma~\ref{lemma:du-kim} proved, proving Theorem~\ref{thm:2-precise} is then straightforward.

\textit{Proof of Theorem~\ref{thm:2-precise}.} By definition, the distance of the Kraus ensemble of DU KIM to Ginibre is
\begin{equation}
    \Delta_\mathrm{Gin}^{(k)} = \frac{\left\| d^k M_\mathrm{DU-KIM}^{(k)} - M_\mathrm{Gin}^{(k)}\right\|_2}{\left\| M_\mathrm{Gin}^{(k)}\right\|_2},
\end{equation}
and by $\lim_{N_B\rightarrow\infty}d^k M_\mathrm{DU-KIM}^{(k)}=\frac{2^{kt}}{(2^t)_k}M_\mathrm{Gin}^{(k)}$, we now have
\begin{equation}
    \lim_{N_B\rightarrow\infty} \Delta_\mathrm{Gin}^{(k)} =  \left| \frac{2^{kt}}{(2^t)_k} - 1\right| = 1 - \frac{2^{kt}}{(2^t)_k}.
\end{equation}
The rising factorial has asymptotic expansion
\begin{equation}
    (x)_k = x^k\left( 1 + \frac{k(k-1)}{2}\frac{1}{x} + \cdots \right),
\end{equation}
which allows us to conclude that
\begin{equation}
    \lim_{N_B\rightarrow\infty}\Delta_\mathrm{Gin}^{(k)} = \mathcal{O}\left(\frac{k^2}{2^t}\right),
\end{equation}
for large $t$. This concludes the proof.\quad$\blacksquare$

\section{Functional dependence on $\sigma^2\propto1/\gamma^t$}
\label{app:shrink}

In this section, we lay out the rationale of the inverse exponential scaling of $t$ appearing in $\sigma^2 \propto 1/\gamma^t$ in the log-normal part of our Ansatz. This originates from the non-unitary dynamics of the circuit bulk when viewed sideways, unlike that for a DU circuit. %
To study the $t$-dependence of the log-normal structure, we consider the average growth/shrink of a random state vector passing through a single-step transfer matrix.

\begin{figure}[!h]
    \centering
    \includegraphics[width=0.3\linewidth]{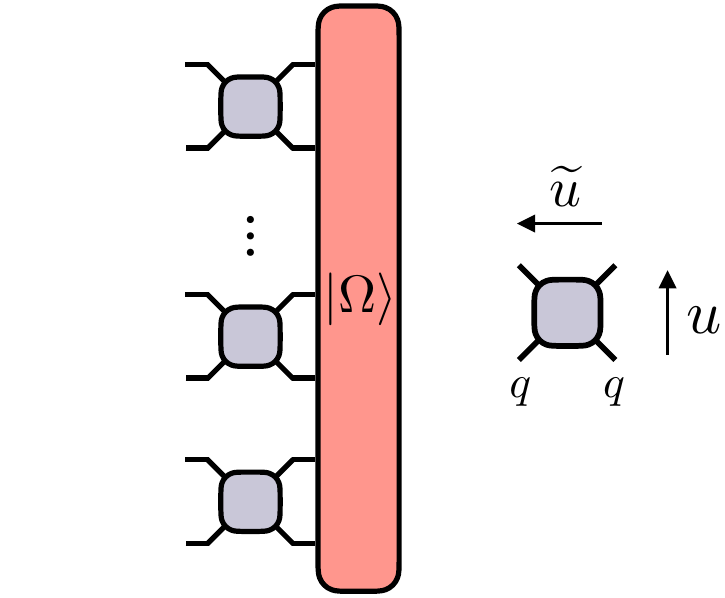}
    \caption{Toy case of a Haar random $t$-qudit state $|\Omega\rangle$ passing through a single-step transfer matrix made of disjoint dual gates $\widetilde u$. Here, a $2$-qudit gate is denoted $u$ in the time direction, and its spatial counterpart, the dual gate, is denoted $\widetilde u$ and depicted as above. After the action of dual gates, the state $|\widetilde\Omega\rangle$ is no longer normalized, and we show that its mean log norm follows a shrink $\mathrm{log}(\|\widetilde\Omega\|)\sim-{1}/{\gamma^t}$ inverse exponential in $t$. }
    \label{fig:single-step}
\end{figure}

We start with a toy example of qudits (local dimension $q$) where we assume: (i) the input state is a random $t$-qudit state drawn from the Haar measure, and (ii) the transfer matrix consists of an array of $2$-qudit Haar random gates applied sideways to pairs of qudits, illustrated in Fig.~\ref{fig:single-step}. For simplicity, let us assume $t$ is even, then there are $\frac{t}{2}$ gates in total. Let $|\Omega\rangle$ be the $t$-qudit Haar random state. The unnormalized state after the action is
\begin{equation}
    |\widetilde\Omega\rangle = \left(\bigotimes_{\tau\text{ odd}}\widetilde u_{\tau, \tau+1}\right)|\Omega\rangle,
\end{equation}
where $\widetilde{u}$ represents a dual gate, defined as in Fig.~\ref{fig:single-step}. Let us denote
\begin{equation}
    \langle\widetilde\Omega|\widetilde\Omega\rangle = v_0 (1+\nu),
\end{equation}
where we set $v_0 = \mathbb{E}[\langle\widetilde{\Omega}|\widetilde{\Omega}\rangle]$ as the mean of the squared norm, so $\mathbb{E}[\nu] = 0$ by definition. By Weingarten calculus,
\begin{equation}
    v_0 = \mathbb{E}\left[\langle\widetilde\Omega|\widetilde\Omega\rangle\right] = 1, \quad\mathbb{E}\left[\langle\widetilde\Omega|\widetilde\Omega\rangle^2\right] = \frac{q^t}{q^t+1} + \frac{\left(\frac{2q^2}{q^2+1}\right)^{t/2}}{q^t+1},
\end{equation}
which suggests that the average norm squared remains the same, but we will argue that when we take the average of the logarithm of the norm, there is a nonzero shrink per step. To see this, we note that
\begin{equation}
    \mathbb{E}[\nu^2] = \mathbb{E}\left[\langle\widetilde\Omega|\widetilde\Omega\rangle^2\right] - 1 = \frac{\left(\frac{2q^2}{q^2+1}\right)^{t/2}-1}{q^t + 1}.
\end{equation}
The averaged log norm follows from
\begin{equation}
\begin{aligned}
    \mathbb{E}\left[\log(\||\widetilde\Omega\rangle \|)\right] &= \frac{1}{2}\log(v_0) + \frac{1}{2}\mathbb{E}\left[\log(1+\nu)\right]
    \\
    &\approx -\frac{1}{4}\mathbb{E}[\nu^2],
\end{aligned}
\end{equation}
where we have used the expansion $\log(1+\nu)=\nu - \frac{\nu^2}{2}+\cdots$ around $\nu=0$ and neglected the higher order terms. For asymptotically large $t$, this suggests
\begin{equation}
    \mathbb{E}\left[\log(\||\widetilde\Omega\rangle \|)\right] \sim-\frac{1}{4}\frac{1}{\left(\sqrt{\frac{q^2+1}{2}}\right)^{t}},
\end{equation}
hence an average shrink inverse exponential in $t$.

For our setting of generic 1D circuits, the two assumptions in our toy example would fail: (i) generically the relevant stationary distribution of the $t$-qudit state is the Furstenberg measure, which is usually not the Haar measure, and (ii) the transfer matrix generically takes the form of a matrix product operator (MPO), which is not made up of disjoint non-unitary local gates. Nevertheless, we expect the qualitative features of our toy example to persist. Let $V$ be a single-step random transfer matrix, and $|\Omega_F\rangle$ be a normalized random vector sampled from the Furstenberg measure of $V$. Then we expect
\begin{equation}
\label{eq:shrink}
    \mathbb{E}\left[\log( \left\| V|\Omega_F\rangle\right\|)\right] = \frac{1}{2}\log(v_0) - \frac{\lambda_\mathrm{eff}}{\gamma_\mathrm{eff}^t},
\end{equation}
where $v_0 = \mathbb{E}\left[\langle \Omega_F|V^\dagger V|\Omega_F\rangle\right]$, and $\lambda_\mathrm{eff}$ and $\gamma_\mathrm{eff}$ are model-dependent constants.

\begin{figure}[!h]
    \centering
    \includegraphics[width=0.7\linewidth]{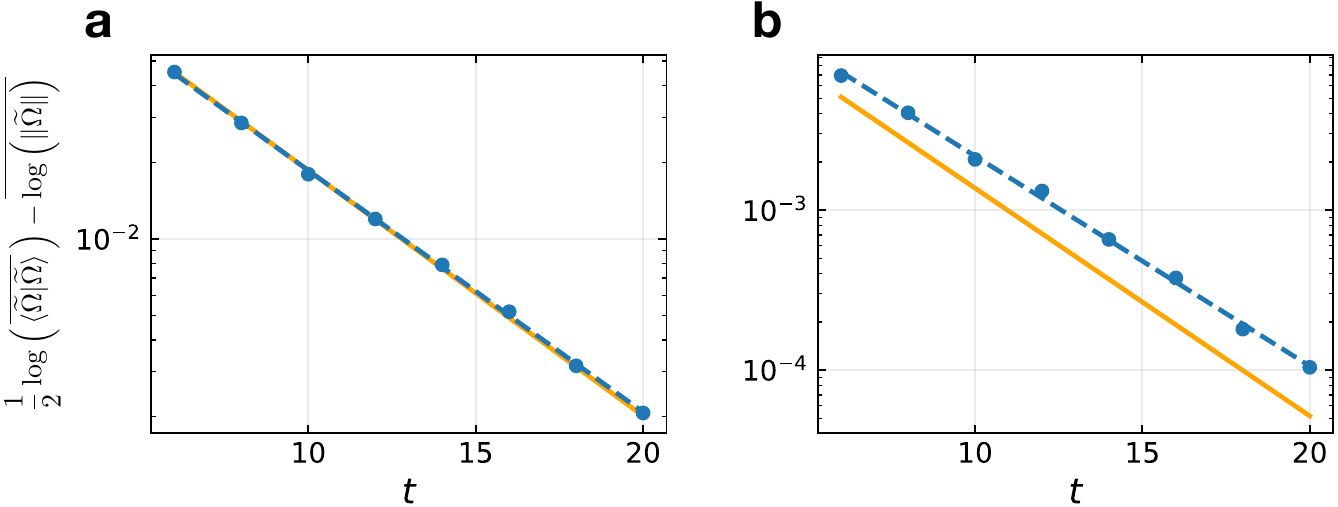}
    \caption{Numerics for exponentially decaying shrink rate per single-step transfer matrix. (a) RUC with transfer matrix defined as in Fig.~\ref{fig:ruc}, note that a single-step transfer matrix has width $\Delta N_B = 2$ qubits in the bulk. The blue dots are evaluated numerically from $1000$ samples and fitted with the blue dashed line $\lambda_\mathrm{eff}^{\mathrm{RUC}} / (\gamma_\mathrm{eff}^{\mathrm{RUC}})^t$, for which we obtain numerically $\lambda_\mathrm{eff}^\mathrm{RUC}\approx 0.166$ and $\gamma_\mathrm{eff}^\mathrm{RUC}\approx 1.245$. The orange line depicts $2\lambda^\mathrm{RUC} / (\gamma^{\mathrm{RUC}})^t$, with $\lambda^\mathrm{RUC}\approx0.0858$ and $\gamma^{\mathrm{RUC}}\approx 1.249$ extracted from the projected Kraus ensemble to fit the Ansatz (see Section~\ref{sec:numerics}). The factor of $2$ accounts for $\Delta N_B = 2$ in a single-step transfer matrix, and the orange line matches well with the blue dashed line. (b) Off-DU KIM with transfer matrix defined as in Fig.~\ref{fig:kim}. The blue dots are evaluated numerically from $1000$ samples and fitted with the blue dashed line $\lambda_\mathrm{eff}^{\mathrm{KIM}} / (\gamma_\mathrm{eff}^{\mathrm{KIM}})^t$, for which we obtain numerically $\lambda_\mathrm{eff}^{\mathrm{KIM}}\approx 0.0447$ and $\gamma_\mathrm{eff}^\mathrm{KIM}\approx 1.353$. The orange line depicts $\lambda^\mathrm{KIM} / (\gamma^\mathrm{KIM})^t$, with $\lambda^\mathrm{KIM}\approx 0.0364$ and $\gamma^\mathrm{KIM}\approx 1.388$ extracted from the projected Kraus ensemble to fit the Ansatz (see Appendix~\ref{app:extra-numerics}). The blue dashed line and the orange line match only qualitatively, and we attribute their difference to the action of the edge operator $W$ of the KIM. Hence, the constants $\lambda^\mathrm{KIM}$ and $\gamma^\mathrm{KIM}$ are dressed versions of the effective constants $\lambda^\mathrm{KIM}_\mathrm{eff}$ and $\gamma^\mathrm{KIM}_\mathrm{eff}$.}
    \label{fig:shrink-numerics}
\end{figure}

We demonstrate this scaling numerically for the brickwork RUC and off-DU KIM in Fig.~\ref{fig:shrink-numerics} panel (a) and (b), respectively. Specifically, we sample $\|\widetilde\Omega\| = \left\| V|\Omega_F\rangle\right\|$ with a Markov chain Monte Carlo (MCMC) procedure: for each $t$, we start with a Haar random $t$-qubit state and evolve under the random transfer matrices for a burn-in stage of $500$ steps. Next, we record the shrink of the vector in a single step per every $10$ steps of the random evolution, collecting a total of $1000$ such samples. We then collate the numerical average of:
\begin{equation}
    \frac{1}{2}\log\left(\,\overline{\langle\widetilde\Omega|\widetilde\Omega\rangle}\,\right) - \overline{\log\left( \|\widetilde\Omega\| \right)},
\end{equation}
over these $1000$ samples and plot against $t=6,8,10,12,14,16,20$. The plots in Fig.~\ref{fig:shrink-numerics} agree with the exponential decay in the average shrink ${\lambda_\mathrm{eff}}/{\gamma_\mathrm{eff}^t}$ as in Eq.~\eqref{eq:shrink}.

Indeed, this prediction is matched numerically with the physical models tested with our Ansatz. For the brickwork RUC and KIM off from DU point, the average shrink of a random Furstenberg state follows the decay pattern $-\frac{\lambda }{\gamma^t}$, with $\lambda$ and $\gamma$ derived numerically for the projected Kraus ensemble.

\section{More details on Ansatz}
\label{app:ansatz-more-details}
In this section, we discuss the details of the Ansatz regarding the fluctuations of the norm. First, we elaborate on the role of $\sigma^2$ when the Kraus ensemble is probed by the norm of the Kraus only. Next, we argue that the Ansatz is self-consistent even in the $\sigma^2\rightarrow\infty$ limit. We then address the limiting form of the projected ensemble implied by the Ansatz of Kraus operators.

\paragraph{Norm of Kraus} As discussed in the main text, for a random matrix following the Ansatz $\widetilde K\sim\mathbf{GinUE}\times\mathrm{LogNormal}_{\sigma^2}$, the logarithm of its norm $R\equiv \log(\| \widetilde K\|_2)$ follows
\begin{equation}
    R \overset{d}{=} \underbrace{\frac{1}{2}\log(X)}_{S} + \underbrace{\log(Y)}_{T},
\end{equation}
where $X\sim\Gamma(d_A^2,1)$ and $\log(Y)\sim\mathcal{N}(-\sigma^2,\sigma^2)$ are independent random scalars of an Erlang distribution and a normal distribution, respectively. Here, we have the statistics
\begin{equation}
    p_S(s) = \frac{2}{\Gamma(d_A^2)}\exp(2d_A^2 \, s - e^{2s}),
\end{equation}
and
\begin{equation}
    \mathbb{E}[S] = \frac{1}{2}\psi(d_A^2),\quad\mathrm{Var}[S] = \frac{1}{4}\psi_1(d_A^2),
\end{equation}
where $\psi$ and $\psi_1$ denote the digamma and trigamma functions, respectively. Likewise, 
\begin{equation}
    p_T(t) = \frac{1}{\sqrt{2\pi\sigma^2}}\exp(- \frac{(t+\sigma^2)^2}{2\sigma^2}),
\end{equation}
with
\begin{equation}
    \mathbb{E}[t] = -\sigma^2,\quad\mathrm{Var}[t] = \sigma^2.
\end{equation}
Therefore, by convolution, the PDF of $R$ satisfies
\begin{equation}
\label{eq:convolution}
\begin{aligned}
    p_R(r) &= (p_S * p_T)(r) \equiv \int_{-\infty}^{\infty}\mathrm{d}s\, p_S(s) p_T(r-s)
    \\
    &= \frac{2}{\Gamma(d_A^2)\sqrt{2\pi\sigma^2}}\int_{-\infty}^{\infty}\mathrm{d}s\,\exp(2d_A^2\,s - e^{2s} - \frac{(r-s+\sigma^2)^2}{2\sigma^2}),
\end{aligned}
\end{equation}
which is nothing but averaging a Gaussian kernel over the underlying $S$-distribution. And the statistics of $R$ follows
\begin{equation}
    \mathbb{E}[R] = \frac{1}{2}\psi(d_A^2) - \sigma^2,\quad \mathrm{Var}[R] = \frac{1}{4}\psi_1(d_A^2) + \sigma^2.
\end{equation}
Given the PDF of $R$ parametrized by $\sigma^2$, we find the best fit $\hat\sigma^2$ for a given empirical ensemble by maximum likelihood estimation (MLE), namely for an empirical dataset $\left\{ R_\mathbf{z}=\log(\sqrt{d}\| K_\mathbf{z} \|_2) \right\}$, we find $\hat\sigma^2_\mathrm{MLE}$ by
\begin{equation}
    \hat\sigma^2_\mathrm{MLE} = \arg\max_{\sigma^2\ge 0}\sum_\mathbf{z}\log(p_{\sigma^2}(R_\mathbf{z})),
\end{equation}
where $p_{\sigma^2}(r)$ is the same PDF $p_R(r)$ as in Eq.~\eqref{eq:convolution} with the $\sigma^2$ dependence explicit.

\paragraph{Consistency check in the $\sigma^2\rightarrow\infty$ limit} In the shallow circuit regime with circuit depth $t$ scaling slower than $t = \log_\gamma(N_B) + c$, where $c$ is an arbitrary real constant, the Ansatz predicts a diverging ensemble $\mathbf{GinUE}\times\mathrm{LogNormal}_{\sigma^2}$ with $\sigma^2\rightarrow\infty$. Here, we argue that the Ansatz is still reasonable while the moment-by-moment comparison fails for large $k$. 

By design of the Ansatz, the normalization condition is automatically satisfied by
\begin{equation}
    \sum_\mathbf{z}K_\mathbf{z}^\dagger K_\mathbf{z} \leftrightarrow d_B\mathbb{E}\left[K^\dagger K\right]=\mathbbm{1}_A,
\end{equation}
as discussed in Appendix~\ref{app:log-normal-introduction}. Here, $\sigma^2$ does not affect the expectation value of the normalization condition, but rather the fluctuations of such condition, namely
\begin{equation}
\begin{aligned}
    \left\|\mathbbm{1}_A - \sum_\mathbf{z}K_\mathbf{z}^\dagger K_\mathbf{z} \right\|_2^2 &= \mathrm{Tr}\left(\mathbbm{1}_A\right) - 2\,\mathrm{Tr}\left( \sum_\mathbf{z}K_\mathbf{z}^\dagger K_\mathbf{z} \right) + \mathrm{Tr}\left(\sum_{\mathbf{z},\mathbf{z}'} K_\mathbf{z}^\dagger K_\mathbf{z} K_{\mathbf{z}'}^\dagger K_{\mathbf{z}'} \right)
    \\
    &\leftrightarrow \mathrm{Tr}\left(\mathbbm{1}_A\right) - 2\,\mathrm{Tr}\left(d_B\mathbb{E}\left[K^\dagger K\right] \right) +  d_B(d_B-1)\mathrm{Tr}\left(\mathbb{E}\left[K^\dagger K\right]\mathbb{E}\left[K^\dagger K\right]\right) + d_B\mathrm{Tr}\left(\mathbb{E}\left[K^\dagger K K^\dagger K\right]\right),
    \\
    &= \frac{d_A}{d_B}\left[ 2\exp(4\sigma^2) - 1 \right],
\end{aligned}
\end{equation}
where we have used properties of Wick calculus and the log-normal distribution introduced in Appendix~\ref{app:review}. Specializing to qubits, if we substitute the parametric dependence $\sigma^2 = \frac{\lambda N_B}{\gamma^t}$ and $d_B = 2^{N_B}$, this fluctuation term can be written as
\begin{equation}
    \frac{d_A}{2^{N_B}}\left[2\exp(\frac{4\lambda N_B}{\gamma^t}) - 1\right] \sim d_A\exp(\left(\frac{4\lambda}{\gamma^t} - \log(2)\right)N_B),\quad\text{as }N_B\rightarrow\infty.
\end{equation}
To attain self-consistency, we require that this fluctuation goes to zero as $N_B\rightarrow\infty$, ensuring concentration to the normalization condition. This requires $4\lambda/\gamma^t < \log(2)$, i.e. a condition on the circuit depth
\begin{equation}
    t > \frac{\log(\frac{4\lambda}{\log(2)})}{\log(\gamma)}=: t_\mathrm{sc},
\end{equation}
which is a constant depth $t_\mathrm{sc}$. Hence, our Ansatz is self-consistent in terms of the normalization condition for diverging $\sigma^2 = \frac{\lambda N_B}{\gamma^t}\rightarrow\infty$ past a constant circuit depth. Numerically, for the brickwork RUC and off-DU KIM, both self-consistent depths are negative, suggesting self-consistency condition at any positive circuit depth.

\paragraph{Implications on the projected ensemble of states} As discussed in Section~\ref{sec:consequence-ansatz-deep-thermalization}, by the limiting assumption, the projected ensemble has limiting $k$-th moment
\begin{equation}
    \rho^{(k)} \approx d_B\underset{\widetilde\psi}{\mathbb{E}}\left[ \frac{|\widetilde\psi\rangle\langle\widetilde\psi|^{\otimes k}}{\langle\widetilde\psi|\widetilde\psi\rangle^{k-1}}\right],
\end{equation}
where $\sqrt{d}\widetilde\psi\sim\mathbf{Gauss}_\mathrm{d_A}\times\mathrm{LogNormal}_{\sigma^2}$. By left unitary invariance of the Gaussian distribution, we have
\begin{equation}
\begin{aligned}
    \underset{\widetilde\psi}{\mathbb{E}}\left[ \frac{|\widetilde\psi\rangle\langle\widetilde\psi|^{\otimes k}}{\langle\widetilde\psi|\widetilde\psi\rangle^{k-1}}\right] &= \underset{\widetilde\psi}{\mathbb{E}}\left[ \frac{U^{\otimes k}|\widetilde\psi\rangle\langle\widetilde\psi|^{\otimes k}U^{\dagger \,\otimes k}}{\langle\widetilde\psi|\widetilde\psi\rangle^{k-1}}\right]
    \\
    &= \underset{U\sim\mathrm{Haar}}{\mathbb{E}} \underset{\widetilde\psi}{\mathbb{E}}\left[ \frac{U^{\otimes k}|\widetilde\psi\rangle\langle\widetilde\psi|^{\otimes k}U^{\dagger \,\otimes k}}{\langle\widetilde\psi|\widetilde\psi\rangle^{k-1}}\right]
    \\
    &= \underset{\widetilde\psi}{\mathbb{E}}\left[ \frac{\underset{U\sim\mathrm{Haar}}{\mathbb{E}}\left[U^{\otimes k}|\widetilde\psi\rangle\langle\widetilde\psi|^{\otimes k}U^{\dagger \,\otimes k}\right]}{\langle\widetilde\psi|\widetilde\psi\rangle^{k-1}}\right] \propto \rho_H^{(k)},
\end{aligned}
\end{equation}
where in the first line, $U$ is an arbitrary unitary introduced at no cost because of the left unitary invariance of $\widetilde\psi$; in the second line, we perform an average of $U$ over the Haar measure, which does not affect the RHS; in the third line, we exchange the order of the expectations, and the integral over the Haar measure gives rise to a tensor proportional to $\rho_H^{(k)}$. Then we check the trace of the limiting $k$-th moment
\begin{equation}
    \mathrm{Tr}\left(d_B \underset{\widetilde\psi}{\mathbb{E}}\left[ \frac{|\widetilde\psi\rangle\langle\widetilde\psi|^{\otimes k}}{\langle\widetilde\psi|\widetilde\psi\rangle^{k-1}}\right] \right) = d_B \underset{\widetilde\psi}{\mathbb{E}}\left[ \langle\widetilde\psi|\widetilde\psi\rangle \right] = d_B \times \frac{d_A}{d} = 1,
\end{equation}
which does not depend on $\sigma^2$ in the Ansatz. Therefore,
\begin{equation}
    \rho^{(k)}\approx \rho_H^{(k)}.
\end{equation}

\section{Additional numerics}
\label{app:extra-numerics}

In this section, we present the numerical evidence for Conjecture~\ref{conj:1-precise} in the case of a kicked Ising model (KIM) off from the dual-unitary (DU) point. We also provide a probe of the Ginibre part of the Ansatz by studying the spectrum of fixed-trace Kraus ensemble, which corresponds to the fixed-trace Wishart--Laguerre ensemble.

\begin{figure}[!h]
    \centering
    \includegraphics[width=0.55\linewidth]{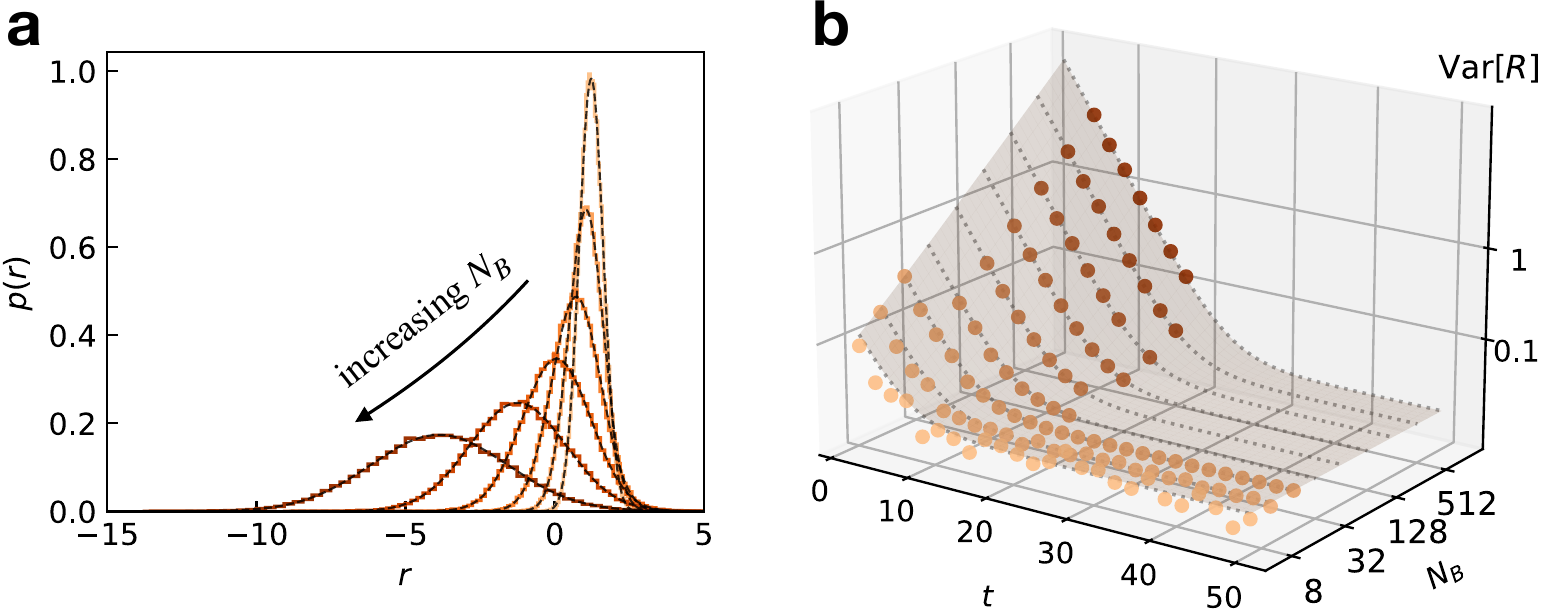}
    \caption{Fitting data shown for KIM. (a) Histogram of rescaled log norms $R_\mathbf{z}=\log(\sqrt{d}\|K_\mathbf{z}\|_2)$ for fixed $N_A=2$, $t = 6$, and $N_B = 32,64,128,256,512,1024$ (light-dark). %
    The dashed curves represent theoretical distributions $p_{\hat\sigma^2}(r)$ demonstrating excellent fits. (b) Variance of $\{R_\mathbf{z}\}$ in log scale, with $N_B = 6, 10, 18, 32, 64, 128, 256, 512, 1024$ (light-dark). The light brown surface outlines theoretical prediction $\mathrm{Var}[R]=\sigma^2_\mathrm{th}(N_B, t) + \frac{1}{4}\psi_1(2^{2N_A})$, where $\sigma^2_\mathrm{th} = \lambda N_B / \gamma^t$ is fitted from the best fits $\hat\sigma^2(N_B, t)$. Here for KIM, $\lambda\approx0.0364$ and $\gamma\approx 1.388$. The dashed curves represent slices of that surface with fixed $N_B$ in accordance with the plotted data.}
    \label{fig:sigma2_fit_KIM}
\end{figure}

First, we extract the best fits $\hat\sigma^2$ for the KIM by probing the norm of the Kraus operators, shown in Fig.~\ref{fig:sigma2_fit_KIM}. Here, we set $J = 1.000$, $h=0.4488$ away from DU point, and $g_i \sim \mathrm{Uniform}\left[
\frac{\pi}{7}-\frac{\pi}{12},
\frac{\pi}{7}+\frac{\pi}{12}
\right]$ independently for each site.

Next, we plot the running distance and the purity errors for the KIM, shown in Fig.~\ref{fig:dist-purity_KIM}. Numerics for KIM show the same qualitative behavior as that of RUC discussed in Section~\ref{sec:numerics}, providing additional evidence for Conjecture~\ref{conj:1-precise}.

\begin{figure}[!h]
    \centering
    \includegraphics[width=0.85\linewidth]{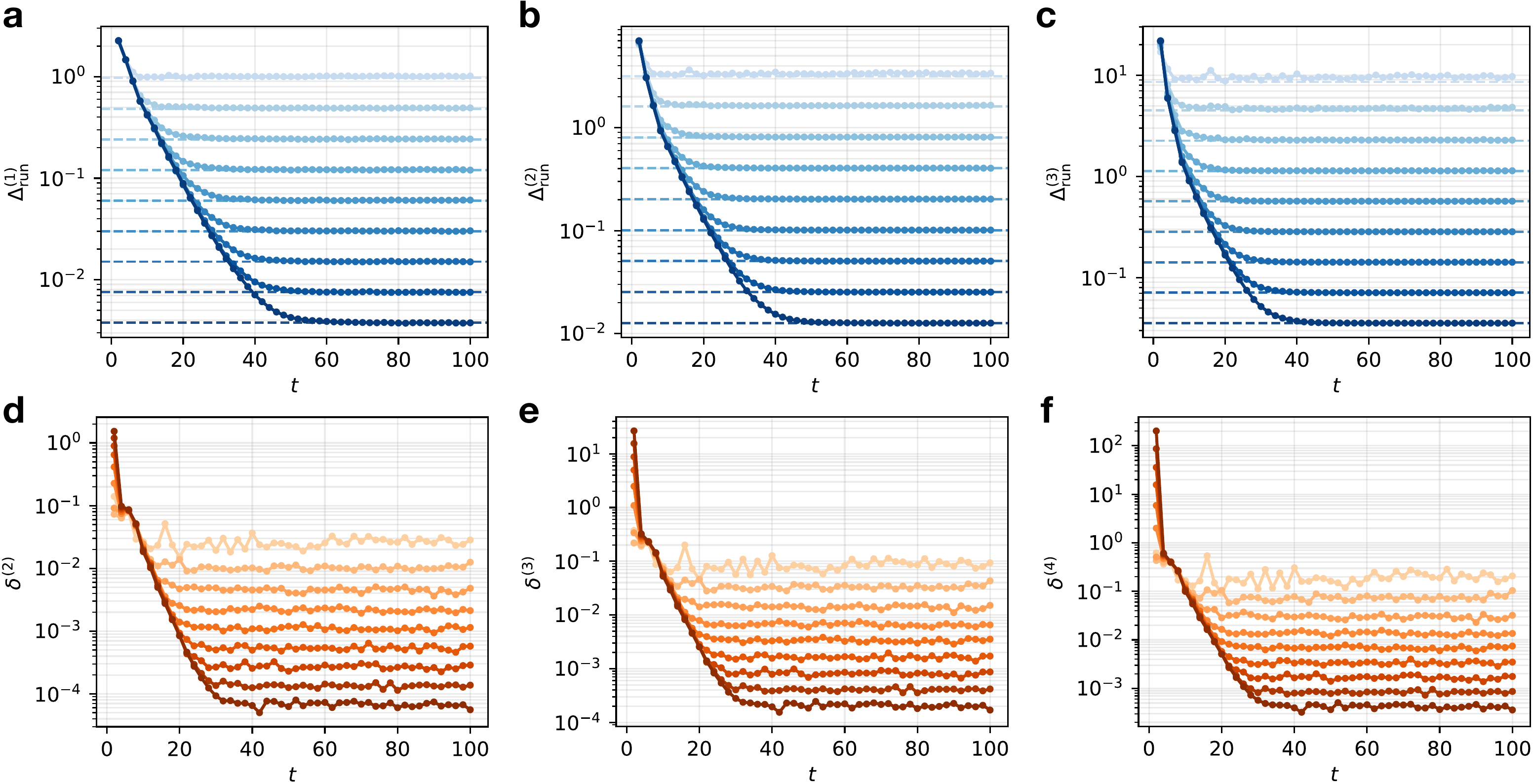}
    \caption{Numerical tests for KIM. Here we fix $N_A=2$ and show $N_B=4,6,8,10,12,14,16,18,20$ (light-dark). (a-c) Running distance $\Delta^{(k)}_\mathrm{run}$ shown for $k=1,2,3$. Each data point is averaged over 100 circuit realizations. These curves collapse in early times and follow an envelope showing a multi-rate exponential decay in $t$ before the curves plateau at a turning time $t_c$ scaling linearly in $N_B$. We attribute the multi-rate decay to early-time transient of local dynamics, and assume that the envelope crosses over to a simple exponential decay in $t$ at later times. The dashed lines represent the averaged distance $\Delta_\mathrm{Gin}^{(k)}$ for global Haar random unitary (also over $100$ realizations) with corresponding bath size $N_B$.
    (d-f) Purity relative error  $\delta^{(k)}$ compared to theoretical value, shown for $k=2,3,4$ (as $k=1$ is trivial). %
    Each data point is averaged over the same 100 circuit realizations. The envelope of the curves also follows exponential decay in $t$ at sufficiently large $t$.}
    \label{fig:dist-purity_KIM}
\end{figure}

Apart from the above numerics, we provide an additional test for the spectral properties of the matrix part of the Ansatz. To this end, we consider the norm-rescaled Kraus operators $K / \sqrt{\mathrm{Tr}(K^\dagger K)}$ and study the eigenvalue spectrum of $K^\dagger K / \mathrm{Tr}(K^\dagger K)$, the trace-normalized POVM element, %
which is agnostic of the norm of $K$. Correspondingly, we compare it with the fixed-trace Wishart--Laguerre ensemble, the relevant RMT baseline.

The generic Wishart--Laguerre ensemble consists of $N\times N$ matrices of the form~\cite{livan_moments_2011}
\begin{equation}
    W = G^\dagger G,
\end{equation}
where $G$ is a $M\times N$ Ginibre matrix (allowing real, complex, or quaternionic elements) with $M - N = \nu \ge 0$. Its eigenvalues satisfy the joint probability density function (JPDF)
\begin{equation}
    \mathcal{P}_\beta(x_1,\cdots,x_N)
    =
    C
    \prod_{j<k}|x_j-x_k|^\beta
    \prod_j
    x_j^{\frac{\beta}{2}(\nu+1)-1}
    e^{-\frac{1}{2}x_j},
    \qquad x_j\geq 0,
\end{equation}
where index $\beta = 1,2,4$ denotes the symmetry class of the ensemble, i.e., orthogonal, unitary, or symplectic. We are interested in square complex Ginibre matrices, so $\beta = 2$, $\nu = 0$ and the JPDF becomes
\begin{equation}
    \mathcal{P}(x_1,\cdots,x_N)
    =
    C
    \prod_{j<k}(x_j-x_k)^2
    \prod_j e^{-\frac{1}{2}x_j},
    \qquad x_j\geq 0 .
\end{equation}
Consider a change of variable with $x = \sum_j x_j$ and $\lambda_j = x_j/x$, the Vandermonde factor transforms as
\begin{equation}
\prod_{j<k}(x_j-x_k)^2
=
x^{N(N-1)}
\prod_{j<k}(\lambda_j-\lambda_k)^2 ,
\end{equation}
while the Jacobian contributes a factor $x^{N-1}$. Therefore,
\begin{equation}
\mathcal{P}(x,\lambda_1,\cdots,\lambda_N)
=
C
x^{N^2-1}
e^{-\frac{1}{2}x}
\prod_{j<k}(\lambda_j-\lambda_k)^2
\,\delta\!\left(1-\sum_{j=1}^N\lambda_j\right),
\end{equation}
with $x\geq 0$ and $\lambda_j\geq 0$. This allows us to study the spectrum of the fixed-trace Wishart--Laguerre ensemble, represented with $W / \mathrm{Tr}(W)$. The normalized eigenvalues $\{\lambda_j\}$ live on the simplex $\sum_j\lambda_j=1$, with the unordered JPDF~\cite{nechita_asymptotics_2007}
\begin{equation}
    \mathcal{P}(\vec\lambda) \propto \prod_{j<k} (\lambda_j - \lambda_k)^2,\quad \lambda_j\ge 0,~ \sum_j \lambda_j = 1.
\end{equation}
The one-point marginal is thus
\begin{equation}
\rho_N(\lambda)
:=
\int_{\lambda_2,\cdots,\lambda_N\geq 0,\, \sum_j\lambda_j=1}
\mathcal{P}(\lambda,\lambda_2,\cdots,\lambda_N)
\mathrm{d}\lambda_2\cdots\mathrm{d}\lambda_N,
\end{equation}
with the normalization $\int\rho_N(\lambda)\mathrm{d}\lambda=1$. Specializing to $N=4=d_A$ in our case ($N_A=2$ qubits), we have
\begin{equation}
\rho_4(\lambda)
=
15(1-\lambda)^8
\left(
73116\lambda^6
-94128\lambda^5
+44934\lambda^4
-9904\lambda^3
+1044\lambda^2
-48\lambda
+1
\right),
\quad 0\leq \lambda\leq 1 .
\end{equation}

Given our Ansatz for the Kraus operators $\sqrt{d}K\sim\mathbf{GinUE}\times\mathrm{LogNormal}_{\sigma^2}$, the spectrum of a trace-normalized POVM, i.e. $K^\dagger K / \mathrm{Tr}(K^\dagger K)$, should follow the fixed-trace Wishart--Laguerre prediction at sufficiently large $t$ and $N_B$. This is demonstrated numerically in Fig.~\ref{fig:wishart-laguerre}, providing additional evidence for the validity of the Ansatz.

\begin{figure}[!h]
    \centering
    \includegraphics[width=0.975\linewidth]{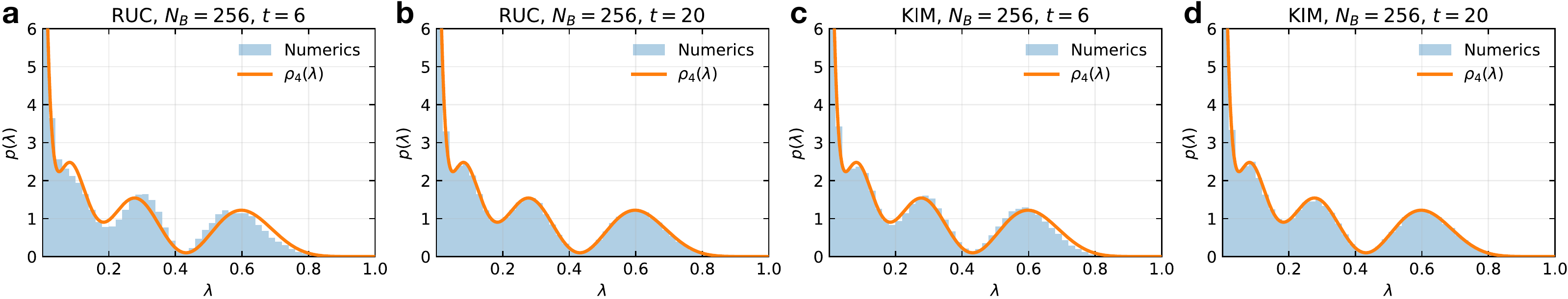}
    \caption{Histogram of eigenvalue density of $K_\mathbf{z}^\dagger K_\mathbf{z}/\mathrm{Tr}(K_\mathbf{z}^\dagger K_\mathbf{z})$ over $2^{16}$ samples from the Kraus ensemble $\{K_\mathbf{z}\}$, compared against the theoretical density for the fixed-trace Wishart--Laguerre $\rho_4(\lambda)$. Here, $N_A=2$ and $N_B=256$. (a, b) show numerics for RUC at $t=6$ and $t=20$, respectively. (c, d) show numerics for off-DU KIM at $t=6$ and $t=20$, respectively. For both models, the spectral density shows deviation from the theoretical curve at early time ($t=6$), while agreeing better with the theoretical curve at a later time ($t=20$).}
    \label{fig:wishart-laguerre}
\end{figure}

\section{Recoverability with side information}
\label{app:quantum-information-recovery}

Here, we provide the detailed analysis of local quantum information recoverability with: (a) classical side information, where we calculate the coherent information by utilizing the universal Ansatz for Kraus, and discuss the recovery protocol in detail; (b) quantum side information, corresponding to the Hayden--Preskill protocol~\cite{hayden_mirrors_2007}.

\paragraph{Classical side information} In Section~\ref{sec:consequence-ansatz-qi-recovery}, we have established that the coherent information with classical side information is
\begin{equation}
    I_\mathrm{coh}(R\rangle A C) = \sum_\mathbf{z} p(\mathbf{z})S(\rho_\mathbf{z}^R), %
\end{equation}
where we have defined the (unnormalized) conditional Choi state $|\widetilde{\Psi}_\mathbf{z}^{RA}\rangle = \frac{1}{\sqrt{d_R}}\sum_\alpha|\alpha\rangle_R\otimes K_\mathbf{z}|\alpha\rangle_A$ with probability $p(\mathbf{z})=\langle\widetilde{\Psi}_\mathbf{z}^{RA}|\widetilde{\Psi}_\mathbf{z}^{RA}\rangle$. This can be written explicitly as
\begin{equation}
\label{eq:recovery-z-prob}
    p(\mathbf{z}) = \frac{1}{d_R}\sum_{\alpha=1}^{d_R} \langle \alpha|K_\mathbf{z}^\dagger K_\mathbf{z}|\alpha\rangle_A = \mathrm{Tr}\left( \pi_R K_\mathbf{z}^\dagger K_\mathbf{z}\right),
\end{equation}
where $\pi_R \equiv \frac{1}{d_R}\sum_\alpha |\alpha\rangle\langle\alpha|_A$ is the maximally mixed state supported on the code space. For convenience, we denote $\mathbb{P}^{(R)}:= \sum_\alpha |\alpha\rangle\langle\alpha|_A$ as the projector onto the code space, thus $\pi_R = \mathbb{P}^{(R)} / d_R$. Then $\rho_\mathbf{z}^{R} = \mathrm{Tr}_A|\Psi_{\mathbf{z}}^{RA}\rangle\langle\Psi_{\mathbf{z}}^{RA}|$ is the reduced density matrix on $R$, where $|\Psi_\mathbf{z}^{RA}\rangle\equiv |\widetilde{\Psi}_\mathbf{z}^{RA}\rangle/\sqrt{p(\mathbf{z})}$ is the normalized conditional Choi state. The entropy $S(\rho_\mathbf{z}^R)$ is then the von Neumann entanglement entropy between $R$ and $A$ in $|\Psi_\mathbf{z}^{RA}\rangle$.

Substituting our Ansatz for $\{K_\mathbf{z}\}$, we have
\begin{equation}
    \sqrt{d}|\widetilde\Psi_\mathbf{z}^{RA}\rangle\sim \frac{1}{\sqrt{d_R}}\mathbf{Gauss}_{d_R d_A} \times\mathrm{LogNormal}_{\sigma^2},
\end{equation}
where $d=d_Ad_B$ and $\mathbf{Gauss}_{d_Rd_A}$ is a $(d_Rd_A)$-dimensional standard complex Gaussian vector. The angular part of this distribution is decoupled from the norm, and we obtain effectively
\begin{equation}
    I_\mathrm{coh}(R\rangle AC)=\sum_\mathbf{z}p(\mathbf{z}) S(\rho_\mathbf{z}^{R}) \approx \left\langle S_{d_R, d_A}\right\rangle,
\end{equation}
where $\left\langle S_{d_R, d_A}\right\rangle$ denotes the average entanglement entropy of a $d_R$-dimensional subsystem within a $(d_Rd_A)$-dimensional Haar random state, with $d_R\le d_A$. The closed form of this quantity was studied in~\cite{page_average_1993}, often referred to as the Page correction of entanglement entropy. Employing Page's result, we have
\begin{equation}
\begin{aligned}
    I_\mathrm{coh}(R\rangle AC) &= \left\langle S_{d_R, d_A}\right\rangle=\sum_{m=d_A+1}^{d_Rd_A}\frac{1}{m} - \frac{d_R - 1}{2d_A}
    \\
    &= H_{d_Rd_A} - H_{d_A} - \frac{d_R - 1}{2d_A},
\end{aligned}
\end{equation}
where $H_m \equiv \sum_{x=1}^m{1}/{x}$ denotes the harmonic number for $m\in\mathbb{N}$. It follows that the harmonic numbers satisfy
\begin{equation}
\label{eq:harmonic-number}
    H_m = \log(m) + \gamma + \frac{1}{2m} \underbrace{- \frac{1}{12m^2} + \frac{1}{120 m^4} - \cdots}_{\eta_m}\,,
\end{equation}
where $\gamma=0.57721\!\cdots$ is the Euler--Mascheroni constant, and the residue $\eta_m = -\frac{1}{12m^2}+\frac{1}{120 m^4} + \cdots < 0$ increases monotonically with $m$. We thus have
\begin{equation}
\begin{aligned}
    I_\mathrm{coh}(R\rangle AC) &= \log(d_R) + \frac{1}{2d_Rd_A} - \frac{d_R}{2d_A} + (\eta_{d_Rd_A} - \eta_{d_A})\ge\log(d_R) - \frac{d_R}{2d_A}.
\end{aligned}
\end{equation}
By assuming a constant qubit overhead $n = (N_A - N_R)$, the discrepancy between the coherent information and its maximal value follows
\begin{equation}
    \log(d_R) - I_\mathrm{coh}(R\rangle AC) \le \frac{1}{2}\frac{d_R}{d_A} = \frac{1}{2}\cdot\frac{1}{2^n}.
\end{equation}
Therefore, for arbitrary $\epsilon>0$, a qubit overhead of order $n=(N_A-N_R)\sim\log(\frac{1}{\epsilon})$ suffices to make the discrepancy between the coherent information and its maximal value smaller than $\epsilon$.

Next, let us consider the operational implication of this calculation. Schumacher and Westmoreland~\cite{Schumacher_approximate_2001} proved that if $\log(d_R) - I_\mathrm{coh}(R\rangle AC) \le \epsilon$, there exists a recovery action restricted to subsystems $AC$ of the Choi state $|\Psi^{RABC}\rangle$ that produces a state $\sigma^{RA}$ with fidelity $F(\sigma^{RA}, |\Phi^{RA}\rangle) \ge 1 - \sqrt{\epsilon}$, where $|\Phi^{RA}\rangle = \frac{1}{\sqrt{d_R}}\sum_\alpha |\alpha\rangle_R\otimes |\alpha\rangle_A$ is the maximally entangled state representing the encoding. In what follows, we reconstruct that recovery protocol and show that the Schumacher--Westmoreland recovery corresponds to a conditional unitary action $U_\mathbf{z}^\dagger$ on $A$ conditioned on the measurement outcome $\mathbf{z}$ on the classical register $C$.

Recall first that the actual Choi state is
\begin{equation}
    |\Psi^{RABC}\rangle = \frac{1}{\sqrt{d_R}}\sum_{\alpha,\mathbf{z}} |\alpha\rangle_R \otimes K_\mathbf{z}|\alpha\rangle_A \otimes |\mathbf{z}\rangle_C \otimes |\mathbf{z}\rangle_B,
\end{equation}
and the reduced density matrix on $RB$ is
\begin{equation}
\begin{aligned}
    \rho^{RB} &= \frac{1}{d_R}\sum_{\alpha,\beta,\mathbf{z}}\langle\beta|K_\mathbf{z}^\dagger K_\mathbf{z}|\alpha\rangle_A\, \left(|\alpha\rangle\langle\beta|_R\otimes |\mathbf{z}\rangle\langle\mathbf{z}|_B\right)
    \\
    &=\sum_\mathbf{z} \widetilde\rho_\mathbf{z}^R\otimes |\mathbf{z}\rangle\langle\mathbf{z}|_B,
\end{aligned}
\end{equation}
where $\widetilde \rho_\mathbf{z}^R = p(\mathbf{z})\rho_\mathbf{z}^R = \mathrm{Tr}_A\left(|\widetilde\Psi_\mathbf{z}^{RA}\rangle\langle\widetilde\Psi_\mathbf{z}^{RA}|\right)$ is the unnormalized reduced conditional density matrix. Its constituents $R$ and $B$ have reduced density matrices that satisfy:
\begin{equation}
    \rho^R = \mathrm{Tr}_B\left(\rho^{RB}\right) = \frac{1}{d_R}\sum_{\alpha,\beta} |\alpha\rangle\langle\beta|_R\left\langle \beta\left| \sum_\mathbf{z} K^\dagger_\mathbf{z}K_\mathbf{z}\right|\alpha\right\rangle_A = \frac{1}{d_R}\sum_\alpha|\alpha\rangle\langle\alpha|_R = \frac{\mathbbm{1}_R}{d_R},
\end{equation}
the maximally mixed state on $R$, where we have used the normalization condition $\sum_\mathbf{z}K_\mathbf{z}^\dagger K_\mathbf{z} = \mathbbm{1}_A$, and
\begin{equation}
    \rho^B = \mathrm{Tr}_R\left(\rho^{RB}\right) = \frac{1}{d_R}\sum_{\alpha,\mathbf{z}}\langle\alpha|K_\mathbf{z}^\dagger K_\mathbf{z}|\alpha\rangle_A |\mathbf{z}\rangle\langle \mathbf{z}|_B = \sum_\mathbf{z} p(\mathbf{z})|\mathbf{z}\rangle\langle\mathbf{z}|_B,
\end{equation}
where $p(\mathbf{z})$ is defined as Eq.~\eqref{eq:recovery-z-prob}. By assuming that the coherent information is near maximal, we have the following for the relative entropy
\begin{equation}
    S(\rho^{RB}\| \rho^R\otimes \rho^B) = S(R) + S(B) - S(RB) = S(R) - I_\mathrm{coh}(R\rangle AC)\le \epsilon,
\end{equation}
where we have used the second line of Eq.~\eqref{eq:coherent-info}, and assumed that $I_\mathrm{coh}(R\rangle AC) \ge \log(d_R) - \epsilon$, for some $\epsilon > 0$. This further implies that
\begin{equation}
    F(\rho^{RB}, \rho^R \otimes \rho^B) \ge 1 - \sqrt{\epsilon},
\end{equation}
where $F$ is the fidelity defined as $F(\rho,\sigma)\equiv\mathrm{Tr}\sqrt{\sqrt{\rho}\sigma\sqrt{\rho}} = \left\|\sqrt{
\rho}\sqrt{\sigma} \right\|_1$. Uhlmann's theorem~\cite{Uhlmann} claims that the fidelity $F(\rho,\sigma) = \max_{|\psi_\rho\rangle, |\psi_\sigma\rangle} \left| \langle\psi_\rho|\psi_\sigma\rangle\right|$, where the maximization is taken over all possible purifications $|\psi_\rho\rangle$ and $|\psi_\sigma\rangle$ of $\rho$ and $\sigma$, respectively. Furthermore, the maximum is attained for some purifications $|\psi_\rho\rangle$ and $|\psi_\sigma\rangle$ in the finite-dimensional Hilbert spaces of $\rho$ and $\sigma$. In fact, for our purpose, we will construct a purification of $\rho^R\otimes \rho^B$ that has an overlap with the actual Choi state $|\Psi^{RABC}\rangle$ that exactly equals this fidelity $F(\rho^{RB}, \rho^R \otimes \rho^B)$.

Consider the auxiliary Choi state
\begin{equation}
    |\hat\Psi^{RABC}\rangle = \frac{1}{\sqrt{d_R}}\sum_{\alpha,\mathbf{z}}|\alpha\rangle_R\otimes \widetilde U_\mathbf{z}|\alpha\rangle_A\otimes |\mathbf{z}\rangle_C\otimes|\mathbf{z}\rangle_B,
\end{equation}
where we simply replaced $K_\mathbf{z}$ in $|\Psi^{RABC}\rangle$ with $\widetilde U_\mathbf{z}$, defined through the polar decomposition
\begin{equation}
    K_\mathbf{z}\mathbb{P}^{(R)} = U_\mathbf{z}\sqrt{\left(K_\mathbf{z}\mathbb{P}^{(R)}\right)^\dagger\left(K_\mathbf{z}\mathbb{P}^{(R)}\right)} = U_\mathbf{z}\sqrt{\mathbb{P}^{(R)}K_\mathbf{z}^\dagger K_\mathbf{z}\mathbb{P}^{(R)}},\quad \text{and }\widetilde U_\mathbf{z}:= \sqrt{p(\mathbf{z})}\, U_\mathbf{z}.
\end{equation}
One can then check immediately that
\begin{equation}
    \hat\rho^R = \frac{1}{d_R}\sum_{\alpha,\beta,\mathbf{z}} \langle \beta|\widetilde U_\mathbf{z}^\dagger \widetilde U_\mathbf{z}|\alpha\rangle_A |\alpha\rangle\langle\beta|_R = \frac{1}{d_R}\sum_{\alpha,\mathbf{z}} p(\mathbf{z})|\alpha\rangle\langle\alpha|_R = \frac{\mathbbm{1}_R}{d_R} = \rho^R,
\end{equation}
and
\begin{equation}
    \hat\rho^{B} = \frac{1}{d_R}\sum_{\alpha,\mathbf{z}} \langle\alpha|\widetilde U_\mathbf{z}^\dagger \widetilde U_\mathbf{z}|\alpha\rangle_A |\mathbf{z}\rangle\langle\mathbf{z}|_B = \sum_\mathbf{z}p(\mathbf{z})|\mathbf{z}\rangle\langle\mathbf{z}|_B = \rho^B,
\end{equation}
and furthermore
\begin{equation}
\begin{aligned}
    \hat\rho^{RB} &= \frac{1}{d_R} \sum_{\alpha,\beta,\mathbf{z}} \langle\beta|\widetilde U_\mathbf{z}^\dagger \widetilde U_\mathbf{z}|\alpha\rangle_A \left(|\alpha\rangle\langle\beta|_R\otimes |\mathbf{z}\rangle\langle\mathbf{z}|_B\right)
    = \frac{1}{d_R}\sum_{\alpha,\beta}\sum_\mathbf{z}p(\mathbf{z})\langle\beta|\alpha\rangle_A\left(|\alpha\rangle\langle\beta|_R\otimes |\mathbf{z}\rangle\langle\mathbf{z}|_B\right)
    \\
    &= \left(\frac{1}{d_R}\sum_{\alpha}|\alpha\rangle\langle\alpha|_R\right)\otimes \left(\sum_\mathbf{z} p(\mathbf{z})|\mathbf{z}\rangle\langle\mathbf{z}|_B\right) = \hat\rho^R\otimes\hat\rho^B = \rho^R\otimes\rho^B,
\end{aligned}
\end{equation}
and one may thus claim that $|\hat\Psi^{RABC}\rangle$ is a purification of $\rho^R\otimes\rho^B$. Now note that
\begin{equation}
    F(\rho^{RB},\rho^R\otimes \rho^B) = \mathrm{Tr}\sqrt{\sqrt{\rho^{RB}}\rho^R\otimes \rho^B\sqrt{\rho^{RB}}},
\end{equation}
where we have
\begin{equation}
\begin{aligned}
    \sqrt{\sqrt{\rho^{RB}}\rho^R\otimes \rho^B\sqrt{\rho^{RB}}} &= \sqrt{ \sum_\mathbf{z} \frac{p(\mathbf{z})}{d_R}{\widetilde\rho_\mathbf{z}^R}\otimes |\mathbf{z}\rangle\langle\mathbf{z}|_B} = \sum_\mathbf{z} \sqrt{\frac{p(\mathbf{z})}{d_R}}\sqrt{{\widetilde\rho_\mathbf{z}^R}}\otimes |\mathbf{z}\rangle\langle\mathbf{z}|_B,
\end{aligned}
\end{equation}
and thus
\begin{equation}
\label{eq:fidelity-term}
\begin{aligned}
    F(\rho^{RB},\rho^R\otimes \rho^B) &= \mathrm{Tr}\sqrt{\sqrt{\rho^{RB}}\rho^R\otimes \rho^B\sqrt{\rho^{RB}}}= \sum_\mathbf{z}\sqrt{\frac{p(\mathbf{z})}{d_R}}\mathrm{Tr}\sqrt{\widetilde\rho_\mathbf{z}^{R}},
\end{aligned}
\end{equation}
while the overlap
\begin{equation}
\label{eq:overlap-term}
\begin{aligned}
    \langle\hat\Psi^{RABC}|\Psi^{RABC}\rangle &= \frac{1}{d_R}\sum_{\alpha,\mathbf{z}} \langle \alpha | \widetilde U_\mathbf{z}^\dagger K_\mathbf{z}|\alpha\rangle_A = \sum_\mathbf{z}\mathrm{Tr}\left( \pi_R \widetilde U_\mathbf{z}^\dagger \underbrace{K_\mathbf{z}\mathbb{P}^{(R)}}_{U_\mathbf{z}\sqrt{\mathbb{P}^{(R)} K_\mathbf{z}^\dagger K_\mathbf{z}\mathbb{P}^{(R)}}    
    }\right)
    \\
    &= \sum_{\mathbf{z}}\mathrm{Tr}\left( \pi_R \sqrt{p(\mathbf{z})}  \sqrt{\mathbb{P}^{(R)} K^\dagger_\mathbf{z} K_\mathbf{z}\mathbb{P}^{(R)}}\right)
    = \sum_\mathbf{z}\sqrt{\frac{p(\mathbf{z})}{d_R}}\mathrm{Tr}\left( \frac{1}{\sqrt{d_R}} \sqrt{\mathbb{P}^{(R)} K^\dagger_\mathbf{z} K_\mathbf{z}\mathbb{P}^{(R)}} \right),
\end{aligned}
\end{equation}
now note that
\begin{equation}
\begin{aligned}
\label{eq:elementwise-id-transpose}
    \frac{1}{d_R}\langle \alpha|\mathbb{P}^{(R)} K_\mathbf{z}^\dagger K_\mathbf{z}\mathbb{P}^{(R)}|\beta\rangle_A & = \frac{1}{d_R}\langle \alpha|K_\mathbf{z}^\dagger K_\mathbf{z}|\beta\rangle_A
    \\
     &= \langle\beta|\widetilde\rho_\mathbf{z}^R|\alpha\rangle_R,
\end{aligned}
\end{equation}
where in the second line we used $\langle\alpha|\widetilde\rho_\mathbf{z}^{R}|\beta\rangle_R = \left\langle\alpha\left| \frac{1}{d_R}\sum_{\alpha',\beta'}\langle\beta'|K_\mathbf{z}^\dagger K_\mathbf{z}|\alpha'\rangle_A|\alpha'\rangle\langle\beta'|_R \right|\beta\right\rangle_R = \frac{1}{d_R}\langle\beta|K_\mathbf{z}^\dagger K_\mathbf{z}|\alpha\rangle_A$. Eq.~\eqref{eq:elementwise-id-transpose} suggests that the matrix representation of $\widetilde\rho_\mathbf{z}^R$ with respect to the basis $\{|\alpha\rangle_R\}$ and the matrix representation of $\frac{1}{d_R}\mathbb{P}^{(R)} K_\mathbf{z}^\dagger K_\mathbf{z}\mathbb{P}^{(R)}$ restricted to the subspace spanned by basis $\{|\alpha\rangle_A\}$ is related by a transposition. Therefore, combined with Eqs.~\eqref{eq:fidelity-term} and~\eqref{eq:overlap-term}, we have
\begin{equation}
    F(\rho^{RB},\rho^R\otimes \rho^B) = \langle \hat\Psi^{RABC}|\Psi^{RABC}\rangle,
\end{equation}
namely that $|\Psi^{RABC}\rangle$ and $|\hat\Psi^{RABC}\rangle$ are purifications of $\rho^{RB}$ and $\rho^R\otimes \rho^B$, respectively, that saturate the Uhlmann's theorem exactly.

The next step follows from the Schumacher--Westmoreland argument~\cite{Schumacher_approximate_2001} that if we focus on the auxiliary
\begin{equation}
    |\hat\Psi^{RABC}\rangle = \frac{1}{\sqrt{d_R}}\sum_{\alpha,\mathbf{z}} \sqrt{p(\mathbf{z})}|\alpha\rangle_R\otimes U_\mathbf{z}|\alpha\rangle_A\otimes|\mathbf{z}\rangle_C\otimes |\mathbf{z}\rangle_B,
\end{equation}
we could restore the exact initial state $|\Phi^{RA}\rangle$ by first measuring the classical register $C$
\begin{equation}
    \Pi_\mathbf{z} = \mathbbm{1}_R\otimes \mathbbm{1}_A\otimes |\mathbf{z}\rangle\langle\mathbf{z}|_C,
\end{equation}
and upon obtaining an outcome $\mathbf{z}$, the (unnormalized) conditional state on $RA$ is
\begin{equation}
    |\widetilde\Xi_\mathbf{z}\rangle = \sqrt{\frac{p(\mathbf{z})}{d_R}} \sum_\alpha |\alpha\rangle_R\otimes U_\mathbf{z}|\alpha\rangle_A,
\end{equation}
and we follow up with the conditional unitary action $U_\mathbf{z}^\dagger$ on $A$. We therefore have
\begin{equation}
    (\mathbbm{1}_R\otimes U_\mathbf{z}^\dagger)|\widetilde\Xi_\mathbf{z}\rangle = \sqrt{\frac{p(\mathbf{z})}{d_R}} \sum_\alpha |\alpha\rangle_R \otimes U_\mathbf{z}^\dagger U_\mathbf{z}|\alpha\rangle_A = \sqrt{p(\mathbf{z})} |\Phi^{RA}\rangle,
\end{equation}
which guarantees perfect fidelity with $|\Phi^{RA}\rangle$ for every possible outcome $\mathbf{z}$. We denote this action of first measuring $\mathbf{z}$ on $C$ followed by a conditional unitary action $U_\mathbf{z}^\dagger$ on $A$ by $\mathcal{R}^{AC\rightarrow A}$. Now we apply exactly the same action $\mathcal{R}^{AC\rightarrow A}$ to the actual Choi state $|\Psi^{RABC}\rangle$, where $B$ is traced out, thus obtaining a state $\sigma^{RA}$. By the quantum information processing inequality, fidelity is non-decreasing upon any quantum operations, therefore
\begin{equation}
\begin{aligned}
    F(\sigma^{RA}, |\Phi^{RA}\rangle) &\ge F(|\Psi^{RABC}\rangle, |\hat\Psi^{RABC}\rangle)
    \\
    &= \langle \hat\Psi^{RABC}|\Psi^{RABC}\rangle = F(\rho^{RB},\rho^R\otimes \rho^B)
    \\
    &\ge 1 - \sqrt{\epsilon}.
\end{aligned}
\end{equation}
Thus, we may conclude that there exists a recovery protocol by applying a conditional unitary $U_\mathbf{z}^\dagger$ upon measuring $\mathbf{z}$ on the bath that produces a recovery $\sigma^{RA}$ with at least $(1-\sqrt{\epsilon})$ fidelity with $|\Phi^{RA}\rangle$.

\paragraph{Quantum side information} Similarly, for the Hayden--Preskill setting shown in Fig.~\ref{fig:recovery}(b), one could also compute the coherent information
\begin{equation}
    I_\mathrm{coh}(R\rangle AQ) = S(AQ)- S(RAQ) = S(RB) - S(B),
\end{equation}
here, we operate under a slight abuse of notation by assigning $R$ and $Q$ to be associated with the input systems and letting $A$ and $B$ be the newly emitted radiation and the remaining black hole, respectively. Thus, $d = d_Rd_Q = d_Ad_B$. 
We make use of the inequalities: $S(RB)\ge -\log(\mathrm{Tr}\left(\rho_{RB}^{2}\right))$ (monotonicity of Rényi entropies), and $S(B)\le \log(d_B)$, and we arrive at the coherent information averaged over Haar random $U_{AB}$
\begin{equation}
\begin{aligned}
    \mathbb{E}\left[I_\mathrm{coh}(R\rangle AQ)\right] &\ge \mathbb{E}\left[-\log(\mathrm{Tr}\left(\rho_{RB}^2\right))\right] - \log(d_B)
    \\
    &\ge -\log(\mathbb{E}\left[ \mathrm{Tr}\left(\rho_{RB}^2\right) \right]) - \log(d_B)
    \\
    &= -\log(\frac{d^2 + d_R^2d_B^2 - d_R^2 - d_B^2}{d_Rd_B(d^2-1)}) - \log(d_B)
    \\
    &= \log(d_R) - \log(1 + \frac{(d_R^2-1)(d_B^2-1)}{d^2-1}).
\end{aligned}
\end{equation}
In the second line, we make use of the convexity of $-\log(\cdot)$. In the third line, we use Weingarten calculus to evaluate the purity. Furthermore, for finite $d$ with $d_A, d_B, d_R, d_Q > 1$, we have
\begin{equation}
    \frac{(d_R^2 - 1)(d_B^2-1)}{d_A^2 d_B^2-1} = \frac{(d_R^2 - 1)(d_B^2-1)}{d_A^2( d_B^2-1) + (d_A^2-1)} \le \frac{d_R^2-1}{d_A^2} \le \frac{d_R^2}{d_A^2},
\end{equation}
recall that $n = (N_A - N_R)$ is the qubit overhead and we may thus substitute and conclude
\begin{equation}
    \mathbb{E}\left[I_\mathrm{coh}(R\rangle AQ)\right] \ge \log(d_R) - \log(1+\frac{d_R^2}{d_A^2})\ge\log(d_R) - \frac{d_R^2}{d_A^2} = \log(d_R) - 2^{-2n}.
\end{equation}
Therefore, one could similarly conclude that for arbitrary $\epsilon > 0$, a constant qubit overhead $n = (N_A - N_R)\sim\log(\frac{1}{\epsilon})$ suffices to make the typical discrepancy between the coherent information and its maximal value $\epsilon$ small. %

\end{document}